\documentclass[11pt]{article}

\usepackage{amssymb,amsmath,latexsym,amscd,amsfonts}
\usepackage{enumerate}
\usepackage{graphicx}
\usepackage[export]{adjustbox}
\usepackage{color}
\usepackage{fullpage}
\usepackage{hyperref}
\usepackage{amsthm,bbm,mathtools}
\usepackage[utf8x]{inputenc}
\usepackage{enumitem}
\usepackage[font={small}]{caption}
\usepackage{algorithmicx}
\usepackage{algpseudocode}
\usepackage{amscd}
\usepackage[dvipsnames]{xcolor}

\usepackage{tikz}
\usetikzlibrary{shapes,arrows,calc,positioning,fit}
\tikzstyle{vertex} = [fill,shape=circle,node distance=80pt]
\tikzstyle{edge} = [fill,opacity=.5,fill opacity=.5,line cap=round, line join=round, line width=50pt]
\tikzstyle{elabel} =  [fill,shape=circle,node distance=30pt]

\def\ignore #1 {}

\definecolor{darkorange}{RGB}{210,50,30}

\newcommand{\norm}[1]{\left\|\,#1\,\right\|}       

\newcommand{\set}[1]{{\left\{#1\right\}}}    

\newcommand{\ve}[1]{\mathbf{#1}}
\newcommand{\abs}[1]{\left\lvert #1 \right\rvert}

\newcommand{\complex}{{\mathbb C}}

\newcommand{\nats}{{\mathbb N}}

\def\ket#1{ | #1 \rangle}
\def\bra#1{{\langle #1 | }}
\newcommand{\ketbra}[2]{\ket{#1}\!\bra{#2}}        

\newcommand{\braket}[2]{\mbox{$\langle #1  | #2 \rangle$}}

\newcommand\tqsat{3-QSAT}
\newcommand{\class}[1]{\textup{#1}}
\newcommand{\spa}[1]{\mathcal{#1}}

\newcommand{\lin}{\mathcal{L}}
\newcommand{\qi}{{\color{darkorange}(QI)}}

\makeatletter
\newtheorem*{rep@theorem}{\rep@title}
\newcommand{\newreptheorem}[2]{%
\newenvironment{rep#1}[1]{%
 \def\rep@title{#2 \ref{##1}}%
 \begin{rep@theorem}}%
 {\end{rep@theorem}}}
\makeatother

\newreptheorem{theorem}{Theorem}

\newtheorem{theorem}{Theorem}

\newtheorem{lemma}[theorem]{Lemma}
\newtheorem{corollary}[theorem]{Corollary}

\theoremstyle{definition}
\newtheorem{definition}[theorem]{Definition}
\newtheorem{rem}[theorem]{Remark}
\newtheorem{example}[theorem]{Example}

\def\X{{\mathcal X}}
\def\K{{\mathbb K}}

\def\C{{\mathcal C}}

\def\S{{\mathcal S}}

\title{On efficiently solvable cases of Quantum $k$-SAT\footnote{A short version of this article has appeared as a long abstract in Proceedings of the 43rd International Symposium on Mathematical Foundations of Computer Science (MFCS), volume 117 of Leibniz International Proceedings in Informatics (LIPIcs), pages 38:1--38:16, 2018.}}

\author{Marco Aldi\footnote{Department of Mathematics and Applied Mathematics, Virginia Commonwealth University, Richmond, VA, U.S.A., maldi2@vcu.edu.} \and Niel de Beaudrap\footnote{Department of Computer Science, University of Oxford, Oxford, U.K., niel.debeaudrap@cs.ox.ac.uk.} \and Sevag Gharibian\footnote{Department of Computer Science, University of Paderborn, Paderborn North-Rhine-Westphalia, Germany, and Virginia Commonwealth University, Richmond, VA, U.S.A., sevag.gharibian@upb.de, +49-0176-5888-7082.} \and Seyran Saeedi\footnote{Department of Computer Science, Virginia Commonwealth University, Richmond, VA, U.S.A., saeedis@mymail.vcu.edu.}}
\date{}

\begin{document}

\maketitle

\begin{abstract}
Estimating ground state energies of local Hamiltonian models is a central problem in quantum physics. The question of whether a given local Hamiltonian is \emph{frustration-free}, meaning the ground state is the simultaneous ground state of \emph{all} local interaction terms, is known as the Quantum $k$-SAT ($k$-QSAT) problem. In analogy to its classical Boolean constraint satisfaction counterpart, the NP-complete problem $k$-SAT, Quantum $k$-SAT is QMA$_1$-complete (for $k\geq 3$, and where QMA$_1$ is a quantum generalization of NP with one-sided error), and thus likely intractable. But whereas $k$-SAT has been well-studied for special tractable cases, as well as from a ``parameterized complexity'' perspective, much less is known in similar settings for $k$-QSAT. Here, we study the open problem of computing satisfying assignments to $k$-QSAT instances which have a  ``dimer covering'' or ``matching''; such systems are known to be frustration-free, but it remains open whether one can \emph{efficiently compute} a ground state.

Our results fall into three directions, all of which relate to the ``dimer covering'' setting: (1) We give a polynomial-time classical algorithm for $k$-QSAT when all qubits occur in at most two interaction terms or clauses. (2) We give a ``parameterized algorithm'' for $k$-QSAT instances from a certain non-trivial class, which allows us to obtain exponential speedups over brute force methods in some cases. This is achieved by reducing the problem to solving for a single root of a single univariate polynomial. An explicit family of hypergraphs, denoted Crash, for which such a speedup is obtained is introduced. (3) We conduct a structural graph theoretic study of $3$-QSAT interaction graphs which have a ``dimer covering''. We remark that the results of (2), in particular, introduce a number of new tools to the study of Quantum SAT, including graph theoretic concepts such as transfer filtrations and blow-ups from algebraic geometry.
 \end{abstract}

\section{Introduction}\label{scn:intro}
Estimating ground state energies is a central problem in quantum physics. Specifically, given a $k$-local Hamiltonian $H=\sum_i H_i$ acting on $n$ qubits, and where each local interaction term or \emph{clause} $H_i$ acts non-trivially on a subset of at most $k$ qubits\footnote{For clarity, in the general formulation of $k$-LH, there is no geometric restriction on where the interaction terms act (e.g. on a 1D chain or a 2D lattice).}, the aim is to estimate the smallest eigenvalue of $H$. In 1999, Kitaev showed that this problem, dubbed the $k$-local Hamiltonian problem ($k$-LH), is complete for a quantum generalization of NP known as Quantum Merlin Arthur (QMA)~\cite{KSV02}. This not only demonstrated that $k$-LH is likely intractable in the worst case, but also spawned an entire area of research at the intersection of mathematics, condensed matter physics, and computational complexity theory known as \emph{Quantum Hamiltonian Complexity} (see e.g.~\cite{AN02,O11,GHLS14} for surveys).

Let us explore this connection to classical complexity theory further. In the canonical NP-complete Boolean constraint satisfaction problem MAX-$k$-SAT, one is given a set of $k$-local Boolean functions or \emph{constraints/clauses} $\phi:=\set{\phi_i}_{i=1}^m$, such that each clause acts on $k$ out of all $n$ bits in the system. (Moreover, for MAX-$k$-SAT, each $\phi_i$ is a logical OR of $k$ literals, such as $\phi_i =x_2\vee\overline{x_3} \vee\cdots\vee x_1$, for $\overline{x}$ the negation of $x$.) Given such a set of clauses $\phi$, the aim is to compute the maximum number of clauses which can be simultaneously satisfied over all assignments to the bits $x_1,\ldots, x_n$. Note the \emph{maximization} flavor of this problem; the correct generalization of this to the quantum setting is $k$-LH, since in $k$-LH the ground space may also be \emph{frustrated}, meaning it may be impossible to satisfy \emph{all} constraints simultaneously. If we instead demand that the ground space be \emph{frustration-free}, i.e. the ground state $\ket{\psi}\in(\complex^2)^n$ lives in the simultaneous ground spaces of each constraint $H_i$, then we are more accurately considering a quantum analogue of the NP-complete problem $k$-SAT, in which the question is given $\phi$, can we simultaneously satisfy \emph{all} constraints? This quantum analogue, denoted Quantum $k$-SAT ($k$-QSAT), was introduced by Bravyi in 2006~\cite{B06}, and is QMA$_1$-complete for $k\geq 3$~\cite{B06,GN13}. (See ``previous work'' below for further known results.) (Here, QMA$_1$ is QMA with perfect completeness, meaning in a YES instance of a given problem, the quantum proof must be accepted with certainty by the quantum verifier.)

Given the importance of such constraint satisfaction problems (classical or quantum) and their intractability (up to standard complexity theoretic conjectures such as $\class{P}\neq\class{NP}$ and $\class{BQP}\neq\class{QMA}_1$), much effort has been devoted by the classical community to approaches for MAX-$k$-SAT and $k$-SAT, including approximation algorithms, heuristic algorithms, and exact algorithms\footnote{The distinction between these three approaches is as follows. Approximation algorithms run in polynomial time, and provably provide a solution which is within some fixed prescribed relative error of optimal. A heuristic typically is also polynomial time, but does not provide any guarantee as to the accuracy of the solution found. An exact algorithm produces the optimal solution, but either takes superpolynomial time (as expected) but better than brute force time, or exactly solves a restricted case of the problem in polynomial time.}.  A particular direction of research involving \emph{exact} algorithms is what restrictions to the input are sufficient to allow polynomial-time solvability. In this paper, we focus on this theme, and ask:
\begin{quote}
\emph{Which special cases of $k$-QSAT can be solved efficiently on a classical computer?}
\end{quote}

Unfortunately (but also intriguingly), this problem appears to be markedly more difficult quantumly than classically. Consider the motivating example for this work: Classically, if each clause $c$ of a $k$-SAT instance can be matched with a unique variable $v_c$ (such that $v_c$ appears in clause $c$), then clearly the $k$-SAT instance is satisfiable. For example, in the instance $\phi=\set{\phi_1=x_1\vee \overline{x_2},\phi_2=\overline{x_1}\vee x_2}$, we can match $x_1$ with clause $\phi_1$ and $x_2$ with clause $\phi_2$. For any such ``matched'' $k$-SAT instance, finding a solution is trivial: Set variable $v_c$ to satisfy clause $c$. In our example, this means setting $x_1=1$ and $x_2=1$. (Note that the matching can be found efficiently via a reduction to network flow, which can then be solved in polynomial time via, e.g., the Ford-Fulkerson algorithm~\cite{FF56}.) Let us compare and contrast with the quantum setting. Quantumly, it has been known~\cite{LLMSS10} since 2010 that $k$-QSAT instances with ``matchings'' (also called a ``dimer covering'' in~\cite{LLMSS10}) are also satisfiable/frustration-free. Moreover, a ground state exists which is of a particularly nice form --- it can be assumed to be a {tensor product state}. Despite this, \emph{finding} said ground state efficiently has proven elusive (indeed, the proof of~\cite{LLMSS10} is \emph{non-constructive}). In complexity theoretic terms, we have a trivial NP \emph{decision} problem (since the answer is always YES, the system is frustration-free, and the proof is a tensor product state which can be described via polynomially many bits) whose analogous \emph{search} version is not known to be efficiently solvable. It should be pointed out here that classically, the question of whether decision versus search complexity coincide for NP problems is a long-standing open problem (see, e.g.,~\cite{BG94}), and so it is perhaps surprising that this phenomenon seems to potentially also play a role in ``matched'' QSAT. This is the starting point of the present work.

\paragraph*{Results and techniques.} Our results fall under three directions, all of which are related to $k$-QSAT with matchings, and which can very broadly be summarized as follows: (1) We show how to efficiently solve $k$-QSAT with sufficiently bounded occurrence of variables, (2) we give a general framework for parameterized algorithms for Quantum $k$-SAT, which in principle applies to \emph{any} Quantum $k$-SAT instance, (3) we attempt to characterize the set of all $3$-uniform hypergraphs (which encode interaction graphs for $3$-QSAT, formal definitions given shortly) with a ``dimer covering'', or formally, a ``system of distinct representatives''. Results (1) and (2) partially resolve the open problem of~\cite{LLMSS10} by giving efficient classical algorithms for special cases of ``matched QSAT'', and result (2) in particular gives (as far as we are aware) the first known parameterized algorithm for QSAT, obtaining exponential speedups over brute force in some cases (Section~\ref{sscn:runtime}). However, we are not yet able to fully resolve whether the decision and search complexity of ``matched QSAT'' coincide; indeed, as the structural graph theoretic study of (3) suggests, the geometry of all $3$-QSAT instances with matchings/dimer coverings appears quite complex.\\

\noindent\emph{Definitions and precise statements of results.} To more precisely state our results, we now require formal definitions. For this, we first define Quantum $k$-SAT ($k$-QSAT)~\cite{B06} and the notion of a system of distinct representatives (SDR). For $k$-QSAT, the input is a two-tuple $\Pi=(\set{\Pi_i=\ketbra{\psi_i}{\psi_i}}_i, \alpha)$ of rank $1$ projectors or \emph{clauses}\footnote{The original definition~\cite{B06} of $k$-QSAT did not require each $\Pi_i$ to be rank $1$. As in~\cite{LLMSS10}, we require the rank-$1$ condition to make our definitions and results well-defined and valid, respectively. Nevertheless, we do allow one to ``stack'' multiple clauses on a fixed set of $k$ qubits to simulate higher rank clauses.} $\Pi_{i}\in\lin{(\complex^2)^{\otimes k}}$, each acting non-trivially on a set of $k$ (out of $n$) qubits, and non-negative real number $\alpha>1/p(n)$ for some fixed polynomial $p$. The output is to decide whether there exists a satisfying assignment on $n$ qubits $\ket{\psi}\in (\complex^2)^{\otimes n}$, i.e. to distinguish between the cases $\Pi_i\ket{\psi}=0$ for all $i$ (YES case), or whether $\bra{\psi}\sum_i\Pi_i\ket{\psi}\geq \alpha$ (NO case). Note that $k$-QSAT generalizes $k$-SAT. As for a \emph{system of distinct representatives (SDR)} (which is formal terminology for the notion of ``matching'' for $k$-QSAT instances, as adopted from the combinatorics literature~\cite{SJ011}), given a set system such as a hypergraph $G=(V,E)$, an SDR is a set of vertices $V'\subseteq V$ such that each edge in $e\in E$ is paired with a distinct vertex $v_e\in V'$ such that $v_e\in e$. In previous work on QSAT, an SDR has been referred to as a ``dimer covering''~\cite{LLMSS10}.\\

\noindent \emph{1. Quantum $k$-SAT with bounded occurrence of variables.} Our first result concerns the natural restriction of limiting the number of times a variable can appear in a clause. For example, $3$-SAT with at most $3$ occurrences per variable is known to be NP-hard. We now show this threshold is tight\footnote{In addition to being NP-hard, whether $k$-QSAT with at most $3$ occurrences of variables is also QMA$_1$-hard is not clear. For example, the QMA$_1$-completeness constructions of~\cite{B06,GN13} utilize more than $3$ occurrences per variable.}, in the following sense.
\begin{theorem}\label{thm:bounded}
        There exists a polynomial time classical algorithm which, given an instance $\Pi$ of $k$-QSAT in which each variable occurs in at most $2$ clauses, outputs a satisfying product state if $\Pi$ is satisfiable, and otherwise rejects. Moreover, the algorithm works for clauses ranging from $1$-local to $k$-local in size.
\end{theorem}
 \noindent To show this, our idea is to ``partially reduce'' the $k$-QSAT instance to a $2$-QSAT instance. We then use the transfer matrix techniques of~\cite{B06,LMSS10,dBG15} (particularly the notion of \emph{chain reactions} from~\cite{dBG15}), along with a new notion of ``fusing'' chain reactions, to deal with the remaining clauses of locality at least $3$ in the instance.

 Whereas \emph{a priori} this setting may seem unrelated to the open question of computing solutions to $k$-QSAT instances with SDRs, we connect the two settings as follows. Denote the \emph{interaction hypergraph} $G=(V,E)$ of a $k$-QSAT instance as a $k$-uniform hypergraph (i.e. all edges have size precisely $k$), in which the vertices correspond to qubits, and each clause $c$ acting on a set of $k$ qubits $S_c$, is represented by a hyperedge of size $k$ containing the vertices corresponding to $S_c$.
\begin{theorem}\label{thm:matchgraph}
    Let $G=(V,E)$ be a hypergraph with all hyperedges of size at least $2$, and such that each vertex has degree at most $2$. Then, $G$ has an SDR.
\end{theorem}
\noindent Thus, Theorem~\ref{thm:bounded} resolves the open question of~\cite{LLMSS10} for $k$-QSAT instances with SDRs in which additionally (1) each variable occurs in at most two clauses and (2) there are no $1$-local clauses. (The latter of these restrictions is necessary, as allowing edges of size $1$ easily makes Theorem~\ref{thm:matchgraph} false in general; see Section~\ref{sscn:SDR}.)
\\

 \noindent \emph{2. Parameterized algorithms for Quantum $k$-SAT.} Our next result, and the main contribution of this paper, gives a parameterized algorithm\footnote{Roughly, parameterized complexity characterizes the complexity of computational problems with respect to specific \emph{parameters} of interest other than just the input size. For example, for an NP-complete graph theoretic problem $\Pi$ on $n$ vertices, a parameterized algorithm might run in time polynomial in $n$, but potentially exponentially with respect to some other graph parameter, such as the treewidth of the input graph. In this example, $\Pi$ would be tractable on input graphs of logarithmic treewidth. Parameterized complexity is the focus of an entire research community classically; see ``previous work'' for further details.} for explicitly computing (product state) solutions for a non-trivial class of $k$-QSAT instances. As discussed in Section~\ref{sscn:runtime}, this algorithm in some cases provides an \emph{exponential} speedup over brute force diagonalization.\\\vspace{-1mm}

 \noindent\emph{Transfer filtrations.} To sketch the idea of the algorithm at a high level, we first introduce a new graph theoretic notion of a \emph{transfer filtration} for a $k$-uniform hypergraph $G=(V,E)$ (see Definition~\ref{def:2} for a formal definition). Intuitively, a ``transfer filtration of type $b>0$'' should be thought of as a set of $b$ of qubits (out of all $n$ qubits the $k$-QSAT instance acts on) which form the ``hard core'' or \emph{foundation} of the instance. Roughly, this is in the sense that \emph{if} one could determine what states to assign to the qubits in the foundation, then resolving the remainder of the instance would be ``easy'' (this is not entirely accurate, hence the quotes on ``easy'').\\\vspace{-1mm}

 \noindent\emph{Algorithm sketch.} With the notion of transfer filtration in hand, our framework for attacking $k$-QSAT can be sketched at a high level as follows.
 \begin{enumerate}
    \item (Decoupling step) First, given a $k$-QSAT instance $\Pi$ on $G$ with transfer filtration of type $b$, we ``blow-up'' $\Pi$ to a larger, \emph{decoupled} instance $\Pi^+$ (Decoupling Lemma, Lemma~\ref{lem:p}).

    \item (Solve decoupled instance) The decoupled nature of $\Pi^+$ makes it ``easier'' to solve (Transfer Lemma, Lemma~\ref{prop:7}), in that any assignment to the $b$ ``foundation'' qubits can be \emph{extended} to a solution to all of $\Pi^+$. This raises the question --- how does one map the solution of $\Pi^+$ back to a solution of $\Pi$?

    \item (Map solutions back to original instance) We next give a set of ``qualifier'' constraints $\set{h_s}$ (Qualifier Lemma, Lemma~\ref{l:qualifier}) acting on \emph{only the $b$ foundation qubits}, with the following strong property: If a (product state) assignment $\ve{v}$ to the $b$ foundation qubits satisfies the constraints $\set{h_s}$, then not only can we extend $\ve{v}$ via the Transfer Lemma to a full solution for $\Pi^+$ as in Step 1 above, but we can also map this extended solution {back} to one for the \emph{original} $k$-QSAT instance $\Pi$. (For clarity, there is an additional required technical condition on the \emph{transfer functions} $g_i$ in the Transfer Lemma, which we circumvent in our main Theorem~\ref{theorem}.)
 \end{enumerate}
 Once the framework above is developed, we show that it applies to the non-trivial family of $k$-QSAT instances whose $k$-uniform hypergraph $G=(V,E)$ has a transfer filtration of type $b=\abs{V}-\abs{E}+1$. This family includes, e.g., the semi-cycle of Figure~\ref{fig:chain}, (a slight modification of) the tiling of the torus (Figure~\ref{fig:torus}), and ``fir tree'' (Figure~\ref{fig:firtree}). Our main result (Theorem~\ref{theorem}) says the following: For any $k$-QSAT instance $\Pi$ on such a $G$ and whose constraints are generic (see Section~\ref{scn:parameterized}), computing a (product state) solution to $\Pi$ reduces to solving for a root of a \emph{single univariate} (see Remark~\ref{rem:reduction}) polynomial $P$ --- \emph{any} root thereof (which always exists if the field $\K$ is algebraically closed) can then be extended back to a full solution for $\Pi$.\\\vspace{-1mm}

\noindent \emph{Advantages of framework.} The key advantage of this approach, and what makes it a parameterized algorithm, is the following --- the \emph{degree} of polynomial $P$, and hence the runtime of the algorithm, scale exponentially only in the foundation size $b$ and a ``radius'' parameter $r$ of the transfer filtration (assuming $k\in O(1)$, see Equation~(\ref{eqn:runtime}) for an explicit bound on runtime). Thus, given a transfer filtration where $b$ and $r$ are at most logarithmic, finding a (product state) solution to $k$-QSAT reduces to solving for a single root over $\complex$ for a single univariate polynomial $P$ of polynomial degree, which can be done in polynomial time~\cite{SCHONHAGE1985118,S86}. Indeed, in Section~\ref{scn:parameterized} we give a non-trivial family of $k$-uniform hypergraphs, denoted Crash (Figure~\ref{fig:crash}), for which our algorithm runs in polynomial time, whereas brute force diagonalization would require exponential time.

If, on the other hand, the foundation size $b$ and radius $r$ of the input $k$-QSAT instance are \emph{super}logarithmic, then our algorithm requires superpolynomial time. Even here, however, the framework yields a new result: it gives a \emph{constructive} proof that all $k$-QSAT instances satisfying the preconditions of Theorem~\ref{theorem} have a (product state) solution. In particular, in Corollary~\ref{cor:actually} and Theorem~\ref{thm:general}, we observe that such hypergraphs must have SDRs, and so we {constructively} reproduce the result of \cite{LLMSS10} that any $3$-QSAT instance with an SDR is satisfiable by a product state (again, assuming the additional conditions of Theorem~\ref{theorem} are met).

 Finally, although this result stems primarily from tools of projective algebraic geometry (AG), the presentation herein avoids any explicit mention of AG terminology (with the exception of defining the term ``generic'' in Section~\ref{sscn:generic}) to be accessible to readers without an AG background. For completeness, a brief overview of the ideas in AG terms is given at the end of Section~\ref{scn:parameterized}.\\

 \noindent \emph{3. A study of $3$-uniform hypergraphs with SDRs.} Given that computing ground states for $k$-QSAT instances with SDRs appears challenging, in the hope of helping guide future studies on the topic, our final contribution aims to better understand the set of all $k$-QSAT instances with SDRs, particularly in the boundary case when $\abs{E}=\abs{V}$. Unfortunately, this in itself appears to be a difficult task (if not potentially impossible, see comments about a ``finite characterization'' below). Our results here are as follows. We first give various characterizations involving intersecting families (i.e. when each pair of edges has non-empty intersection). We then study the setting of \emph{linear} hypergraphs (i.e. each pair of edges intersects in at most one vertex), which are generally more complex. (For example, the set of edge-intersection graphs of $3$-uniform linear hypergraphs is known not to have a ``finite'' characterization in terms of a finite list of forbidden induced subgraphs~\cite{NAIK1982159}.) We study ``extreme cases'' of linear hypergraphs with SDRs, such as the Fano plane and ``tiling of the torus'', and in contrast to these two examples, demonstrate a (somewhat involved) linear hypergraph we call the iCycle which also satisfies the \emph{Helly property} (which generalizes the notion of ``triangle-free'').

 A main conclusion of this study is that even with multiple additional restrictions in place (e.g. linear, Helly), the set of $3$-uniform hypergraphs with SDRs remains non-trivial. To complement these results, we show how to fairly systematically construct large linear hypergraphs with $\abs{E}=\abs{V}$ without SDRs. We hope this work highlights the potential complexity involved in dealing with even the ``simple'' case of \tqsat\ with SDRs.

 \paragraph*{Discussion.}
Regarding our parameterized algorithm, we stated earlier that our framework applies to the non-trivial family of $k$-QSAT instances whose $k$-uniform hypergraph $G=(V,E)$ has a transfer filtration of ``type $b=\abs{V}-\abs{E}+1$''. Let us elaborate on this further, and in particular stress which parts of our framework require this additional assumption on $b$, and which parts of the framework apply to arbitrary instances of $k$-QSAT.

First, our notions of transfer filtrations and blow-ups apply to \emph{any} instance of $k$-QSAT (and thus also\footnote{For the special case of $k$-SAT, note that it is not \emph{a priori} clear that having a transfer filtration with a small foundation suffices to solve the system trivially. This is because the genericity assumption on constraints, which $k$-SAT constraints do not satisfy, is required to ensure that any assignment to the foundation propagates to all bits in the instance. Thus, the brute force approach of iterating through all $2^b$ assignments to the foundation does not obviously succeed.} $k$-SAT), including $\class{QMA}_1$-complete instances. (For example, every $k$-uniform hypergraph has a trivial foundation obtained by iteratively removing vertices until the resulting set contains no edges. A key question\footnote{Whereas our algorithm's runtime benefits greatly from small foundation size, a natural question is whether a ``small'' versus ``large'' foundation generally should make $k$-QSAT easier or harder. Intuitively, indeed one expects a smaller foundation to yield ``easier'' instances of $k$-QSAT, as recall the foundation captures the hard ``core'' of the $k$-QSAT instance (in that, roughly speaking, assignments to the foundation qubits generically force assignments onto the remaining, non-foundation, qubits). In the ``Open Questions'' paragraph of Section~\ref{scn:intro}, we conjecture that computing the minimum foundation size is NP-hard of a given input hypergraph is actually NP-hard (similar to how computing common parameterized algorithmic parameters such as treewidth is NP-hard~\cite{ACP87}).} is how \emph{small} the foundation size $b$ and radius $r$ of the filtration can be chosen for a given hypergraph, as our algorithm's runtime scales exponentially in these parameters; see Equation~(\ref{eqn:runtime}).) More precisely, our techniques in Section~\ref{scn:parameterized}, up to and including the Qualifier Lemma, apply to arbitrary $k$-QSAT instances. If one also assumes that constraints are \emph{generic} (but still on an arbitrary hypergraph), then the Surjectivity Lemma (Lemma~\ref{l:surjectivity}) also holds.

The remaining question is Phase 3 of the framework (``Map solutions back to original instance'' above) --- when can \emph{local} solutions to the qualifier constraints be extended to \emph{global} solutions for the entire $k$-QSAT instance? Answering this requires solving for a common root of a set of high-degree multi-variate polynomials (hence the connection to algebraic geometry). This is where we consider the restriction to input $k$-uniform hypergraphs of transfer type $b=\abs{V}-\abs{E}+1$, which allows us to reduce the entire framework to the solution of a \emph{single univariate} (high-degree) polynomial (to which one can now apply the root-finding algorithm of Sch\"{o}nhage~\cite{S86}; see Section~\ref{sscn:runtime}). We stress that the family of $k$-QSAT instances with $b=\abs{V}-\abs{E}+1$ is non-trivial, and includes the semi-cycle (Figure~\ref{fig:chain}), tiling of the torus (Figure~\ref{fig:torus}), fir tree (Figure~\ref{fig:firtree}, and crash (Figure~\ref{fig:crash}) families of hypergraphs. In particular, for the latter family we obtain exponential speedups over brute force diagonalization in Section~\ref{sscn:runtime}. For clarity, all runtimes in this paper are rigorous worst-case runtimes.

\paragraph*{Previous work.} Quantum $k$-SAT was introduced by Bravyi \cite{B06}, who gave an efficient (quartic time) algorithm for $2$-QSAT, and showed that $4$-QSAT is QMA$_1$-complete. Subsequently, Gosset and Nagaj~\cite{GN13} showed that Quantum $3$-SAT is also QMA$_1$-complete, and independently and concurrently, Arad, Santha, Sundaram, Zhang~\cite{ASSZ15} and de Beaudrap, Gharibian~\cite{dBG15} gave linear time algorithms for $2$-QSAT. The original inspiration for this paper was the work of Laumann, L\"{a}uchli, Moessner, Scardicchio and Sondhi~\cite{LLMSS10}, which showed \emph{existence} of a product state solution for any $k$-QSAT instance with an SDR. Thus, the decision version of $k$-QSAT with SDRs is in NP and trivially efficiently solvable. However, whether the search version (i.e. compute an explicit satisfying assignment) is also polynomial-time solvable remains open. The question of whether the decision and search complexities of NP problems are the same is a longstanding open problem in complexity theory; conditional results separating the two are known (see e.g. Bellare and Goldwasser~\cite{BG94}).

In terms of approaches to $k$-QSAT, a complementary and well-studied tool which should also be mentioned is the Quantum Lov\'{a}sz Local Lemma (QLLL), introduced by Ambainis, Kempe, and Sattath~\cite{AKS12}. The QLLL non-constructively guarantees frustration-freeness of the input Hamiltonian under certain conditions. Subsequently, Schwarz, Cubitt and Verstraete~\cite{SCV13} and Sattath and Arad~\cite{SA15} provided \emph{constructive} (i.e. algorithmic) versions of QLLL under the assumption of commuting local terms~\cite{SCV13, SA15}, and more recently Gily\'{e}n and Sattath~\cite{GS17} gave an algorithm QLLL for non-commuting terms by introducing a uniformly gapped assumption. The QLLL has since been extended to Shearer's bound by Sattath, Morampudi, Laumann and Moessner~\cite{SMLM16}, in that a quantum analogue of Shearer's bound is a sufficient criterion for frustration-freeness. Most recently, He, Li, Sun, and Zhangit showed~\cite{HLSZ18} that in fact Shearer's bound is tight in the quantum setting. We remark that the approaches in these works appear distinct from those taken here; for example, the QLLL allows certain cases of $k$-QSAT to be resolved even when it is not clear the ground space contains a tensor product state, whereas our algorithms search for product state solutions (guaranteed to exist for any $k$-QSAT instance with an SDR, which recall was the motivation for this work). This also means algorithmic QLLL results such as~\cite{GS17} use \emph{quantum} algorithms to prepare a ground state, whereas our framework here is classical. As a result, our work should be viewed as a complementary approach for tackling $k$-QSAT instances, which introduces genuinely distinct tools such as parameterized algorithms and blowups from algebraic geometry.

 Finally, in terms of classical $k$-SAT, in stark contrast to $k$-QSAT, solutions to $k$-SAT instances with an SDR can be efficiently computed. As for parameterized complexity,  classically it is a well-established field of study (see, e.g.,~\cite{Downey:2012:PC:2464827} for an overview). The parameterized complexity of SAT and $\#$SAT, in particular, has been studied by a number of works, such as~\cite{Szeider2004,FISCHER2008511,SAMER201050,Ganian:2013:BAS:2594865.2594870,Saether:2015:SSM:2910557.2910559,Paulusma:2016:MCC:2989743.2989766,Gaspers2016OnSP}, which consider a variety of parameterizations including based on tree-width, modular tree-width, branch-width, clique-width, rank-width, and incidence graphs which are interval bipartite graphs. Regarding parameterized complexity of \emph{Quantum} SAT, as far as we are aware, our work is the first to initiate a ``formal'' study of the subject. However, we should be clear that existing works in Quantum Hamiltonian Complexity~\cite{O11,GHLS14} have long implicitly used ``parameterized'' ideas. For example, Markov and Shi~\cite{MS08} give a classical simulation algorithm for quantum circuits whose runtime scales polynomially in the size of the circuit, but exponentially in its treewidth. Their algorithm is based on \emph{tensor networks} (e.g.~\cite{RTPCTSL20}), whose bond dimension can be viewed as a parameter constraining the complexity of the contraction. In particular, 1d tensor networks whose bond dimension is at most polynomial in the length of the chain can be contracted efficiently (a fact used, for example, in the polynomial time algorithm of Landau, Vidick and Vazirani~\cite{LVV13} for solving the 1D gapped Local Hamiltonian~\cite{KSV02} problem). This is in stark contrast to 2D tensor networks (often denoted PEPS, for Projected Entangled Pair States), which are $\#$P-hard to contract even for constant bond dimension~\cite{SWVC07}.

 \paragraph{Open questions.} We close with a number of open questions. We showed that $k$-QSAT with at most $2$ occurences per variable is efficiently solvable (Theorem~\ref{thm:bounded}), but $3$ occurrences per variable is well-known to be at least NP-hard (for $3$-SAT, which is a special case of $3$-QSAT). Is $3$-QSAT with at most $3$ occurrences per variable similarly QMA$_1$-complete? Can ideas from classical parameterized complexity be generalized to the quantum setting? We have developed a number of tools herein for studying Quantum SAT --- can these be applied in more general settings, for example beyond the families of $k$-QSAT instances considered in Theorem~\ref{theorem}? The ``parameters'' in our results of Section~\ref{scn:parameterized} include the radius of a transfer filtration --- whether a transfer filtration (of a fixed type $b$) of \emph{minimum} radius can be computed efficiently, however, is left open for future work. Similarly, it is not clear that given $b\in \nats$, the problem of deciding whether a given hypergraph $G$ has a transfer filtration of type at most $b$ is in P. We conjecture, in fact, that this latter problem is NP-complete. Finally, the question of whether solutions to arbitrary instances of $k$-QSAT with SDRs can be computed efficiently (recall they are guaranteed to exist~\cite{LLMSS10}) remains open.

 \paragraph*{Organization.} Section~\ref{scn:preliminaries} gives basic notation and definitions. Section~\ref{scn:bounded} gives an efficient algorithm for \tqsat\ with bounded occurrence of variables, which also introduces the notions of transfer matrices (which are generalized via transfer functions in Section~\ref{scn:parameterized}). Our main result is given in Section~\ref{scn:parameterized}, and concerns a new parameterized complexity approach for solving $k$-QSAT. The algorithm and its runtime, along with the study of asymptotic speedups, are given in Section~\ref{sscn:runtime}. Finally, Section~\ref{scn:hypergraphs} conducts a structural graph theoretic study of hypergraphs with SDRs.

\section{Preliminaries}\label{scn:preliminaries}

\paragraph*{Notation.} For complex Euclidean space $\spa{X}$, $\lin(\spa{X})$ denotes the set of linear operators mapping $\spa{X}$ to itself. For unit vector $\ket{\psi}\in\complex^2$, the unique orthogonal unit vector (up to phase) is denoted $\ket{\psi^\perp}$, i.e. $\braket{\psi}{\psi^\perp}=0$.

We now give some definitions from graph theory, some of which are used primarily in Section~\ref{scn:hypergraphs}.

\begin{definition}[Hypergraph]
A $\emph{hypergraph}$ is a pair $G=(V,E)$ of a set $V$ (vertices), and a family $E$ (edges) of subsets of $V$. If each vertex has degree $d$, we say $G$ is $d$-regular. Alternatively, when convenient we use $V(G)$ and $E(G)$ to denote the vertex and edge sets of $G$, respectively. A \emph{simple} hypergraph has no repeated edges, i.e. if $e_i \subseteq e_j$, then $i=j$. We say $G$ is $k$-uniform if all edges have size $k$.
\end{definition}

\begin{definition}[Chain~\cite{KATONA20161884}, Figure~\ref{fig:chain}]\label{definition:chain}
A $k$-uniform hypergraph $G=(V,E)$ is a \emph{chain} if there exists a sequence $(v_1,v_2,...,v_l)\in V^l$ for $l\geq n$ such that (1) the sequence contains all elements of $V$ at least once, (2) $v_1\neq v_l$, and (3) $E=\bigcup_{1\leq i \leq l-k+1}e_i$ for distinct edges $e_i=\set{v_i,v_{i+1},...,v_{i+k-1}}$. The length of the chain $G$ is $m=l-k+1$.
\end{definition}
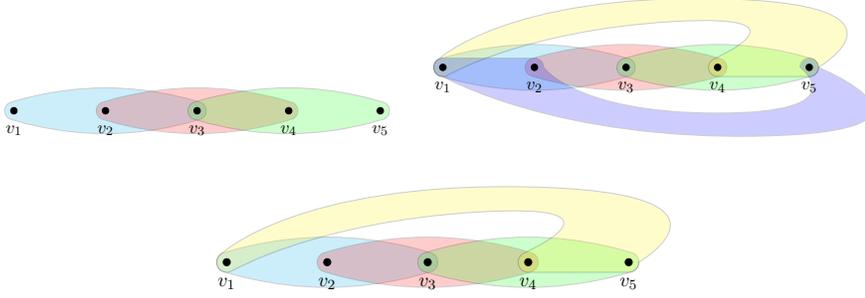
\begin{figure}[t]
\begin{center}
\resizebox{370pt}{!}{%
\begin{tikzpicture}

    \coordinate (v1) at (-4,0);
    \coordinate (v2) at (-2,0) {};
    \coordinate (v3) at (0,0);
    \coordinate (v4) at (2,0) {};
    \coordinate (v5) at (4,0) {};

    \foreach \v in {v1,v2,v3,v4,v5}{
      \node [circle, minimum size=0.4cm, line width=0pt] (\v') at (\v) {};
    }

    \filldraw [draw=black, fill=cyan, opacity=0.2, line width=0.5pt]
		plot[smooth, tension=1] coordinates { (v1'.120) (-2,0.5) (v3'.60)} arc (60:-60:0.2)--plot[smooth, tension=1] coordinates { (v3'.-60) (-2,-0.5) (v1'.-120)} arc(-120:-240:0.2);
    \filldraw [draw=black, fill=red, opacity=0.2, line width=0.5pt]
		plot[smooth, tension=1] coordinates { (v2'.120) (0,0.5) (v4'.60)} arc (60:-60:0.2)--plot[smooth, tension=1] coordinates { (v4'.-60) (0,-0.5) (v2'.-120)} arc(-120:-240:0.2);
    \filldraw [draw=black, fill=green, opacity=0.2, line width=0.5pt]
		plot[smooth, tension=1] coordinates { (v3'.120) (2,0.5) (v5'.60)} arc (60:-60:0.2)--plot[smooth, tension=1] coordinates { (v5'.-60) (2,-0.5) (v3'.-120)} arc(-120:-240:0.2);
  	
    \foreach \l in {1,...,5}{
      \filldraw [black] (v\l) circle (2pt) node [inner sep=5pt, label=below:$v_{\l}$] {};
    }

\end{tikzpicture}
\hspace{5mm}
\begin{tikzpicture}
	    \coordinate (v1) at (-4,0);
    \coordinate (v2) at (-2,0) {};
    \coordinate (v3) at (0,0);
    \coordinate (v4) at (2,0) {};
    \coordinate (v5) at (4,0) {};

    \foreach \v in {v1,v2,v3,v4,v5}{
      \node [circle, minimum size=0.4cm, line width=0pt] (\v') at (\v) {};
    }

    \filldraw [draw=black, fill=cyan, opacity=0.2, line width=0.5pt]
		plot[smooth, tension=1] coordinates { (v1'.120) (-2,0.5) (v3'.60)} arc (60:-60:0.2)--plot[smooth, tension=1] coordinates { (v3'.-60) (-2,-0.5) (v1'.-120)} arc(-120:-240:0.2);
    \filldraw [draw=black, fill=red, opacity=0.2, line width=0.5pt]
		plot[smooth, tension=1] coordinates { (v2'.120) (0,0.5) (v4'.60)} arc (60:-60:0.2)--plot[smooth, tension=1] coordinates { (v4'.-60) (0,-0.5) (v2'.-120)} arc(-120:-240:0.2);
    \filldraw [draw=black, fill=green, opacity=0.2, line width=0.5pt]
		plot[smooth, tension=1] coordinates { (v3'.120) (2,0.5) (v5'.60)} arc (60:-60:0.2)--plot[smooth, tension=1] coordinates { (v5'.-60) (2,-0.5) (v3'.-120)} arc(-120:-240:0.2);
  	\filldraw [draw=black, fill=yellow, opacity=0.2, line width=0.5pt]
		plot [smooth, tension=2] coordinates {(v1'.120) (2,1.5) (v5'.0)}
	arc(0: -90:0.2) -- (v4'.270) arc(270:120:0.2)
	plot [smooth, tension=2] coordinates {(v4'.120) (1,1) (v1'.-60)} arc(-60:-240:0.2);
	
	\filldraw [draw=black, fill=blue, opacity=0.2]
	plot [smooth, tension=2] coordinates {(v1'.270) (3,-1.5) (v5'.0)}
	arc(0:180:0.2)
	plot [smooth, tension=2] coordinates {(v5'.180) (2,-1) (v2'.0)}
	arc(0:90:0.2) -- (v1'.90) arc(90:270:0.2)
	;

    \foreach \l in {1,...,5}{
      \filldraw [black] (v\l) circle (2pt) node [inner sep=5pt, label=below:$v_{\l}$] {};
    }
\end{tikzpicture}
}
\\
\vspace{5mm}
\resizebox{210pt}{!}{%
\begin{tikzpicture}
	
    \coordinate (v1) at (-4,0);
    \coordinate (v2) at (-2,0) {};
    \coordinate (v3) at (0,0);
    \coordinate (v4) at (2,0) {};
    \coordinate (v5) at (4,0) {};

    \foreach \v in {v1,v2,v3,v4,v5}{
      \node [circle, minimum size=0.4cm, line width=0pt] (\v') at (\v) {};
    }

    \filldraw [draw=black, fill=cyan, opacity=0.2, line width=0.5pt]
		plot[smooth, tension=1] coordinates { (v1'.120) (-2,0.5) (v3'.60)} arc (60:-60:0.2)--plot[smooth, tension=1] coordinates { (v3'.-60) (-2,-0.5) (v1'.-120)} arc(-120:-240:0.2);
    \filldraw [draw=black, fill=red, opacity=0.2, line width=0.5pt]
		plot[smooth, tension=1] coordinates { (v2'.120) (0,0.5) (v4'.60)} arc (60:-60:0.2)--plot[smooth, tension=1] coordinates { (v4'.-60) (0,-0.5) (v2'.-120)} arc(-120:-240:0.2);
    \filldraw [draw=black, fill=green, opacity=0.2, line width=0.5pt]
		plot[smooth, tension=1] coordinates { (v3'.120) (2,0.5) (v5'.60)} arc (60:-60:0.2)--plot[smooth, tension=1] coordinates { (v5'.-60) (2,-0.5) (v3'.-120)} arc(-120:-240:0.2);
  	\filldraw [draw=black, fill=yellow, opacity=0.2, line width=0.5pt]
		plot [smooth, tension=2] coordinates {(v1'.120) (2,1.5) (v5'.0)}
	arc(0: -90:0.2) -- (v4'.270) arc(270:120:0.2)
	plot [smooth, tension=2] coordinates {(v4'.120) (1,1) (v1'.-60)} arc(-60:-240:0.2);

    \foreach \l in {1,...,5}{
      \filldraw [black] (v\l) circle (2pt) node [inner sep=5pt, label=below:$v_{\l}$] {};
    }

\end{tikzpicture}
}
\end{center}
\caption{(Top left) A chain. (Top right) A cycle. (Bottom middle) A semicycle. Note that, in contrast to the cycle, the semicycle is missing the blue edge $\set{v_1,v_2,v_5}$.}
\label{fig:chain}
\end{figure}

\begin{definition}[Cycle~\cite{KATONA20161884}, Figure~\ref{fig:chain}]\label{definition:cycle}
A $k$-uniform hypergraph $G=(V,E)$ is a \emph{cycle} if there exists a sequence $(v_1,v_2,...,v_l)\in V^l$ for $l\geq n$ such that (1) the sequence contains all elements of $V$ at least once, (2) for all $1\leq i \leq l$, $e_i=\set{v_i,v_{i+1},...,v_{i+k-1}}$ are distinct edges in $E$, where indices are understood modularly. The length of the cycle $G$ is $m=l$.	
\end{definition}
\begin{definition}[Semicycle~\cite{KATONA20161884}, Figure~\ref{fig:chain}]\label{definition:semicyle}
A $k$-uniform hypergraph $G=(V,E)$ is a \emph{semicycle} if there exists a sequence $(v_1,v_2,...,v_l)\in V^l$ for $l\geq n$ such that (1) the sequence contains all elements of $V$ at least once, (2) $v_1=v_l$, and (3) for all $1\leq i \leq l-k+1$, $e_i=\set{v_i,v_{i+1},...,v_{i+k-1}}$ are distinct edges of $G$. The length of the semicycle is $m=l-k+1$.	
\end{definition}
\begin{definition}[Tight star~\cite{KATONA20161884}, Figure~\ref{fig:hyperstar}]
A $k$-uniform hypergraph $G=(V,E)$ is a \emph{tight star} if there exists a set of $k-1$ vertices $S$, such that the intersection of any pair of edges is precisely $S$.
\end{definition}

\begin{definition}[$t$-stacked set]
Let $G=(V,E)$ be a $k$-uniform hypergraph. A subset $S\subseteq E$ of $t$ edges is called a \emph{$t$-stacked set} if every edge in $S$ is incident to the same $k$ vertices in $V$, i.e.\ $\abs{\bigcup_{e\in S}e}=k$.
\end{definition}

\section{Quantum SAT with bounded occurrence of variables}\label{scn:bounded}

In this section, we study $k$-QSAT when each qubit occurs in at most two constraints. For this, we first recall tools from the study of Quantum $2$-SAT~\cite{B06,LMSS10,dBG15}. Recall that throughout this paper, each clause is assumed to be rank $1$ (in general, however, one can ``stack'' multiple rank $1$ clauses on a set of vertices to simulate a higher rank clause).

\paragraph*{Transfer matrices, chain reactions, and cycle matrices.} For any rank-$1$ constraint $\Pi_i=\ketbra{\psi}{\psi}\in \lin((\complex^2)^{\otimes k})$, consider Schmidt decomposition $\ket{\psi}=\alpha\ket{a_0}\ket{b_0}+\beta\ket{a_1}\ket{b_1}$, where $\ket{a_i}\in(\complex^2)^{\otimes (k-1)}$ lives in the Hilbert space of the first $k-1$ qubits and $\ket{b_i}\in\complex^2$ the last qubit. Then, the \emph{transfer} matrix $T_{\psi}:(\complex^2)^{\otimes k-1}\mapsto\complex^2$ is given by $T_{\psi} = \beta \ket{b_0}\bra{a_1} - \alpha \ket{b_1}\bra{a_0}$. In words, given any assignment $\ket{\phi}$ to the first $k-1$ qubits, if $T_{\psi}\ket{\phi}\in\complex^2$ is non-zero, then it is the \emph{unique} assignment to qubit $k$ (given $\ket{\phi}$ on qubits $1$ to $k-1$) which satisfies $\Pi_i$.

In the special case of $k=2$, transfer matrices are particularly useful. Consider first a $2$-QSAT interaction graph (which is a $2$-uniform hypergraph, or just a graph) $G=(V,E)$ which is a path, i.e. a sequence of edges $e_1=(v_1,v_2),e_2=(v_2,v_3),\ldots,e_m=(v_{m-1},v_m)$ for distinct $v_i\in V$, and where edge $e_i$ corresponds to constraint $\ket{\psi_i}$. Then, any assignment $\ket{\phi}\in\complex^2$ to qubit $1$ \emph{induces a chain reaction (CR)} in $G$, meaning qubit $2$ is assigned $T_{\psi_1}\ket{\phi}$, qubit $3$ is assigned $T_{\psi_2}T_{\psi_1}\ket{\phi}$, and so forth. If this CR terminates before all qubits labelled by $V$ receive an assignment, which occurs if $T_{\psi_i}\ket{\phi'}=0$ for some $i$, this means that constraint $i$ (acting on qubits $i$ and $i+1$) is satisfied by the assignment $\ket{\phi'}$ to qubit $i$ alone, and no residual constraint is imposed on qubit $i+1$. Thus, the graph $G$ is reduced to a path $e_{i+1},\ldots, e_m$. In this case, we say the CR is \emph{broken}. Note that if $G$ is a path, then it is a satisfiable $2$-QSAT instance with a product state solution.

Finally, consider a $2$-QSAT instance whose interaction graph $G$ is a cycle $C=(v_1,\ldots, v_{m+1})$ with $m$. Then, a CR induced on vertex $v_1$ with any assignment $\ket{\psi}\in\complex^2$ will in general propagate around the cycle and impose a consistency constraint on $v_1$. Formally, denote the product $T_C=T_{\psi_{m}}\cdots T_{\psi_1}\in\lin(\complex^2)$ as the \emph{cycle matrix} of $C$. Then, if the cycle matrix is not the zero matrix, it be shown that the satisfying assignments for the cycle are precisely the eigenvectors of $T_C$. (If $T_C=0$, any assignment on $v_1$ will only propagate partially around the cycle, thus decoupling the cycle into two paths.) Thus, if $G$ is a cycle, then it also has a product state solution.

In this section, when we refer to ``solving the path or cycle'', we mean applying the transfer matrix techniques above to efficiently compute a product state solution to the path or cycle.

\paragraph*{$k$-QSAT with bounded occurence of variables.} We now restate and prove Theorem~\ref{thm:bounded}. Subsequently, we demonstrate the algorithm on an example (Figure~\ref{fig:example}) and discuss its applicability more generally. Section~\ref{sscn:SDR} shows the connection between Theorem~\ref{thm:bounded} and $k$-QSAT with SDRs.

\begin{reptheorem}{thm:bounded}
        There exists a polynomial time classical algorithm which, given an instance $\Pi$ of $k$-QSAT in which each variable occurs in at most two clauses, outputs a satisfying product state if $\Pi$ is satisfiable, and otherwise rejects. Moreover, the algorithm works for clauses ranging from $1$-local to $k$-local in size.
\end{reptheorem}
\begin{proof}
We begin by setting terminology. Let $\Pi$ be an instance of $k$-QSAT with $k$-uniform interaction graph $G=(V,E)$. For any clause $c$, let $Q_c$ denote the set of qubits acted on $c$, i.e. $Q_c$ is the edge in $G$ representing $c$. We say $c$ is \emph{stacked} if $Q_c$ is contained within another edge/clause $Q_{c'}$, i.e. if there exists $c'\neq c$ such that $Q_c\subseteq Q_{c'}$. For a qubit $v$, we use shorthand $\ket{v}$ to denote the current assignment from $\complex^2$ to $v$. For a clause $c$, $\ket{c}$ denotes the bad subspace of $c$, i.e. clause $c$ is given by rank-$1$ projector $I-\ketbra{c}{c}$. The set of clauses vertex $v$ appears in is denoted $C_v$. For any assignment $\ket{v}$, we introduce shorthand $S_{\ket{v}}=\set{\braket{v}{c}\mid c\in C_v}\subseteq\bigcup_{i=0}^{k-1}\complex^{2^i}$, where recall $c$ can be a clause on $1,\ldots,k$ qubits, and we implicitly assume $\bra{v}$ acts as the identity on the qubits of $c$ which are not $v$. For example, if $c$ acts on qubits $v,w,x$, then $\braket{v}{c}=(\bra{v}\otimes I_{w,x})\ket{c}\in\complex^4$ is the residual constraint on qubits $w,x$, given assignment $\ket{v}$ to $v$. Thus, $S_{\ket{v}}$ is the set of constraints we obtain by taking the clauses in $C_v$, and projecting down qubit $v$ in each clause onto assignment $\ket{v}$. As a result, the clauses in $S_{\ket{v}}$ no longer act on $v$. The algorithm we design will satisfy that the only possible element of $\complex$ in $S_{\ket{v}}$ is $0$, which can only be obtained by projecting a constraint $\ket{c}\in\complex^2$ onto its orthogonal complement to satisfy it; thus, we can assume without loss of generality that $S_{\ket{v}}\subseteq\set{\braket{v}{c}\mid c\in C_v}\subseteq\bigcup_{i=1}^{k-1}\complex^{2^i}$. Finally, we say two $1$-local clauses $\ket{c},\ket{c'}\in\complex^2$ \emph{conflict} if $\ket{c}$ and $\ket{c'}$ are linearly independent (i.e. $\ketbra{c}{c}+\ketbra{c'}{c'}$ has an empty null space, meaning there is no satisfying assignment).

The algorithm proceeds as follows. Let $\Pi$ satisfy the conditions of our claim. For clarity, any time a CR on a path is broken by a transfer matrix $T_{\psi}$ on edge $(u,v)$, i.e. $T_{\psi}\ket{u}=0$, we shall implicitly assume we continue by choosing assignment $\ket{0}$ on $v$ to induce a new CR and continue solving the path. (This is important for the correctness analysis later in which Step 3 must not create a $1$-local constraint.)

\paragraph*{Statement of algorithm.} The intuitive idea behind the algorithm is to ``partially reduce'' $\Pi$ to a $2$-QSAT instance, and use the transfer matrix techniques outlined above to solve this segment. Combining this with a new notion of fusing CRs, the technique can be applied iteratively to reduce $k$-local constraints to $2$-local ones until the entire instance is solved.  More informally, the algorithm exploits the extra degrees of freedom which arise when each qubit appears in at most two clauses. In particular, by considering clauses in the ``right'' order below, we are able to (iteratively) greedily satisfy a given clause locally, and propagate the implications of such an assignment via transfer matrix techniques. \\

\noindent\emph{Algorithm A.}
\begin{enumerate}
    \item While there exists a $1$-local constraint $c$ acting on some qubit $v$:
        \begin{enumerate}
            \item If $c$ conflicts with another $1$-local clause on $v$, reject. Else, set $\ket{v}=\ket{c^\perp}\in\complex^2$. Set\footnote{Setting $C_v=S_{\ket{v}}$ means we update all clauses $c$ acting on $v$ by projecting qubit $v$ of $c$ onto $\ket{v}$. Note that each $v$ does not keep a local copy of $C_v$, i.e. all clauses are referenced in a global fashion.} $C_v=S_{\ket{v}}$, and remove $v$ from $\Pi$.
        \end{enumerate}
    \item While there exists a qubit $v$ appearing only in clauses of size at least $ k'\geq 3$:
        \begin{enumerate}
            \item Set $\ket{v}=\ket{0}$ and $C_v=S_{\ket{v}}$. Remove $v$ from $\Pi$.
        \end{enumerate}
    \item While there exists a $2$-local clause:

    \begin{enumerate}
        \item If there exists a {stacked} $2$-local clause $c$, i.e. $c'\neq c$ such that $Q_c\subseteq Q_{c'}$:
    \begin{enumerate}
        \item If $Q_c= Q_{c'}$, remove the qubits $c$ acts on, and set their values to satisfy $c$ and $c'$.
        \item Else, $Q_c\subset Q_{c'}$. Thus, $c'$ is $k'$-local for $3\leq k'\leq k$. Set the values of the qubits in $Q_c$ so as to satisfy $c$. This collapses $c'$ to a $(k'-2)$-local constraint on qubits $Q_{c'}\setminus Q_c$.
            \begin{enumerate}
                \item  If $k'-2=1$, then $c'$ has been collapsed to a $1$-local constraint on some vertex $v\in Q_{c'}\setminus Q_c$, creating a path rooted at $v$. Set $v$ so as to satisfy $c'$, and use a CR to solve the resulting path until either the path ends, or a $k''$-local constraint is hit for $3\leq k''\leq k'$. In the latter case (Figure~\ref{fig:path}, Left), the $k''$-local constraint is reduced to a $(k''-1)$-local constraint and we return to the beginning of Step 3.
            \end{enumerate}
        \end{enumerate}
        \item Else, pick an arbitrary $2$-local clause $c$ acting on variables $v_1$ and $v_2$. Then, $v_1$ ($v_2$) is the start of a path $h_1$ ($h_2$) (e.g., Figure~\ref{fig:path}, Middle).
            \begin{enumerate}
                \item If the path forms a cycle from $v_1$ to $v_2$, use the cycle matrix to solve the cycle. Remove the corresponding qubits and clauses from $\Pi$.
                \item Else, set $v_1$ and $v_2$ so as to satisfy $c$. Solve the resulting paths $h_1$ ($h_2$) until a $k'$-local ($k''$-local) constraint $l_1$ ($l_2$) is hit for $3\leq k'\leq k$ ($3\leq k''\leq k$). If both $l_1$ and $l_2$ are found:
                \begin{enumerate}
                    \item If $l_1=l_2$ (i.e. $k'=k''$) and $k'-2=1$, then fuse the paths $h_1$ and $h_2$ into a new path beginning at the qubit in $l_1$ which is not in $h_1$ or $h_2$ (e.g., Figure~\ref{fig:path}, Right). Iteratively solve the resulting path until a $k'$-local constraint is hit for $3\leq k'\leq k$.
                \end{enumerate}
            \end{enumerate}
    \end{enumerate}
    \item If any qubits are unassigned, set their values to $\ket{0}$.
\end{enumerate}
\begin{figure}[t]
	\begin{center}
\resizebox{400pt}{!}{%
\begin{tikzpicture}

    \coordinate (v1) at (-4,0);
    \coordinate (v2) at (-2,0) {};
    \coordinate (v3) at (0,0);
    \coordinate (v4) at (2,0) {};
    \coordinate (v5) at (4,0) {};

    \foreach \v in {v1,v2,v3,v4,v5}{
      \node [circle, minimum size=0.4cm, line width=0pt] (\v') at (\v) {};
    }

    \filldraw [draw=black, fill=cyan, opacity=0.2, line width=0.5pt]
		plot[smooth, tension=1] coordinates { (v1'.120) (-3,0.5) (v2'.60)} arc (60:-60:0.2)--plot[smooth, tension=1] coordinates { (v2'.-60) (-3,-0.5) (v1'.-120)} arc(-120:-240:0.2);
    \filldraw [draw=black, fill=red, opacity=0.2, line width=0.5pt]
		plot[smooth, tension=1] coordinates { (v2'.120) (-1,0.5) (v3'.60)} arc (60:-60:0.2)--plot[smooth, tension=1] coordinates { (v3'.-60) (-1,-0.5) (v2'.-120)} arc(-120:-240:0.2);
    \filldraw [draw=black, fill=green, opacity=0.2, line width=0.5pt]
		plot[smooth, tension=1] coordinates { (v3'.120) (2,0.5) (v5'.60)} arc (60:-60:0.2)--plot[smooth, tension=1] coordinates { (v5'.-60) (2,-0.5) (v3'.-120)} arc(-120:-240:0.2);
  	
    \foreach \l in {1,...,5}{
      \filldraw [black] (v\l) circle (2pt) node [inner sep=5pt, label=below:$v_{\l}$] {};
    }
\end{tikzpicture}
\hspace{10mm}
\begin{tikzpicture}

    \coordinate (v1) at (-4,0);
    \coordinate (v2) at (-2,0) {};
    \coordinate (v3) at (0,0);
    \coordinate (v4) at (2,0) {};

    \foreach \v in {v1,v2,v3,v4}{
      \node [circle, minimum size=0.4cm, line width=0pt] (\v') at (\v) {};
    }

    \filldraw [draw=black, fill=cyan, opacity=0.2, line width=0.5pt]
		plot[smooth, tension=1] coordinates { (v1'.120) (-3,0.5) (v2'.60)} arc (60:-60:0.2)--plot[smooth, tension=1] coordinates { (v2'.-60) (-3,-0.5) (v1'.-120)} arc(-120:-240:0.2);
    \filldraw [draw=black, fill=red, opacity=0.2, line width=0.5pt]
		plot[smooth, tension=1] coordinates { (v2'.120) (-1,0.5) (v3'.60)} arc (60:-60:0.2)--plot[smooth, tension=1] coordinates { (v3'.-60) (-1,-0.5) (v2'.-120)} arc(-120:-240:0.2);
    \filldraw [draw=black, fill=green, opacity=0.2, line width=0.5pt]
		plot[smooth, tension=1] coordinates { (v3'.120) (1,0.5) (v4'.60)} arc (60:-60:0.2)--plot[smooth, tension=1] coordinates { (v4'.-60) (1,-0.5) (v3'.-120)} arc(-120:-240:0.2);
  	
    \foreach \l in {1,...,4}{
      \filldraw [black] (v\l) circle (2pt) node [inner sep=5pt, label=below:$v_{\l}$] {};
    }
\end{tikzpicture}

\hspace{10mm}
\begin{tikzpicture}

	\coordinate (v0) at (0,0) {};
    \coordinate (v1) at (-2,1){};
    \coordinate (v2) at (0,1) {};
    \coordinate (v3) at (2,1){};
    \coordinate (v4) at (4,0) {};
    \coordinate (v5) at (2,-1) {};
    \coordinate (v6) at (0,-1) {};
    \coordinate (v7) at (-2,-1) {};

    \foreach \v in {v1,v2,v3,v4,v5,v6,v7}{
      \node [circle, minimum size=0.4cm, line width=0pt] (\v') at (\v) {};
    }

    \filldraw [draw=black, fill=cyan, opacity=0.2, line width=0.5pt]
		plot[smooth, tension=1] coordinates { (v1'.120) (-1,1.5) (v2'.60)} arc (60:-60:0.2)--plot[smooth, tension=1] coordinates { (v2'.-60) (-1,0.5) (v1'.-120)} arc(-120:-240:0.2);
    \filldraw [draw=black, fill=red, opacity=0.2, line width=0.5pt]
		plot[smooth, tension=1] coordinates { (v2'.120) (1,1.5) (v3'.60)} arc (60:-60:0.2)--plot[smooth, tension=1] coordinates { (v3'.-60) (1,0.5) (v2'.-120)} arc(-120:-240:0.2);
    \filldraw [draw=black, fill=blue, opacity=0.2, line width=0.5pt]
		plot[smooth, tension=1] coordinates { (v7'.120) (-1,-0.5) (v6'.60)} arc (60:-60:0.2)--plot[smooth, tension=1] coordinates { (v6'.-60) (-1,-1.5) (v7'.-120)} arc(-120:-240:0.2);
  	\filldraw [draw=black, fill=yellow, opacity=0.2, line width=0.5pt]
		plot[smooth, tension=1] coordinates { (v6'.120) (1,-0.5) (v5'.60)} arc (60:-60:0.2)--plot[smooth, tension=1] coordinates { (v5'.-60) (1,-1.5) (v6'.-120)} arc(-120:-240:0.2);
    \filldraw [draw=black, fill=green, opacity=0.2]
		(v3'.45) -- (v4'.45) arc (45:-45:0.2)
		-- (v5'.-45) arc (-45:-180:0.2)
		-- (v3'.180) arc (180:45:0.2) -- cycle;
    \foreach \l in {1,...,7}{
      \filldraw [black] (v\l) circle (2pt) node [inner sep=5pt, label=below:$v_{\l}$] {};
    }
\end{tikzpicture}
}
\end{center}

\caption{(Left) Solving the path rooted at $v_1$ via CR allows us to satisfy clauses $(v_1,v_2)$ and $(v_2,v_3)$, and since $v_3$ receives an assignment in this process, the clause $(v_3,v_4,v_5)$ is projected onto a $2$-local residual clause on $(v_4,v_5)$. The CR then stops. (Middle) Letting $c$ denote the clause on $(v_2,v_3)$, $v_2$ is the start of a path $(v_2,v_1,\ldots)$ to the left, and $v_3$ is the start of a path $(v_3,v_4,\ldots)$ to the right. (Right) Inducing CRs on $v_1$ and $v_7$, we eventually assign values to $v_3$ and $v_5$. This collapses the $3$-local clause on $(v_3,v_4,v_5)$ into a $1$-local clause on $v_4$ with a unique satisfying assignment, which in turn induces a new CR starting at $v_4$. Thus, the two CR's are ``fused'' into one CR through the $3$-local clause.}
\label{fig:path}
\end{figure}
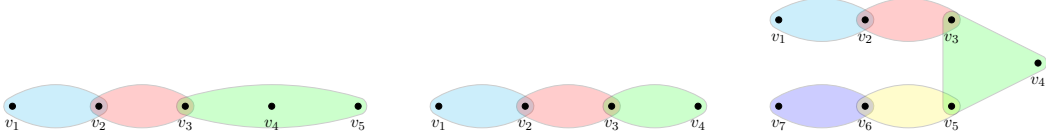

\noindent An illustration of algorithm A on an example input is given after this proof (see also Figure~\ref{fig:example}). It is clear that the algorithm runs in polynomial time. We now prove correctness.\\

\noindent\textbf{Correctness.} In Step $1$, any qubit $v$ acted on by a $1$-local constraint $c$ has only one possible satisfying assignment $\ket{c^\perp}$ (up to phase). Thus, if there are conflicting $1$-local clauses acting on $v$, we must reject; else, we must set $v$ to $\ket{c^\perp}$, and subsequently simplify all remaining clauses acting on $v$ (i.e. map $C_v$ to $S_{\ket{v}}$). Note that since $v$ is removed at this point, we can never have a conflict on it in a future iteration of the algorithm.

If we reach Step 2, we claim that $\Pi$ is satisfiable. To show this, note first that each time Step $3$ is run, if there exists a $2$-local clause, then at least one new clause is satisfied and subsequently removed from $\Pi$. Thus, in order to prove that $\Pi$ is satisfiable, it suffices to show that the following loop invariant holds.\\

\noindent\emph{Loop Invariant:} Before each execution of Step $3$, either $\Pi$ contains some $2$-local clause and no $1$-local clauses, or $\Pi$ contains no clauses. \\

\noindent Two notes are in order here. First, the invariant's constraint on the absence of $1$-local clauses is necessary, as otherwise we are not guaranteed that paths can be iteratively solved as in Step 3(a)(ii)(A). Second, the invariant implies that once the algorithm reaches Step $4$, all clauses must have been satisfied. Thus, all unused qubits at that point can be set arbitrarily.

We now show that the loop invariant holds throughout the execution of the algorithm. First, since after Step $1$, $\Pi$ contained no $1$-local contraints, and since in Step $2$ we only reduce $k'$-local clauses to $k''$-local ones for $k''\geq 2$, it follows that the invariant holds the first time Step $3$ is run.

We next show that if the invariant holds just before Step $3$ is run, then it also holds after Step $3$ is run. We divide the analysis into cases depending on which line of Step 3 is run.
\begin{itemize}
    \item (Step 3(a)(i)) Since each qubit appears in at most $2$ clauses by assumption, it holds that $c$ and $c'$ must be disjoint from all other clauses in $\Pi$. Informally, the joint satisfiability of $c$ and $c'$ is hence independent of the satisfiability of the rest of the instance. In the language of our loop invariant, the removal of $c$ and $c'$ does not create any $1$-local constraints. Moreover, if subsequently $\Pi$ is non-empty, then it is not the case that all remaining clauses are $k'$-local for $3\leq k'\leq k$. This is because otherwise the set of clauses disjoint from $c$ and $c'$ must contain a variable involved only in clauses of size at least $k'\geq 3$, contradicting Step 2. Thus, the loop invariant holds. (Aside: If $c$ and $c'$ act on vertices $u$ and $v$, a satisfying assignment to $c$ and $c'$ can be computed by thinking of $(u,v)$ and $(v,u)$ as a cycle, and subsequently solving the cycle using the cycle matrix technique.)
    \item (Step 3(a)(ii)) Once $c$ and the variables it acts on are removed, we have a $(k'-2)$-local constraint on $c'$. Thus, if $k'-2=1$, this induces a CR on the incident path of $2$-local constraints on $c'$ (this path may have length $0$). Note that this path can never loop back to $c'$, nor can it contain a cycle, as otherwise a variable would appear in more than $2$ clauses. Thus, either the path consists solely of $2$-local constraints, which are all satisfied via the CR (recall we assume that any broken CR's are continued automatically by assignment $\ket{0}$ to the next vertex in the path to induce a new CR), or eventually we hit a $k''$-local constraint for $k''\geq 3$, which we now collapse to a $(k''-2)$-local constraint; call the latter $\phi$. Thus, since the CR removes any possible $1$-local constraint on $c'$ and creates no further $1$-local constraints, no $1$-local constraints exist after Step 3.

        Also, on the existence of a $2$-local clause after Step $3$ (assuming clauses remain): If the path above consisted solely of $2$-local constraints, then the CR removed a connected component of the interaction graph, in which case the remaining connected components must contain a $2$-local clause (otherwise, we again contradict Step 2). If the path encountered a $k''$-local clause for $k''=3$, on the other hand, the CR itself created a $2$-local clause. Similarly, if $k''\geq 4$, then $\phi$ is at least $3$-local, and we claim that at least one of the vertices of $\phi$ must be incident on a $2$-local edge. Indeed, otherwise we again contradict Step 2. Thus, the loop invariant holds.

    \item (Step 3(b)(i)) Since any variable occurs in at most $2$ clauses, the cycle must be disjoint from all other clauses in $\Pi$. Thus, its removal does not create a $1$-local constraint. That there must exist a $2$-local constraint if $\Pi$ is non-empty now follows by similar arguments as in Steps 3(a)(i) and 3(a)(ii).

    \item (Step 3(b)(ii)) There are $2$ possible scenarios: Both $h_1$ and $h_2$ are disjoint paths, or they intersect on a $k'$-local clause. The first case is analogous to Step 3(a)(ii)'s analysis. In the second case, suppose $h_1$ and $h_2$ intersect on clause $c$. Then, since the paths do not form a cycle (otherwise, we would be in Step 3(b)(i)), $c$ must be $k'$-local for $k'\geq 3$. If $k'-2=1$, let $v$ denote the qubit in $c$ not acted on by $h_1$ or $h_2$. Then, solving the paths $h_1$ and $h_2$ collapses $c$ into a $1$-local constraint on $v$. But this, in turn, creates a new path rooted at $v$, whose analysis follows from Step 3(a)(ii) above. If $k'-2> 1$, the CR stops once $c$ is collapsed to a new clause of size at least $2$. That a $2$-local clause still exists at this point (assuming the remaining instance is non-empty) follows by the argument of Step 3(a)(ii).
\end{itemize}
This concludes the proof.
\end{proof}

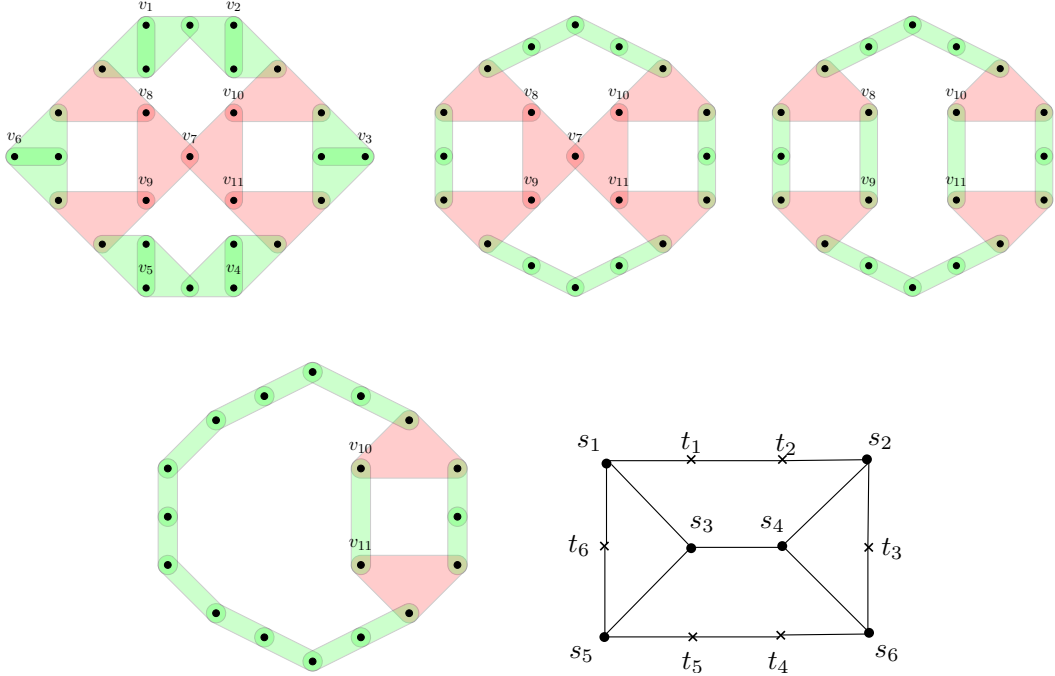
\begin{figure}[t]
	\begin{center}
\resizebox{400pt}{!}{%
\begin{tikzpicture}

    \coordinate (v1) at (-1,3);
    \coordinate (v2) at (1,3) {};
    \coordinate (v3) at (4,0);
    \coordinate (v4) at (1,-3) {};
    \coordinate (v5) at (-1,-3) {};
    \coordinate (v6) at (-4,0) {};
    \coordinate (v7) at (0,0) {};
    \coordinate (v8) at (-1,1) {};
    \coordinate (v9) at (-1,-1) {};
    \coordinate (v14) at (0,3) {};
    \coordinate (v15) at (1,2) {};
    \coordinate (v12) at (2,2) {};
    \coordinate (v13) at (3,1) {};
    \coordinate (v10) at (1,1) {};
    \coordinate (v11) at (1,-1) {};
    \coordinate (v16) at (3,-1) {};
    \coordinate (v17) at (3,0) {};
    \coordinate (v18) at (2,-2) {};
    \coordinate (v19) at (1,-2) {};
    \coordinate (v20) at (0,-3) {};
    \coordinate (v21) at (-1,-2) {};
    \coordinate (v22) at (-2,-2) {};
    \coordinate (v23) at (-3,-1) {};
    \coordinate (v24) at (-3,0) {};
    \coordinate (v25) at (-3,1) {};
    \coordinate (v26) at (-2,2) {};
    \coordinate (v27) at (-1,2) {};

    \foreach \v in {v1,v2,v3,v4,v5,v6,v7,v8,v9,v10,v11,v12,v13,v14,v15,v16,v17,v18,v19,v20,v21,v22,v23,v24,v25,v26,v27}{
      \node [circle, minimum size=0.4cm, line width=0pt] (\v') at (\v) {};
    }

   \filldraw [draw=black, fill=green, opacity=0.2]
		(v1'.90) -- (v14'.90) arc (90:-45:0.2cm)
		-- (v27'.-45) arc (-45:-180:0.2)
		-- (v1'.180) arc (180:90:0.2) -- cycle;
   \filldraw [draw=black, fill=green, opacity=0.2]
		(v14'.90) -- (v2'.90) arc (90:0:0.2cm)
		-- (v15'.0) arc (0:-135:0.2)
		-- (v14'.220) arc (220:90:0.2) -- cycle;
   \filldraw [draw=black, fill=green, opacity=0.2]
		(v2'.180) -- (v15'.180) arc (180:270:0.2cm)
		-- (v12'.-90) arc (-90:45:0.2)
		-- (v2'.45) arc (45:180:0.2) -- cycle;
    \filldraw [draw=black, fill=green, opacity=0.2]
		(v1'.0) -- (v27'.0) arc (0:-90:0.2cm)
		-- (v26'.-90) arc (-90:-220:0.2)
		-- (v1'.130) arc (130:0:0.2) -- cycle;
   \filldraw [draw=black, fill=red, opacity=0.2]
		(v12'.135) -- (v10'.135) arc (135:270:0.2cm)
		-- (v13'.-90) arc (-90:45:0.2)
		-- (v12'.45) arc (45:135:0.2) -- cycle;
   \filldraw [draw=black, fill=red, opacity=0.2]
		(v26'.135) -- (v25'.135) arc (135:270:0.2cm)
		-- (v8'.-90) arc (-90:45:0.2)
		-- (v26'.45) arc (45:135:0.2) -- cycle;
   \filldraw [draw=black, fill=red, opacity=0.2]
		(v11'.90) -- (v16'.90) arc (90:-60:0.2cm)
		-- (v18'.-60) arc (-60:-140:0.2)
		-- (v11'.220) arc (220:90:0.2) -- cycle;
   \filldraw [draw=black, fill=red, opacity=0.2]
		(v23'.90) -- (v9'.90) arc (90:-60:0.2cm)
		-- (v22'.-60) arc (-60:-140:0.2)
		-- (v23'.220) arc (220:90:0.2) -- cycle;
   \filldraw [draw=black, fill=red, opacity=0.2]
		(v10'.0) -- (v11'.0) arc (0:-120:0.2cm)
		-- (v7'.220) arc (220:140:0.2)
		-- (v10'.120) arc (120:0:0.2) -- cycle;
    \filldraw [draw=black, fill=red, opacity=0.2]
		(v8'.180) -- (v9'.180) arc (180:300:0.2cm)
		-- (v7'.-45) arc (-45:45:0.2)
		-- (v8'.45) arc (45:180:0.2) -- cycle;

   \filldraw [draw=black, fill=green, opacity=0.2]
		(v13'.180) -- (v17'.180) arc (180:270:0.2cm)
		-- (v3'.-90) arc (-90:45:0.2)
		-- (v13'.45) arc (45:180:0.2) -- cycle;
   \filldraw [draw=black, fill=green, opacity=0.2]
		(v17'.90) -- (v3'.90) arc (90:-45:0.2cm)
		-- (v16'.-45) arc (-45:-180:0.2)
		-- (v17'.180) arc (180:90:0.2) -- cycle;
  \filldraw [draw=black, fill=green, opacity=0.2]
		(v6'.90) -- (v24'.90) arc (90:0:0.2cm)
		-- (v23'.0) arc (0:-135:0.2)
		-- (v6'.220) arc (220:90:0.2) -- cycle;
   \filldraw [draw=black, fill=green, opacity=0.2]
		(v25'.0) -- (v24'.0) arc (0:-90:0.2cm)
		-- (v6'.-90) arc (-90:-220:0.2)
		-- (v25'.130) arc (130:0:0.2) -- cycle;
  \filldraw [draw=black, fill=green, opacity=0.2]
		(v19'.90) -- (v18'.90) arc (90:-45:0.2cm)
		-- (v4'.-45) arc (-45:-180:0.2)
		-- (v19'.180) arc (180:90:0.2) -- cycle;
  \filldraw [draw=black, fill=green, opacity=0.2]
		(v19'.0) -- (v4'.0) arc (0:-90:0.2cm)
		-- (v20'.-90) arc (-90:-220:0.2)
		-- (v19'.130) arc (130:0:0.2) -- cycle;
  \filldraw [draw=black, fill=green, opacity=0.2]
		(v21'.180) -- (v5'.180) arc (180:270:0.2cm)
		-- (v20'.-90) arc (-90:45:0.2)
		-- (v21'.45) arc (45:180:0.2) -- cycle;
 \filldraw [draw=black, fill=green, opacity=0.2]
		(v22'.90) -- (v21'.90) arc (90:0:0.2cm)
		-- (v5'.0) arc (0:-135:0.2)
		-- (v22'.220) arc (220:90:0.2) -- cycle;

  \foreach \l in {1,...,11}{
      \filldraw [black] (v\l) circle (2pt) node [inner sep=5pt, label=above:$v_{\l}$] {};
    }
    \foreach \l in {12,...,27}{
      \filldraw [black] (v\l) circle (2pt) node [inner sep=5pt] {};
    }
\end{tikzpicture}

\hspace{10mm}
\begin{tikzpicture}

    \coordinate (v1) at (0,3);
    \coordinate (v2) at (1,2.5) {};
    \coordinate (v3) at (2,2);
    \coordinate (v10) at (1,1) {};
    \coordinate (v5) at (3,1) {};
    \coordinate (v6) at (3,0) {};
    \coordinate (v7) at (0,0) {};
    \coordinate (v8) at (-1,1) {};
    \coordinate (v9) at (-1,-1) {};
    \coordinate (v4) at (3,-1) {};
    \coordinate (v11) at (1,-1) {};
    \coordinate (v12) at (2,-2) {};
    \coordinate (v13) at (1,-2.5) {};
    \coordinate (v14) at (0,-3) {};
    \coordinate (v15) at (-1,-2.5) {};
    \coordinate (v16) at (-2,-2) {};
    \coordinate (v17) at (-3,-1) {};
    \coordinate (v18) at (-3,0) {};
    \coordinate (v19) at (-3,1) {};
    \coordinate (v20) at (-2,2) {};
    \coordinate (v21) at (-1,2.5) {};

    \foreach \v in {v1,v2,v3,v4,v5,v6,v7,v8,v9,v10,v11,v12,v13,v14,v15,v16,v17,v18,v19,v20,v21}{
      \node [circle, minimum size=0.4cm, line width=0pt] (\v') at (\v) {};
    }

   \filldraw [draw=black, fill=red, opacity=0.2]
		(v3'.135) -- (v10'.135) arc (135:270:0.2cm)
		-- (v5'.-90) arc (-90:45:0.2)
		-- (v3'.45) arc (45:135:0.2) -- cycle;
   \filldraw [draw=black, fill=red, opacity=0.2]
		(v20'.135) -- (v19'.135) arc (135:270:0.2cm)
		-- (v8'.-90) arc (-90:45:0.2)
		-- (v20'.45) arc (45:135:0.2) -- cycle;
   \filldraw [draw=black, fill=red, opacity=0.2]
		(v11'.90) -- (v4'.90) arc (90:-60:0.2cm)
		-- (v12'.-60) arc (-60:-140:0.2)
		-- (v11'.220) arc (220:90:0.2) -- cycle;
   \filldraw [draw=black, fill=red, opacity=0.2]
		(v17'.90) -- (v9'.90) arc (90:-60:0.2cm)
		-- (v16'.-60) arc (-60:-140:0.2)
		-- (v17'.220) arc (220:90:0.2) -- cycle;
   \filldraw [draw=black, fill=red, opacity=0.2]
		(v10'.0) -- (v11'.0) arc (0:-120:0.2cm)
		-- (v7'.220) arc (220:140:0.2)
		-- (v10'.120) arc (120:0:0.2) -- cycle;
    \filldraw [draw=black, fill=red, opacity=0.2]
		(v8'.180) -- (v9'.180) arc (180:300:0.2cm)
		-- (v7'.-45) arc (-45:45:0.2)
		-- (v8'.45) arc (45:180:0.2) -- cycle;

    \filldraw [draw=black, fill=green, opacity=0.2]
		(v1'.60) -- (v2'.60) arc (60:-120:0.2)
		-- (v1'.-120) arc (-120:-300:0.2)-- cycle;
     \filldraw [draw=black, fill=green, opacity=0.2]
		(v2'.60) -- (v3'.60) arc (60:-120:0.2)
		-- (v2'.-120) arc (-120:-300:0.2)-- cycle;
	\filldraw [draw=black, fill=green, opacity=0.2]
		(v16'.60) -- (v15'.60) arc (60:-120:0.2)
		-- (v16'.-120) arc (-120:-300:0.2)-- cycle;
     \filldraw [draw=black, fill=green, opacity=0.2]
		(v15'.60) -- (v14'.60) arc (60:-120:0.2)
		-- (v15'.-120) arc (-120:-300:0.2)-- cycle;
	\filldraw [draw=black, fill=green, opacity=0.2]
		(v1'.120) -- (v21'.120) arc (120:300:0.2)
		-- (v1'.-60) arc (-60:120:0.2)-- cycle;
     \filldraw [draw=black, fill=green, opacity=0.2]
		(v21'.120) -- (v20'.120) arc (120:300:0.2)
		-- (v21'.-60) arc (-60:120:0.2)-- cycle;
	\filldraw [draw=black, fill=green, opacity=0.2]
		(v12'.120) -- (v13'.120) arc (120:300:0.2)
		-- (v12'.-60) arc (-60:120:0.2)-- cycle;
     \filldraw [draw=black, fill=green, opacity=0.2]
		(v13'.120) -- (v14'.120) arc (120:300:0.2)
		-- (v13'.-60) arc (-60:120:0.2)-- cycle;
	\filldraw [draw=black, fill=green, opacity=0.2]
		(v5'.0) -- (v6'.0) arc (0:-180:0.2)
		-- (v5'.180) arc (180:0:0.2)-- cycle;
	\filldraw [draw=black, fill=green, opacity=0.2]
		(v6'.0) -- (v4'.0) arc (0:-180:0.2)
		-- (v6'.180) arc (180:0:0.2)-- cycle;
	\filldraw [draw=black, fill=green, opacity=0.2]
		(v19'.0) -- (v18'.0) arc (0:-180:0.2)
		-- (v19'.180) arc (180:0:0.2)-- cycle;
	\filldraw [draw=black, fill=green, opacity=0.2]
		(v18'.0) -- (v17'.0) arc (0:-180:0.2)
		-- (v18'.180) arc (180:0:0.2)-- cycle;	
  \foreach \l in {7,...,11}{
      \filldraw [black] (v\l) circle (2pt) node [inner sep=5pt, label=above:$v_{\l}$] {};
    }
    \foreach \l in {1,...,6}{
      \filldraw [black] (v\l) circle (2pt) node [inner sep=5pt] {};
    }
    \foreach \l in {12,...,21}{
      \filldraw [black] (v\l) circle (2pt) node [inner sep=5pt] {};
    }
\end{tikzpicture}
\hspace{10mm}
\begin{tikzpicture}

    \coordinate (v1) at (0,3);
    \coordinate (v2) at (1,2.5) {};
    \coordinate (v3) at (2,2);
    \coordinate (v10) at (1,1) {};
    \coordinate (v5) at (3,1) {};
    \coordinate (v6) at (3,0) {};
    \coordinate (v7) at (0,0) {};
    \coordinate (v8) at (-1,1) {};
    \coordinate (v9) at (-1,-1) {};
    \coordinate (v4) at (3,-1) {};
    \coordinate (v11) at (1,-1) {};
    \coordinate (v12) at (2,-2) {};
    \coordinate (v13) at (1,-2.5) {};
    \coordinate (v14) at (0,-3) {};
    \coordinate (v15) at (-1,-2.5) {};
    \coordinate (v16) at (-2,-2) {};
    \coordinate (v17) at (-3,-1) {};
    \coordinate (v18) at (-3,0) {};
    \coordinate (v19) at (-3,1) {};
    \coordinate (v20) at (-2,2) {};
    \coordinate (v21) at (-1,2.5) {};

    \foreach \v in {v1,v2,v3,v4,v5,v6,v7,v8,v9,v10,v11,v12,v13,v14,v15,v16,v17,v18,v19,v20,v21}{
      \node [circle, minimum size=0.4cm, line width=0pt] (\v') at (\v) {};
    }

   \filldraw [draw=black, fill=red, opacity=0.2]
		(v3'.135) -- (v10'.135) arc (135:270:0.2cm)
		-- (v5'.-90) arc (-90:45:0.2)
		-- (v3'.45) arc (45:135:0.2) -- cycle;
   \filldraw [draw=black, fill=red, opacity=0.2]
		(v20'.135) -- (v19'.135) arc (135:270:0.2cm)
		-- (v8'.-90) arc (-90:45:0.2)
		-- (v20'.45) arc (45:135:0.2) -- cycle;
   \filldraw [draw=black, fill=red, opacity=0.2]
		(v11'.90) -- (v4'.90) arc (90:-60:0.2cm)
		-- (v12'.-60) arc (-60:-140:0.2)
		-- (v11'.220) arc (220:90:0.2) -- cycle;
   \filldraw [draw=black, fill=red, opacity=0.2]
		(v17'.90) -- (v9'.90) arc (90:-60:0.2cm)
		-- (v16'.-60) arc (-60:-140:0.2)
		-- (v17'.220) arc (220:90:0.2) -- cycle;

       \filldraw [draw=black, fill=green, opacity=0.2]
		(v1'.60) -- (v2'.60) arc (60:-120:0.2)
		-- (v1'.-120) arc (-120:-300:0.2)-- cycle;
     \filldraw [draw=black, fill=green, opacity=0.2]
		(v2'.60) -- (v3'.60) arc (60:-120:0.2)
		-- (v2'.-120) arc (-120:-300:0.2)-- cycle;
	\filldraw [draw=black, fill=green, opacity=0.2]
		(v16'.60) -- (v15'.60) arc (60:-120:0.2)
		-- (v16'.-120) arc (-120:-300:0.2)-- cycle;
     \filldraw [draw=black, fill=green, opacity=0.2]
		(v15'.60) -- (v14'.60) arc (60:-120:0.2)
		-- (v15'.-120) arc (-120:-300:0.2)-- cycle;
	\filldraw [draw=black, fill=green, opacity=0.2]
		(v1'.120) -- (v21'.120) arc (120:300:0.2)
		-- (v1'.-60) arc (-60:120:0.2)-- cycle;
     \filldraw [draw=black, fill=green, opacity=0.2]
		(v21'.120) -- (v20'.120) arc (120:300:0.2)
		-- (v21'.-60) arc (-60:120:0.2)-- cycle;
	\filldraw [draw=black, fill=green, opacity=0.2]
		(v12'.120) -- (v13'.120) arc (120:300:0.2)
		-- (v12'.-60) arc (-60:120:0.2)-- cycle;
     \filldraw [draw=black, fill=green, opacity=0.2]
		(v13'.120) -- (v14'.120) arc (120:300:0.2)
		-- (v13'.-60) arc (-60:120:0.2)-- cycle;
	\filldraw [draw=black, fill=green, opacity=0.2]
		(v5'.0) -- (v6'.0) arc (0:-180:0.2)
		-- (v5'.180) arc (180:0:0.2)-- cycle;
	\filldraw [draw=black, fill=green, opacity=0.2]
		(v6'.0) -- (v4'.0) arc (0:-180:0.2)
		-- (v6'.180) arc (180:0:0.2)-- cycle;
	\filldraw [draw=black, fill=green, opacity=0.2]
		(v19'.0) -- (v18'.0) arc (0:-180:0.2)
		-- (v19'.180) arc (180:0:0.2)-- cycle;
	\filldraw [draw=black, fill=green, opacity=0.2]
		(v18'.0) -- (v17'.0) arc (0:-180:0.2)
		-- (v18'.180) arc (180:0:0.2)-- cycle;
	\filldraw [draw=black, fill=green, opacity=0.2]
		(v8'.0) -- (v9'.0) arc (0:-180:0.2)
		-- (v8'.180) arc (180:0:0.2)-- cycle;
	\filldraw [draw=black, fill=green, opacity=0.2]
		(v10'.0) -- (v11'.0) arc (0:-180:0.2)
		-- (v10'.180) arc (180:0:0.2)-- cycle;
     \foreach \l in {8,...,11}{
      \filldraw [black] (v\l) circle (2pt) node [inner sep=5pt, label=above:$v_{\l}$] {};
    }
    \foreach \l in {1,...,6}{
      \filldraw [black] (v\l) circle (2pt) node [inner sep=5pt] {};
    }
    \foreach \l in {12,...,21}{
      \filldraw [black] (v\l) circle (2pt) node [inner sep=5pt] {};
    }
\end{tikzpicture}
}\\
\vspace{8mm}
\resizebox{120pt}{!}{%

\begin{tikzpicture}

    \coordinate (v1) at (0,3);
    \coordinate (v2) at (1,2.5) {};
    \coordinate (v3) at (2,2);
    \coordinate (v10) at (1,1) {};
    \coordinate (v5) at (3,1) {};
    \coordinate (v6) at (3,0) {};
    \coordinate (v7) at (0,0) {};
    \coordinate (v8) at (-1,1) {};
    \coordinate (v9) at (-1,-1) {};
    \coordinate (v4) at (3,-1) {};
    \coordinate (v11) at (1,-1) {};
    \coordinate (v12) at (2,-2) {};
    \coordinate (v13) at (1,-2.5) {};
    \coordinate (v14) at (0,-3) {};
    \coordinate (v15) at (-1,-2.5) {};
    \coordinate (v16) at (-2,-2) {};
    \coordinate (v17) at (-3,-1) {};
    \coordinate (v18) at (-3,0) {};
    \coordinate (v19) at (-3,1) {};
    \coordinate (v20) at (-2,2) {};
    \coordinate (v21) at (-1,2.5) {};

    \foreach \v in {v1,v2,v3,v4,v5,v6,v7,v8,v9,v10,v11,v12,v13,v14,v15,v16,v17,v18,v19,v20,v21}{
      \node [circle, minimum size=0.4cm, line width=0pt] (\v') at (\v) {};
    }

   \filldraw [draw=black, fill=red, opacity=0.2]
		(v3'.135) -- (v10'.135) arc (135:270:0.2cm)
		-- (v5'.-90) arc (-90:45:0.2)
		-- (v3'.45) arc (45:135:0.2) -- cycle;

   \filldraw [draw=black, fill=red, opacity=0.2]
		(v11'.90) -- (v4'.90) arc (90:-60:0.2cm)
		-- (v12'.-60) arc (-60:-140:0.2)
		-- (v11'.220) arc (220:90:0.2) -- cycle;
          \filldraw [draw=black, fill=green, opacity=0.2]
		(v1'.60) -- (v2'.60) arc (60:-120:0.2)
		-- (v1'.-120) arc (-120:-300:0.2)-- cycle;
     \filldraw [draw=black, fill=green, opacity=0.2]
		(v2'.60) -- (v3'.60) arc (60:-120:0.2)
		-- (v2'.-120) arc (-120:-300:0.2)-- cycle;
	\filldraw [draw=black, fill=green, opacity=0.2]
		(v16'.60) -- (v15'.60) arc (60:-120:0.2)
		-- (v16'.-120) arc (-120:-300:0.2)-- cycle;
     \filldraw [draw=black, fill=green, opacity=0.2]
		(v15'.60) -- (v14'.60) arc (60:-120:0.2)
		-- (v15'.-120) arc (-120:-300:0.2)-- cycle;
	\filldraw [draw=black, fill=green, opacity=0.2]
		(v17'.60) -- (v16'.60) arc (60:-120:0.2)
		-- (v17'.-120) arc (-120:-300:0.2)-- cycle;
	\filldraw [draw=black, fill=green, opacity=0.2]
		(v1'.120) -- (v21'.120) arc (120:300:0.2)
		-- (v1'.-60) arc (-60:120:0.2)-- cycle;
     \filldraw [draw=black, fill=green, opacity=0.2]
		(v21'.120) -- (v20'.120) arc (120:300:0.2)
		-- (v21'.-60) arc (-60:120:0.2)-- cycle;
    \filldraw [draw=black, fill=green, opacity=0.2]
		(v20'.120) -- (v19'.120) arc (120:300:0.2)
		-- (v20'.-60) arc (-60:120:0.2)-- cycle;
	\filldraw [draw=black, fill=green, opacity=0.2]
		(v12'.120) -- (v13'.120) arc (120:300:0.2)
		-- (v12'.-60) arc (-60:120:0.2)-- cycle;
     \filldraw [draw=black, fill=green, opacity=0.2]
		(v13'.120) -- (v14'.120) arc (120:300:0.2)
		-- (v13'.-60) arc (-60:120:0.2)-- cycle;
	\filldraw [draw=black, fill=green, opacity=0.2]
		(v5'.0) -- (v6'.0) arc (0:-180:0.2)
		-- (v5'.180) arc (180:0:0.2)-- cycle;
	\filldraw [draw=black, fill=green, opacity=0.2]
		(v6'.0) -- (v4'.0) arc (0:-180:0.2)
		-- (v6'.180) arc (180:0:0.2)-- cycle;
	\filldraw [draw=black, fill=green, opacity=0.2]
		(v19'.0) -- (v18'.0) arc (0:-180:0.2)
		-- (v19'.180) arc (180:0:0.2)-- cycle;
	\filldraw [draw=black, fill=green, opacity=0.2]
		(v18'.0) -- (v17'.0) arc (0:-180:0.2)
		-- (v18'.180) arc (180:0:0.2)-- cycle;
		\filldraw [draw=black, fill=green, opacity=0.2]
		(v10'.0) -- (v11'.0) arc (0:-180:0.2)
		-- (v10'.180) arc (180:0:0.2)-- cycle;

    \foreach \l in {1,...,6}{
      \filldraw [black] (v\l) circle (2pt) node [inner sep=5pt] {};
    }
    \foreach \l in {10,...,11}{
      \filldraw [black] (v\l) circle (2pt) node [inner sep=5pt, label=above:$v_{\l}$] {};
    }
    \foreach \l in {12,...,21}{
      \filldraw [black] (v\l) circle (2pt) node [inner sep=5pt] {};
    }
\end{tikzpicture}
}
\hspace{10mm}
\includegraphics[height=3.2cm]{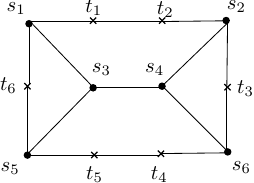}
\end{center}

\caption{(Top left) The initial $3$-QSAT instance, denoted $H(G)$ for the purposes of this example, and where each hyperedge is an arbitrary (rank $1$) constraint. (Top middle) The hypergraph obtained after applying Step 2 to vertices $v_1,v_2,v_3,v_4,v_5,v_6$. (Top right) The hypergraph obtained by applying Step 2 to vertex $v_7$. (Bottom left) The hypergraph obtained by applying Step 3bii to edge $(v_8,v_9)$. Applying Step 3bii to edge $(v_{10},v_{11})$ would thus reduce us to a $2$-QSAT cycle, which is then solved via Step 3bi. Note the order of vertices/edges processed in this example is arbitrary, and was chosen here simply for symmetry. (Bottom right) A pseudo-line graph corresponding to the hypergraph in the top left image.}
\label{fig:example}
\end{figure}
\paragraph*{An example and applicability.} Having shown Theorem~\ref{thm:bounded}, we make two remarks.
\begin{enumerate}
    \item (Demonstration) We illustrate in Figure~\ref{fig:example} how algorithm A repeatedly partially reduces parts of a $3$-QSAT instance to $2$-QSAT instances (this particular example does not require fusing of chain reactions).
    \item (Applicability) Whereas the class of hypergraphs considered in Theorem~\ref{thm:bounded} may seem \emph{a priori} rather restricted, as a ``proof of concept'' we observe that the example of Figure~\ref{fig:example} can be generalized to an entire family of non-trivial examples to which Theorem~\ref{thm:bounded} applies. (We focus on $3$-uniform hypergraphs for simplicity.) Specifically, encode the hypergraph $H(G)$ in the top left of Figure~\ref{fig:example} by a graph $G=(V_d,V_c,E)$ reminiscent of a line graph, which we call a \emph{pseudo-line graph}, as depicted in the bottom right of Figure~\ref{fig:example}. $G$ is defined as follows. Each disc vertex $s_i\in V_d$ corresponds to a red hyperedge (color online; e.g. hyperedge $\set{v_7,v_8,v_9}$ is red). Each cross vertex $t_i\in V_c$ corresponds to a pair of green hyperedges (e.g. the hyperedges containing $v_6$) which intersect on a set of size $2$. An edge in $G$ means the two hyperstructures corresponding to each end point of the edge intersect in precisely one element (since each vertex has degree $2$, this is well-defined). That each disc and cross vertex have degree $3$ and $2$, respectively, ensures the corresponding hypergraph, $H(G)$, is $2$-regular (we assume $G$ has no self-loops or parallel edges). It is not difficult to see that any such pseudo-line graph $G$ gives rise to a $2$-regular $3$-uniform hypergraph $H(G)$, thus yielding an entire family of $3$-QSAT instances to which Theorem~\ref{thm:bounded} applies. Finally, we connect this example with the notion of \emph{transfer filtrations} in Section~\ref{scn:parameterized}. Namely, the hypergraph $H(G)$ in Figure~\ref{fig:example} is of transfer type at most $12$ --- for all pairs of distinct hyperedges $e_i,e_j$ with $e_i\cap e_j=2$, simply add $e_i\cap e_j$ to the foundation.
\end{enumerate}

\subsection{Connection to $k$-QSAT instances with a SDR}\label{sscn:SDR}

Recall that the initial motivation for this study was the open question~\cite{LLMSS10} of constructing satisfying assignments to $k$-QSAT instances which have a system of distinct representatives (SDR). We now show that the class of $k$-QSAT instances solved by Theorem~\ref{thm:bounded} is closely related to this motivation.

Specifically, note that the only line of algorithm A which can reject is 1a, and this is due to the presence of $1$-local clauses. Since any $k$-QSAT instance with an SDR is satisfiable (by a product state)~\cite{LLMSS10}, it follows that hypergraphs with hyperedges of size $1$ in general cannot have SDRs. Indeed, consider the hypergraph with edge set $E=\set{\set{1,2,3},\set{1},\set{2},\set{3}}$ --- clearly, this has no SDR, and if we set the clauses to (respectively) project onto $\ket{000}$, $\ket{1}$, $\ket{1}$, and $\ket{1}$, then the corresponding $3$-QSAT instance is unsatisfiable.

However, if our $k$-QSAT instance has no $1$-local clauses in Theorem~\ref{thm:bounded}, then algorithm A always produces a (product-state) solution. We now show that this is no coincidence --- such hypergraphs must have an SDR. Thus, Theorem~\ref{thm:bounded} answers the open question of~\cite{LLMSS10} for all $k$-QSAT instances in which (1) there are no $1$-local clauses and (2) each qubit occurs in at most $2$ clauses.

\begin{reptheorem}{thm:matchgraph}
    Let $G=(V,E)$ be a hypergraph with all hyperedges of size at least $2$, and such that each vertex has degree at most $2$. Then, $G$ has an SDR.
\end{reptheorem}
\begin{proof}
    The claim is vacuously true if $G$ is empty; thus, assume $G$ is non-empty. We claim the following (poly-time) algorithm constructs an SDR.
    \begin{enumerate}
        \item While there exists a pair of distinct edges $e_i,e_j\in E$ such that $\abs{e_i\cap e_j}\geq 2$:
            \begin{enumerate}
                \item Pick an arbitrary pair of distinct vertices $v_i,v_j\in e_i\cap e_j$. Match $v_i$ to $e_i$ and $v_j$ to $e_j$, and remove $v_i,v_j,e_i,e_j$ from $G$.
            \end{enumerate}
        \item While there exist vertices $v$ of degree $1$ in some hyperedge $c_v$, match $v$ to $c_v$ and remove $v$ and $c_v$.

        \item Remove any vertices of degree $0$.

        \item Let $G'$ denote the remaining hypergraph. (We assume $G'$ is connected; if not, simply repeat the following steps on each connected component.) Construct its line graph $L$ (defined in proof of correctness below). Find a cycle $C$ in $L$, and let $v$ be an arbitrary vertex on $C$.
        \item (Build a ``reverse SDR'' on $L$) Run a depth-first search (DFS) rooted at $v$ on $L$, with the following two modifications: (1) Each time an edge $x=(u,w)$ is used to get to a previously unseen vertex $w$, add $2$-tuple $(w,x)$ to set $M\subseteq V(L)\times E(L)$. (2) When the DFS begins, do not mark the root $v$ as having been ``seen''. (In other words, $v$ should only be marked as ``seen'' once some edge $(w,v)$ is used in the DFS.)
        \item (Convert the ``reverse'' SDR on $L$ to an SDR on $G'$) For each $(w,x)\in M$:
        \begin{enumerate}
            \item Suppose $x=(u,w)$ for some $u\in V(L)$. Then, let $e_u$ and $e_w$ denote the hyperedges in $G'$ corresponding to $u$ and $w$, respectively. Match $e_w$ to the unique element of $e_u\cap e_w$.
        \end{enumerate}
    \end{enumerate}

    \paragraph*{Correctness.} We analyze each step in sequence.
    \begin{enumerate}
        \item (Step 1) Since by assumption $v_i$ and $v_j$ have degree at most $2$, they only appear in edges $e_i$ and $e_j$. Thus, we can safely match them to $e_i$ and $e_j$, respectively, and remove $v_i,v_j,e_i,e_j$.
        \item (Steps 2 and 3) Since $v$ occurs only in $c_v$, we can safely match it to $c_v$ and remove $v$ and $c_v$.

        \item (Step 4) After Steps 1, 2, and 3, $G'$ has two properties: (a) each pair of edges intersects in at most one vertex, and (b) each vertex has degree precisely $2$. By property (a), the ``usual'' (i.e. we do not need to distinguish between different intersection sizes between edges) definition of line graph can be applied to $G'$ to construct $L$: Each hyperedge of $G'$ is a vertex in $L$, and two distinct vertices $u$ and $v$ in $L$ are neighbors if and only if their corresponding hyperedges in $G'$ intersect.

            Moreover, by properties (a) and (b) of $G'$, each vertex of $L$ has degree at least $2$. It follows that $L$ contains a cycle $C$ (e.g. since summing degrees on each vertex yields there are at least as many edges in $L$ as vertices, and $L$ is connected by the assumption that $G'$ is connected).

        \item (Step 5) Since a DFS sees each vertex (at least) once, and since the first time a vertex $w$ is seen it must be that the edge $x=(u,w)$ just used to get to $w$ must have been used for the first time, it follows that no other $2$-tuple in $M$ contains $x$ or $w$. Moreover, as stated in Step 5, note that the root $v$ of the DFS is not marked until some edge $(w,v)$ is used; the latter is guaranteed to happen since (1) we chose $v$ to lie on a cycle $C$, and (2) if the two edges on the cycle $C$ which $v$ is incident on are $\set{w,v}$ and $\set{v,w'}$, then by definition of a DFS at least one of them will be traversed in direction $(w,v)$ or $(w',v)$ (as opposed to both being traversed in directions $(v,w)$ and $(v,w')$; note this property holds since we use a DFS instead of a breadth-first search). It follows that each vertex $w$ of $L$ appears in precisely one $2$-tuple of $(w,x)\in M$. We conclude that $M$ is a ``reverse SDR'' for $L$, by which we mean each \emph{vertex} of $L$ is matched with a unique \emph{edge} of $L$ (whereas our desired SDR for hypergraph $G$ is supposed to match each \emph{hyperedge} to a unique \emph{vertex}\footnote{For further clarity, the distinction is that in a reverse SDR, there may be more edges in the graph than vertices, whereas in an SDR, there may be more vertices than edges.}).

        \item (Step 6) This step converts each two-tuple in $M$ to a matched hyperedge $e_w$ and vertex in $v\in e_u\cap e_w$. Note that $\abs{e_u\cap e_w}=1$ by property (a) of $G'$. Since $M$ is a ``reverse'' SDR, it follows that the matching produced by Step 6 completes the desired SDR.
    \end{enumerate}
\end{proof}

\section{Quantum SAT and parameterized algorithms}\label{scn:parameterized}

We next develop a parameterized algorithm for computing an explicit (product state) solution to a non-trivial class of $k$-QSAT instances (Theorem~\ref{theorem}). Although the inspiration stems from algebraic geometry (AG), we generally avoid AG terminology to increase accessibility. For those versed in the topic, however, we include a brief overview of the ideas of this section in AG terms at the end of Section~\ref{sscn:generic}. The final algorithm and its runtime, along with the study of asymptotic speedups, are given in Section~\ref{sscn:runtime}. For readers who wish to see the outline of the algorithm before reading the details, see Figure~\ref{alg:statement} on Page 29.

\subsection{The transfer type of a hypergraph}\label{sscn:transfertype}

\paragraph{Transfer filtrations.} The core graph theoretic idea on which our framework will be built is the new notion of \emph{transfer type} of a hypergraph, which we now introduce.

\begin{definition}\label{def:2}
A hypergraph $G=(V,E)$ is of {\it transfer type $b$} if there exists a chain of subhypergraphs (denoted a {\it transfer filtration of type $b$})  $G_0\subseteq G_1\subseteq \cdots \subseteq G_m=G$ and an ordering of the edges $E(G)=\{E_1,\dots,E_m\}$ such that
\begin{enumerate}
\item $E(G_i)=\{E_1,\ldots,E_i\}$ for each $i\in \{0,\ldots,m\}$,
\item $|V(G_i)|\le |V(G_{i-1})|+1$ for each $i\in \{1,\ldots,m\}$,
\item if $|V(G_i)|= |V(G_{i-1})|+1$, then $V(G_i)\setminus V(G_{i-1}) \subseteq E_i$,
\item $|V(G_0)|=b$, where we call $V(G_0)$ the \emph{foundation},
\item and each edge of $G$ has at least one vertex not in $V(G_0)$.
\end{enumerate}
\end{definition}
\noindent In other words, a transfer filtration of type $b$ builds up $G$ iteratively by choosing $b$ vertices as a ``foundation'', and in each iteration adding precisely one new edge $E_i$ and \emph{at most} one new vertex. If a new vertex is added in iteration $i$, condition (3) says it must be in edge $E_i$ added in iteration $i$. Note that under certain conditions (i.e. when the foundation size is bounded by $b\le |V(G)|-|E(G)|+k-1$), we later show (Theorem~\ref{thm:general}) that this iterative construction can be exploited to prove the hypergraph $G$ must have an SDR. Also, recall that as stated in Section~\ref{scn:intro}, we conjecture that given $G$ as input, finding a foundation of minimum size is NP-hard. (In the context of our parameterized algorithms, this is reminiscent of common parameters such as treewidth being NP-hard to compute~\cite{ACP87}.) Finally, condition $5$ in Definition~\ref{def:2} is a ``minimality'' condition. For example, the entire vertex set is also a valid foundation, but contains more vertices than required to satisfy conditions $1$ through $4$. To prevent such cases, for convenience in our discussions we have added Condition $5$ to Definition~\ref{def:2}, which holds without loss of generality. In particular, given a filtration satisfying $1$-$4$ but not $5$, it is easy to compute a modified filtration with a smaller foundation satisfying $1$-$5$ (if edge $e$ fails $5$, simply remove a single arbitrary vertex $v\in e$ from $G_0$, and add an additional intermediate graph $G_i$ to the filtration to first ``pick up'' vertex $v$).

\begin{example}[Running example (Figure~\ref{fig:ex10})]\label{ex:4cycle1}
We now introduce a hypergraph $G$ which will serve as a running example to concretely illustrate the ideas of this section. Let $V(G)=\{1,2,3,4\}$ with edges $E_1=\{1,2,3\}$, $E_2=\{1,2,4\}$, $E_3=\{1,3,4\}$ and $E_4=\{2,3,4\}$. By Definition \ref{definition:cycle}, $G$ is a $3$-uniform cycle. Consider hypergraphs $G_0,G_1,G_2,G_3$ such that $V(G_0)=\{1,2\}$, $V(G_1)=\{1,2,3\}$, $V(G_2)=V(G_3)=V(G_4)=V(G)$, $E(G_0)=\emptyset$ and $E(G_j)=\{E_1,\ldots,E_j\}$ for $j=1,2,3$. Then $G_0\subseteq G_1\subseteq G_2\subseteq G_3\subseteq G_4=G$ is a transfer filtration of type $2$. In particular $G_2$ is a chain and $G_3$ is a semicycle in the sense of Definition \ref{definition:chain} and Definition \ref{definition:semicyle}, respectively.

\end{example}

\begin{figure}[t]
\begin{center}

\resizebox{500pt}{!}{%
\hspace{7mm}
\begin{tikzpicture}
	\coordinate (v1) at (-4,0);
    \coordinate (v2) at (-2,0) {};

    \foreach \v in {v1,v2}{
      \node [circle, minimum size=0.4cm, line width=0pt] (\v') at (\v) {};
    }

    \foreach \l in {1,...,2}{
      \filldraw [black] (v\l) circle (2pt) node [inner sep=5pt, label=below:$v_{\l}$] {};
    }
\end{tikzpicture}
\hspace{20mm}
\begin{tikzpicture}
	    \coordinate (v1) at (-4,0);
    \coordinate (v2) at (-2,0) {};
    \coordinate (v3) at (0,0);

    \foreach \v in {v1,v2,v3}{
      \node [circle, minimum size=0.4cm, line width=0pt] (\v') at (\v) {};
    }

    \filldraw [draw=black, fill=cyan, opacity=0.2, line width=0.5pt]
		plot[smooth, tension=1] coordinates { (v1'.120) (-2,0.5) (v3'.60)} arc (60:-60:0.2)--plot[smooth, tension=1] coordinates { (v3'.-60) (-2,-0.5) (v1'.-120)} arc(-120:-240:0.2);

    \foreach \l in {1,...,3}{
      \filldraw [black] (v\l) circle (2pt) node [inner sep=5pt, label=below:$v_{\l}$] {};
    }
\end{tikzpicture}
\hspace{20mm}
\begin{tikzpicture}
	    \coordinate (v1) at (-4,0);
    \coordinate (v2) at (-2,0) {};
    \coordinate (v3) at (0,0);
    \coordinate (v4) at (2,0) {};

    \foreach \v in {v1,v2,v3,v4}{
      \node [circle, minimum size=0.4cm, line width=0pt] (\v') at (\v) {};
    }

    \filldraw [draw=black, fill=cyan, opacity=0.2, line width=0.5pt]
		plot[smooth, tension=1] coordinates { (v1'.120) (-2,0.5) (v3'.60)} arc (60:-60:0.2)--plot[smooth, tension=1] coordinates { (v3'.-60) (-2,-0.5) (v1'.-120)} arc(-120:-240:0.2);
	\filldraw [draw=black, fill=blue, opacity=0.2]
	plot [smooth, tension=2] coordinates {(v1'.270) (3,-1.5) (v4'.0)}
	arc(0:180:0.2)
	plot [smooth, tension=2] coordinates {(v4'.180) (2,-1) (v2'.0)}
	arc(0:90:0.2) -- (v1'.90) arc(90:270:0.2)
	;

    \foreach \l in {1,...,4}{
      \filldraw [black] (v\l) circle (2pt) node [inner sep=5pt, label=below:$v_{\l}$] {};
    }
\end{tikzpicture}
}\\
\vspace{6mm}
\resizebox{430pt}{!}{%
\hspace{25mm}
\begin{tikzpicture}
	    \coordinate (v1) at (-4,0);
    \coordinate (v2) at (-2,0) {};
    \coordinate (v3) at (0,0);
    \coordinate (v4) at (2,0) {};

    \foreach \v in {v1,v2,v3,v4}{
      \node [circle, minimum size=0.4cm, line width=0pt] (\v') at (\v) {};
    }

    \filldraw [draw=black, fill=cyan, opacity=0.2, line width=0.5pt]
		plot[smooth, tension=1] coordinates { (v1'.120) (-2,0.5) (v3'.60)} arc (60:-60:0.2)--plot[smooth, tension=1] coordinates { (v3'.-60) (-2,-0.5) (v1'.-120)} arc(-120:-240:0.2);
  	\filldraw [draw=black, fill=yellow, opacity=0.2, line width=0.5pt]
		plot [smooth, tension=2] coordinates {(v1'.120) (2,1.5) (v4'.0)}
	arc(0: -90:0.2) -- (v3'.270) arc(270:120:0.2)
	plot [smooth, tension=2] coordinates {(v3'.120) (1,1) (v1'.-60)} arc(-60:-240:0.2);
	
	\filldraw [draw=black, fill=blue, opacity=0.2]
	plot [smooth, tension=2] coordinates {(v1'.270) (3,-1.5) (v4'.0)}
	arc(0:180:0.2)
	plot [smooth, tension=2] coordinates {(v4'.180) (2,-1) (v2'.0)}
	arc(0:90:0.2) -- (v1'.90) arc(90:270:0.2)
	;

    \foreach \l in {1,...,4}{
      \filldraw [black] (v\l) circle (2pt) node [inner sep=5pt, label=below:$v_{\l}$] {};
    }
\end{tikzpicture}
\hspace{5mm}
\begin{tikzpicture}
	    \coordinate (v1) at (-4,0);
    \coordinate (v2) at (-2,0) {};
    \coordinate (v3) at (0,0);
    \coordinate (v4) at (2,0) {};

    \foreach \v in {v1,v2,v3,v4}{
      \node [circle, minimum size=0.4cm, line width=0pt] (\v') at (\v) {};
    }

    \filldraw [draw=black, fill=cyan, opacity=0.2, line width=0.5pt]
		plot[smooth, tension=1] coordinates { (v1'.120) (-2,0.5) (v3'.60)} arc (60:-60:0.2)--plot[smooth, tension=1] coordinates { (v3'.-60) (-2,-0.5) (v1'.-120)} arc(-120:-240:0.2);
    \filldraw [draw=black, fill=red, opacity=0.2, line width=0.5pt]
		plot[smooth, tension=1] coordinates { (v2'.120) (0,0.5) (v4'.60)} arc (60:-60:0.2)--plot[smooth, tension=1] coordinates { (v4'.-60) (0,-0.5) (v2'.-120)} arc(-120:-240:0.2);
  	\filldraw [draw=black, fill=yellow, opacity=0.2, line width=0.5pt]
		plot [smooth, tension=2] coordinates {(v1'.120) (2,1.5) (v4'.0)}
	arc(0: -90:0.2) -- (v3'.270) arc(270:120:0.2)
	plot [smooth, tension=2] coordinates {(v3'.120) (1,1) (v1'.-60)} arc(-60:-240:0.2);
	
	\filldraw [draw=black, fill=blue, opacity=0.2]
	plot [smooth, tension=2] coordinates {(v1'.270) (3,-1.5) (v4'.0)}
	arc(0:180:0.2)
	plot [smooth, tension=2] coordinates {(v4'.180) (2,-1) (v2'.0)}
	arc(0:90:0.2) -- (v1'.90) arc(90:270:0.2)
	;

    \foreach \l in {1,...,4}{
      \filldraw [black] (v\l) circle (2pt) node [inner sep=5pt, label=below:$v_{\l}$] {};
    }
\end{tikzpicture}
}
\end{center}
\caption{The running example from Example~\ref{ex:4cycle1}. In clockwise order, beginning with the top left, are the subgraphs $G_i$ in the specified transfer filtration of Example~\ref{ex:4cycle1}: $G_0$ (foundation), $G_1$, $G_2$, $G_3$, $G_4=G$. Note that $G_2$ is a chain, $G_3$ a semicycle, and $G_4$ a cycle.}
\label{fig:ex10}
\end{figure}

\begin{rem}
Let $G$ be a $k$-uniform hypergraph of transfer type $b$. Since $G_1$ has exactly one edge, then $b\ge k-1$.
\end{rem}

\begin{example}\label{ex:L}
The {\it tight line graph} \cite{KATONA20161884} of a $k$-uniform hypergraph $G$ is the (undirected) graph $L(G)$ with vertices labelled by $E(G)$ and edges $\{(E_i,E_j)\mid E_i,E_j\in E(G)\text{ and }\abs{E_i\cap E_j}=k-1\}$. (This generalizes the notion of a line graph; a yet more general definition, called an \emph{edge intersection graph}, is mentioned in Section~\ref{scn:hypergraphs}, which adds an edge $(E_i,E_j)$ so long as $E_i\cap E_j\neq\emptyset$.) A $k$-uniform hypergraph $G$ is {\it line graph connected} if its tight line graph is connected. Arguing as in the proof of Lemma 1 of \cite{KATONA20161884}, it is easy to show that every $k$-uniform line graph connected hypergraph with no isolated vertices has transfer type $k-1$. (An \emph{isolated} vertex is not contained in any hyperedge.) To see this, pick an edge $E_{i_1}$, set $G_1$ to be the hypergraph induced by $E_{i_1}$, and choose $G_0$ to be any subset of $G_1$ of cardinality $k-1$. Since the vertex of $L(G)$ corresponding to $E_{i_1}$ is not isolated, there exists another edge $E_{i_2}$ such that $|E_{i_1}\cap E_{i_2}|=k-1$. Let $G_2$ be the hypergraph with vertices $E_{i_1}\cup E_{i_2}$ and edges $E_{i_1}$, $E_{i_2}$. Since the induced subgraph with vertices $E_{i_1}$ and $E_{i_2}$ is not a connected component of $L(G)$ this procedure can be iterated until the required transfer filtration is constructed.
\end{example}

\begin{example}[Running example]\label{ex:4cycle2}
Let $G$ be the $3$-uniform cycle of Example \ref{ex:4cycle1}. Then $L(G)$ is a cycle of length $4$, in the usual graph-theoretical sense. In particular $G$ is line graph connected.
\end{example}

\begin{example}\label{ex:foundations}
Let $G$ be a $k$-uniform hypergraph with no isolated vertex. Suppose that $L(G)$ has $c$ connected components and apply the construction of Example \ref{ex:L} to each connected component. Taking the union of the corresponding $G_0$ sets (one for each connected component) and adding the edges one at the time, we obtain a transfer filtration of type at most $c(k-1)$ on $G$. More generally, isolated vertices can be taken into account by adding them to the $0$-th term of the filtration. Hence every $k$-uniform hypergraph has transfer type at most $c(k-1)+i$ where $c$ is the number of connected components of $L(G)$ and $i$ is the number of isolated vertices of $G$.
\end{example}

\begin{example}\label{ex:torus}
Let $a_1,\ldots,a_{k-1}\ge 2$ be integers. Consider the hypergraph $G$ with vertices $(\mathbb Z/a_1\mathbb Z)\times \cdots \times (\mathbb Z/a_{k-1}\mathbb Z)$ and edges of the form
\[
\{(i_{1},\ldots,i_{k-1}),(i_1+1,\ldots, i_{k-1}),(i_1,i_2+1,\ldots,i_{k-1}),\ldots,(i_1,\ldots,i_{k-1}+1)\}
\]
for each $(i_{1},\ldots,i_{k-1})\in V(G)$. Let $G_0$ be the subhypergraph of $G$ with vertices $(\mathbb Z/a_1\mathbb Z)\times \cdots \times (\mathbb Z/a_{k-2}\mathbb Z)\times \{0\}$ and no edges. Let $G_1$ be the subhypergraph with vertices $V(G_0)\cup\{(0,\ldots,0,1)\}$ and a single edge $\{(0,\ldots,0),(1,0,\ldots,0),\ldots,(0,\ldots,0,1)\}$. Let $G_2$ be the subhypergraph obtained from  $G_1$ by adding the vertex $(0,\ldots,0,1,1)$ and the edge
\[
\{(0,\ldots,0,1,0),(1,0,\ldots,0,1,0),\ldots,(0,\ldots,0,2,0),(0,\ldots,0,1,1)\}\,.
\]
Proceeding in this way, say by reverse lexicographic order, one constructs all the vertices with last coordinate equal to $1$ and all the edges containing $k-1$ vertices with last coordinate equal to $0$ and exactly one vertex with last coordinate equal to $1$. The corresponding chain of subhypergraphs can be labeled as $G_0\subseteq G_1\subseteq \cdots\subseteq G_{a_{1}\cdots a_{k-2}}$. Iterating this construction on can extend this chain to $G_0\subseteq G_1,\ldots,\subseteq G_{2a_{1}\cdots a_{k-2}}$ in such a way that all the vertices whose last coordinate is $0$, $1$, or $2$ are accounted. Further iterations lead to a transfer filtration of type $|G_0|=a_1\cdots a_{k-2}$. Clearly there is nothing special about the particular choice of $G_0$ for instance we could have picked $V(G_0)=\{0\}\times(\mathbb Z/a_2\mathbb Z)\times \cdots \times (\mathbb Z/a_{k-1}\mathbb Z)$ in which case an obvious adaptation of the construction above leads to a transfer filtration of type $a_2\cdots a_{k-2}$. In particular, this shows that in general a hypergraph $G$ can be of transfer type $b$ for different choices of $b$.
\end{example}

\begin{rem}\label{rem:r}
Let $G$ be a hypergraph admitting a transfer filtration $G_0\subseteq G_1\subseteq \cdots\subseteq G_m=G$ of type $b$. Assume the edges of $G$ are ordered in such a way that $E(G_i)=\{E_1,\ldots,E_i\}$ for each $i\in \{1,\ldots,m\}$. Since by construction each edge contains at least one vertex not in $V(G_0)$, there exists a function $r:\{1,\ldots,m\}\to \{0,\ldots,m-1\}$ such that $r(i)<i$ and $|E_i\setminus V(G_{r(i)})|=1$ for all $i\in \{1,\ldots,m\}$.
\end{rem}

\begin{example}[Running example]\label{ex:4cycle3}
Let $G$ be the $3$-uniform cycle of Example \ref{ex:4cycle1}. Then $r:\{1,2,3,4\}\to \{0,1,2,3\}$ can be chosen to be such that $r(1)=r(2)=0$, $r(3)=1$ and $r(4)=1$.
\end{example}

\paragraph{The Decoupling Lemma.} As the first step in our construction, we show how to map any $k$-uniform hypergraph $G$ of transfer type $b$ to a new $k$-uniform hypergraph $G'$ of transfer type $b$ whose transfer filtration \emph{must} add a vertex in each step (this follows directly from the relationship between $\abs{V(G)}$ and $\abs{E(G)}$ below). This has two effects worth noting: First, ${G'}$ is guaranteed to have an SDR. Second, it \emph{decouples} certain intersections in the hypergraph, as illustrated in Figure~\ref{fig:decouple}. For clarity, in the lemma below, for a function $p$ acting on vertices, we implicitly extend its action to edges in the natural way, i.e. if $e=(v_1,v_2,v_3)$ then $p(e)=(p(v_1),p(v_2),p(v_3))$.

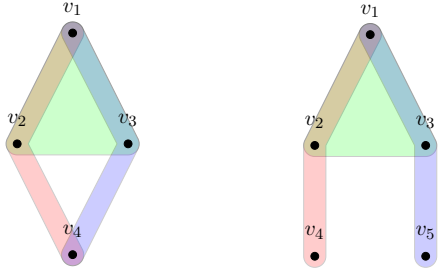
\begin{figure}[t]
	\begin{center}
\resizebox{170pt}{!}{%
\begin{tikzpicture}

	\coordinate (v0) at (0,0) {};
    \coordinate (v1) at (0,2){};
    \coordinate (v2) at (-1,0) {};
    \coordinate (v3) at (1,0){};
    \coordinate (v4) at (0,-2) {};
    \foreach \v in {v1,v2,v3,v4,v4}{
      \node [circle, minimum size=0.4cm, line width=0pt] (\v') at (\v) {};
    }
    \filldraw [draw=black, fill=green, opacity=0.2]
		(v3'.30) -- (v1'.30) arc (30:150:0.2cm)
		-- (v2'.150) arc (150:270:0.2)
		-- (v3'.270) arc (-90:30:0.2) -- cycle;
     \filldraw [draw=black, fill=red, opacity=0.2]
		(v1'.150) -- (v2'.150) arc (150:210:0.2)
		-- (v4'.-150) arc (-150:30:0.2)
		-- (v2'.0) -- (v1'.-30 )arc (-30:150:0.2) -- cycle;
     \filldraw [draw=black, fill=blue, opacity=0.2]
		(v1'.210) -- (v3'.180)--(v4'.-210) arc(-210:-30:0.2)
		-- (v3'.-30) arc(-30:30:0.2)-- (v1'.30)arc (30:210:0.2) -- cycle;
    \foreach \l in {1,...,4}{
      \filldraw [black] (v\l) circle (2pt) node [inner sep=5pt, label=above:$v_{\l}$] {};
    }
\end{tikzpicture}
\hspace{25mm}
\begin{tikzpicture}

	\coordinate (v0) at (0,0) {};
    \coordinate (v1) at (0,2){};
    \coordinate (v2) at (-1,0) {};
    \coordinate (v3) at (1,0){};
    \coordinate (v4) at (-1,-2) {};
    \coordinate (v5) at (1,-2) {};
    \foreach \v in {v1,v2,v3,v4,v4,v5}{
      \node [circle, minimum size=0.4cm, line width=0pt] (\v') at (\v) {};
    }
    \filldraw [draw=black, fill=green, opacity=0.2]
		(v3'.30) -- (v1'.30) arc (30:150:0.2cm)
		-- (v2'.150) arc (150:270:0.2)
		-- (v3'.270) arc (-90:30:0.2) -- cycle;
     \filldraw [draw=black, fill=red, opacity=0.2]
		(v1'.150) -- (v2'.150) arc (150:180:0.2)
		-- (v4'.180) arc (180:360:0.2)
		-- (v2'.0) -- (v1'.-30 )arc (-30:150:0.2) -- cycle;
     \filldraw [draw=black, fill=blue, opacity=0.2]
		(v1'.210) -- (v3'.180)--(v5'.180) arc(180:360:0.2)
		-- (v3'.0) arc(0:30:0.2)-- (v1'.30)arc (30:210:0.2) -- cycle;
    \foreach \l in {1,...,5}{
      \filldraw [black] (v\l) circle (2pt) node [inner sep=5pt, label=above:$v_{\l}$] {};
    }
\end{tikzpicture}
}
\end{center}
\caption{For the hypergraph on the left, consider the transfer filtration in which $G_0$ contains vertices $v_1,v_2$, and we iteratively add edges $\set{v_1,v_2,v_3}$, $\set{v_1,v_2,v_4}$, and $\set{v_1,v_3,v_4}$ to the filtration. Then, the Decoupling Lemma (Lemma~\ref{lem:p}) maps the hypergraph on the left to the hypergraph on the right, in the process decoupling the intersection on vertex $v_4$. The surjective function $p$ ``undoes'' the decoupling by mapping $v_1,v_2,v_3$ to themselves, and $v_4,v_5$ to $v_4$.}
\label{fig:decouple}
\end{figure}

\begin{lemma}[Decoupling lemma]\label{lem:p}
Given a $k$-uniform hypergraph $G$ of transfer type $b$, there exists a $k$-uniform hypergraph $\widetilde G$ of transfer type $b$ with $|E(G)|+b$ vertices and a surjective function $p:V(\tilde G)\to V(G)$ such $p(\widetilde E)\in E(G)$ for every $\widetilde E\in E(\widetilde G)$.
\end{lemma}
\begin{proof} Let $G_0\subseteq G_1\subseteq \cdots \subseteq G_m=G$ be a transfer filtration such that $|V(G_0)|=b$ and let $r:\{1,\ldots,m\}\to \{0,\ldots,m-1\}$ as in Remark \ref{rem:r}. Suppose the vertices of $G$ are labeled as $\{1,\ldots,n\}$ in such a way that $V(G_0)=\{1,\ldots,b\}$. Furthermore, we may assume that the edges $\{E_1,\ldots,E_m\}$ of $G$ are labeled in such a way that $E(G_i)=\{E_1,\ldots,E_i\}$ for each $i\in \{1,\ldots,m\}$. Let $V(\widetilde G)=\{1,\ldots,m+b\}$ and define $p:\{1,\ldots,m+b\}\to \{1,\ldots, n\}$ so that $p(i)=i$ for all $i\in\{1,\ldots,b\}$ and $\{p(i)\}=E_{i-b}\setminus V(G_{r(i-b)})$ for all $i\in \{b+1,\ldots,b+m\}$ (where recall $r$ is from Remark~\ref{rem:r}). By Remark \ref{rem:r}, $p$ is surjective. For each $j\in \{1,\ldots,m+b\}$, let $\underline{j}=\min(p^{-1}(p(j)))$. If we define $E(\widetilde G)=\{\widetilde E_1,\ldots,\widetilde E_m\}$ by setting $\widetilde E_i=\set{i+b}\cup\{\underline{j}\mid j\in p^{-1}(E_i\setminus p(i+b))\}$ for each $i\in \{1,\ldots,m\}$, then $\widetilde G$ is clearly $k$-uniform. Let $\widetilde G_0$ be the hypergraph with vertices $\{1,\ldots,b\}$ and no edges. For each $i\in \{1,\ldots,m\}$, let $\widetilde G_i$ be the subgraph of $\widetilde G$ with vertices $\{1,\ldots,i+b\}$ and edges $\{\widetilde E_1,\ldots,\widetilde E_i\}$. Then $\widetilde G_0\subseteq \widetilde G_1\subseteq \cdots \subseteq \widetilde G_m=\widetilde G$ is a transfer filtration and $\widetilde G$ is a hypergraph of transfer type $b$.
\end{proof}

\begin{figure}[t]
\begin{center}

\resizebox{280pt}{!}{%
\begin{tikzpicture}
	\coordinate (v1) at (-4,0);
    \coordinate (v2) at (-2,0) {};
    \coordinate (v3) at (0,0);
    \coordinate (v5) at (2,0) {};
    \coordinate (v4) at (4,0) {};
    \coordinate (v6) at (6,0) {};

    \foreach \v in {v1,v2,v3,v4,v5,v6}{
      \node [circle, minimum size=0.4cm, line width=0pt] (\v') at (\v) {};
    }

    \filldraw [draw=black, fill=cyan, opacity=0.2, line width=0.5pt]
		plot[smooth, tension=1] coordinates { (v1'.120) (-2,0.5) (v3'.60)} arc (60:-60:0.2)--plot[smooth, tension=1] coordinates { (v3'.-60) (-2,-0.5) (v1'.-120)} arc(-120:-240:0.2);
	
	\filldraw [draw=black, fill=yellow, opacity=0.2]
	plot [smooth, tension=2] coordinates {(v2'.270) (3,-1.5) (v6'.0)}
	arc(0:180:0.2)
	plot [smooth, tension=2] coordinates {(v6'.180) (2,-1) (v3'.0)}
	arc(0:90:0.2) -- (v2'.90) arc(90:270:0.2)
	;

  	\filldraw [draw=black, fill=red, opacity=0.2, line width=0.5pt]
		plot [smooth, tension=2] coordinates {(v1'.120) (2,1.5) (v5'.0)}
	arc(0: -90:0.2) -- (v3'.270) arc(270:120:0.2)
	plot [smooth, tension=2] coordinates {(v3'.120) (1,1) (v1'.-60)} arc(-60:-240:0.2);

	\filldraw [draw=black, fill=blue, opacity=0.2]
	plot [smooth, tension=2] coordinates {(v1'.270) (3,-1.5) (v4'.0)}
	arc(0:180:0.2)
	plot [smooth, tension=2] coordinates {(v4'.180) (2,-1) (v2'.0)}
	arc(0:90:0.2) -- (v1'.90) arc(90:270:0.2)
	;

    \foreach \l in {1,...,6}{
      \filldraw [black] (v\l) circle (2pt) node [inner sep=5pt, label=below:$v_{\l}$] {};
    }
\end{tikzpicture}
}

\end{center}
\caption{The ``decoupled'' hypergraph from the running example of Example~\ref{ex:4cycle4}. Note the order of listing $v_4$ and $v_5$ has been swapped for visual simplicity.}
\label{fig:ex19}
\end{figure}

\begin{example}[Running example (Figure~\ref{fig:ex19})]\label{ex:4cycle4}
Let $G$ be the $3$-uniform cycle of Example \ref{ex:4cycle1}. The proof of Lemma \ref{lem:p} (appendix) produces a $3$-uniform hypergraph $\widetilde G$ with vertices $\{1,2,3,4,5,6\}$ and edges $\widetilde E_1=\{1,2,3\}$, $\widetilde E_2=\{1,2,4\}$, $\widetilde E_3=\{1,3,5\}$, $\widetilde E_4=\{2,3,6\}$, and surjective function $p:\{1,2,3,4,5,6\}\to \{1,2,3,4\}$ defined by $p(1)=1$, $p(2)=2$, $p(3)=3$, $p(4)=p(5)=p(6)=4$. This choice is not unique: setting $\widetilde E_4=\{2,4,6\}$ and $p(6)=3$ also satisfies Lemma \ref{lem:p}.
\end{example}

\begin{rem}
If $k=2$, then $G$ (as defined in Lemma~\ref{lem:p}) is an ordinary graph and the construction of $p$ can be understood topologically in terms of the geometric realization $|G|$ of the standard simplicial set defined by $G$. Assume for simplicity that $G$ is connected. Then the universal cover $\widetilde{|G|}$ is the geometric realization of a (possibly infinite) tree. Moreover, the geometric realization of $p$ is the restriction of the canonical covering map $\widetilde{|G|}\to |G|$ to the closure of a fundamental domain for the canonical action of $\pi_1(|G|)$ on $\widetilde{|G|}$.
\end{rem}

\paragraph{Radius of transfer filtration.} One of the ``parameters'' in our parameterized approach will be the \emph{radius} of a transfer filtration, defined next. The concept is reminiscent of radii of graphs, and roughly measures ``how far'' an edge is from the foundation of $b$ vertices with respect to the filtration.

\begin{definition}[Radius of transfer filtration]
Let $G$ be a hypergraph admitting a transfer filtration $G_0\subseteq \cdots \subseteq G_m=G$ of type $b$. Consider the function (whose existence is guaranteed by Remark \ref{rem:r}) $r:\{0,\ldots,m\}\to \{0,\ldots,m-1\}$ such that $r(0)=0$ and $r(i)$ is the smallest integer such that $|E_i\setminus V(G_{r(i)})|= 1$ for all $i\in\{1,\ldots,m\}$. The {\it radius of the transfer filtration $G_0\subseteq \cdots \subseteq G_m=G$ of type $b$} is the smallest integer $\beta$ such that $r^\beta(i)=0$ $\forall i\in \{1,\ldots,m\}$ ($r^\beta$ denotes the composition of $r$ with itself $\beta$ times). The {\it type $b$ radius of $G$} is the minimum value $\rho(G,b)$ of $\beta$ over the set of all possible transfer filtrations of type $b$ on $G$.
\end{definition}

\begin{example}[Running example]\label{ex:4cycle5}
For $G$ the $4$-cycle from Example \ref{ex:4cycle1}, since function $r$ described in Example \ref{ex:4cycle3} is non-constant and $r(r(i))=0$ for all $i\in \{1,2,3,4\}$, then the transfer filtration of Example \ref{ex:4cycle1} has radius $\beta=2$.
\end{example}

\subsection{The main construction}\label{sscn:construction}

\smallskip\noindent
To describe the main construction, we shall use rather general terminology to describe the most general setting in which the framework applies. For those versed in quantum information, we will attempt to note what special cases for the general definitions we give are most relevant to said community; these notes will be marked with terminology \qi\ (color online).

\paragraph{Definitions.} To begin, let $W$ be a two dimensional vector space over a field $\K$. \qi\ We may set $\K=\mathbb C$ and identify $W$ with $\mathbb C^2$ if desired. We now define $k$-local functions; \qi\ these capture the notion of a $k$-local Hamiltonian constraint evaluated against a tensor product of $k$ input qubit states $\ket{v_{i_1}}\otimes\cdots\otimes\ket{v_{i_k}}$.

\begin{definition}[$k$-local, interaction graph, product satisfiability set]
A function $H_i:W^n\to \K$ is {\it $k$-local} if there exists a subset $E_i=\{i_1,\ldots,i_k\}\subseteq \{1,\ldots,n\}$ and a non-zero functional $H_i^*:W^{\otimes k}\to\K$ such that
\[
H_i(v_1,\ldots,v_n)=H_{i}^*(v_{i_1}\otimes \cdots \otimes v_{i_k})
\]
for all $v_1,\ldots,v_n\in W$, i.e. $H_i$ acts non-trivially only on a subset of $k$ indices. A collection $H=(H_1,\ldots,H_m)$ of $k$-local functions $H_1,\ldots,H_n:W^n\to \K$ is {\it $k$-local}. The corresponding subsets $\{E_1,\ldots,E_m\}$ (i.e. on which $H_1$ through $H_m$ act non-trivially, respectively) are the edges of a hypergraph $G_H$ with vertices $\{1,\ldots,n\}$ known as the {\it interaction graph of $H$}. The {\it product satisfiability set} of the $k$-local collection $H$ is the set $\S_H$ of all $(v_1,\ldots,v_n)\in W^n$ such that $v_i\neq 0$ for all $i\in \{1,\ldots,n\}$ and $H_j(v_1,\ldots,v_n)=0$ for all $j\in \{1,\ldots,m\}$.
\end{definition}

\begin{rem}[$\sharp$ isomorphism]\label{rem:7}
Consider an isomorphism  $\sharp$ between $W$ and its dual $W^\vee$ that to each $v\in W$ assigns a functional $v^\sharp\in W^\vee$ such that $v^\sharp(v)=0$. For instance, if a basis $\{w_1,w_2\}$ for $W$ is chosen then we may define $\sharp$ by setting $((a_1w_1+a_2w_2)^\sharp)(b_1w_1+b_2w_2)=a_1b_2-a_2b_1$ for all $a_1,a_2,b_1,b_2\in \K$. Moreover, given any $v_1,v_2\in W$, then $v_1^\sharp (v_2)=0$ if and only if there exists $\lambda\in \K$ such that $\lambda v_2 = v_1$. \qi\ The $\sharp$ isomorphism essentially maps a vector $\ket{\psi}\in\complex^2$ to its unique (up to scaling) orthogonal state $\ket{\psi^\perp}\in\complex^2$.
\end{rem}

\begin{definition}[Fibonacci numbers of order $N$] Let $N$ be a non-negative integer.
The {\it Fibonacci numbers of order $N$} are the entries of the sequence $(F^{(N)}_r)$ such that $F_r^{(N)}=F_{r-1}^{(N)}+\ldots+F_{r-N}^{(N)}$ for all $r\ge N$, $F_{N-1}^{(N)}=1$ and $F_{r}^{(N)}=0$ for all $r\le N-2$.
\end{definition}

\begin{rem}\label{rem:fib}
It is shown in \cite{wolfram98} that there exists a monotonically increasing sequence $(\psi_N)$ with values in the real interval $[1,2)$ such that, for each $N\ge 1$, $F_{r}^{(N)}\sim \psi_N^r$ as $r\to +\infty$.
\end{rem}

\begin{definition}[Degree]
A function $f$ on $W^l$ with values in a $\K$-vector space has {\it degree} $(d_1,\ldots,d_l)$ if $f(\lambda_1v_1,\ldots,\lambda_l v_l)=\lambda_1^{d_1}\cdots \lambda_l^{d_l}f(v_1,\ldots,v_l)$ for every $\lambda_1,\ldots,\lambda_l\in\K$ and every $v_1,\ldots,v_l\in W$.
\end{definition}

\paragraph{The Transfer Lemma.} Applying the Decoupling Lemma to an input $k$-uniform hypergraph $G$ with transfer type $b$, we obtain a $k$-uniform hypergraph $\widetilde{G}$ of type $b$ with $m=n-b$, for $m$ and $n$ the number of edges and vertices, respectively. The next lemma shows that $\widetilde{G}$ is ``nice'', in that any global (product) solution to the $k$-QSAT system can be derived from a set of assignments to the $b$ foundation vertices, and conversely, any (product) assignment to the latter can be extended to a global (product) solution. For the reader familiar with transfer matrices and CRs (see Section~\ref{scn:bounded}), with a little thought it can be seen that the latter of these claims is similar to picking an arbitrary assignment to the $b$ foundation vertices, and then iteratively applying transfer matrices to satisfy all clauses (i.e. in step $i$ of the filtration in which edge $E_i$ and vertex $v_i$ are added, the $4\times 2$ transfer matrix of edge $E_i$ is applied to the two pre-existing vertices of $E_i$ in the filtration, yielding an assignment to $v_i$ which satisfies clause $E_i$).

\begin{lemma}[Transfer Lemma]\label{prop:7}
Let $H=(H_1,\ldots,H_{n-b})$ be a $k$-local collection of functions $H_i:W^n \to \K$ whose interaction graph is a $k$-uniform hypergraph of transfer type $b$. There exist non-zero (non-constant) functions, which we call ``transfer functions'', $g_1,\ldots,g_n: W^b\to W$ with the following properties.
\begin{enumerate}
\item (Global to local assignments) If $(v_1,\ldots,v_n)\in \S_H$ (recall the $v_i$ are non-zero by definition of $\S_H$) there exist non-zero $\lambda_1,\ldots,\lambda_n\in \K$ such that, for every $i\in\{1,\ldots,n\}$,
\begin{equation}\label{eq:prop}
\lambda_iv_i=g_i(v_1,\ldots, v_b).
\end{equation}
\item (Local to global assignments) For any non-zero $v_1,\ldots,v_b\in W$ there exist $v_{b+1},\ldots,v_n\in W$ such that $(v_1,\ldots,v_n)\in \S_H$ and $v_i=g_i(v_1,\ldots,v_b)$ for every $i$ such that $g_i(v_1,\ldots,v_b)\neq 0$.
\item (Degree bounds) $g_i$ has degree $(d_{i1},\ldots,d_{ib})$ such that $d_{ij}\le F_{i}^{(b)}$ for all $j\in\{1,\ldots,b\}$.
\end{enumerate}
\qi\ Intuitively, a transfer function $g_i$ takes in a tensor product assignment to the $b$ foundation qubits, and gives a closed formula for an assignment to qubit $i$.
\end{lemma}
\begin{proof}
For $i=1,\ldots,b$, define $g_i(v_1, \ldots, v_b)=v_i$. Up to relabeling the $n$ vertices, since the number of edges is $n-b$, we may assume that $G_H$ admits a filtration $G_0\subseteq G_1\subseteq\cdots\subseteq G_{n-b}=G_H$ such that each $G_i$ is of transfer of type $b$ and $V(G_i)\setminus V(G_{i-1})=\{i+b\}$ for all $i\in\{1,\ldots,n-b\}$. Therefore, we may work by induction on $n$ (for each fixed $b$) and assume that $g_{b+1},\ldots, g_{n-1}$ have been constructed. Suppose that $E_{n-b}=\{n,i_1,\ldots,i_{k-1}\}$ (where we shall assume vertex $n$ is added in step $n-b$ of the filtration) and consider the function $g_n^\sharp:W^b\to W^\vee$ such that
\begin{equation}\label{eq:sharp}
(g_n^\sharp(v_1,\ldots, v_b))(v) = H^*_{n-b}(g_{i_1}(v_1,\ldots,v_b)\otimes \cdots \otimes g_{i_{k-1}}(v_1, \ldots, v_b)\otimes v)
\end{equation}
for all $v_1,\ldots,v_b,v\in W$. \qi\ Intuitively, $g_n^\sharp$ evaluates the local Hamiltonian term (i.e. estimates its energy penalty) corresponding to $H^*_{n-b}$ on the tensor product assignment prescribed by transfer functions $g_i$ on the first $k-1$ indices and argument $v$ on the $k$-th index. For (1) of the claim, by induction, we have that $(v_1,\ldots,v_n)\in \S_H$ implies that $g_1,\ldots,g_{n-1}$ as \eqref{eq:prop} exist and $g_n^\sharp(v_1,\ldots,v_b)(v_n)=0$. Given an isomorphism $\sharp$ between $W$ and $W^\vee$ as in Remark \ref{rem:7}, we define a function $g_n:W^b\to W$ by setting $(g_n(v_1,\ldots,v_b))^\sharp=g_n^\sharp(v_1,\ldots,v_b)$ for all $v_1,\ldots,v_b\in W$. Then $(v_1,\ldots,v_n)$ is in $\S_H$ implies \eqref{eq:prop} holds for every $i\in \{1,\ldots,n\}$. This proves (1).

To prove (2), suppose $v_{b+1},\ldots,v_{n-1}$ have been constructed. Let $v$ be such that $H^*_{n-b}(v_{i_1}\otimes\cdots\otimes v_{i_{k-1}}\otimes v)=0$ (such a $v$ exists, for example, since recall each QSAT constraint is rank $1$ in our definition). Then $(v_1,\ldots,v_{n-1}, \lambda v)\in \S_H$ for any $\lambda\in \K$. Moreover, if $g_n(v_1,\ldots,v_b)\neq 0$ then by Remark \ref{rem:7} it has to equal $\mu v$ for some non-zero $\mu \in \K$. Setting $v_n =\mu v$, proves (2).

To prove (3), we observe that the degree of $g_n$ equals the degree of $g_n^\sharp$. Using induction and \eqref{eq:sharp} it is easy to see that the latter satisfies the claimed bounds.
\end{proof}

\begin{example}[Running example]\label{ex:4cycle6}
Let $H=(H_1,H_2,H_3,H_4)$ be a $3$-local collection of functions $H_i:W^6\to \mathbb K$ whose interaction graph is the $3$-uniform chain $\widetilde G$ described in Example \ref{ex:4cycle4} (obtained by plugging the $4$-cycle $G$ of Example~\ref{ex:4cycle1} into the Decoupling Lemma). For clarity, $H_i$ is defined on hyperedge $\widetilde{E}_i$, where the order of vertices in each edge is fixed by the transfer filtration chosen; in particular, we use the natural ordering $\widetilde E_1=(1,2,3)$, $\widetilde E_2=(1,2,4)$, $\widetilde E_3=(1,3,5)$, $\widetilde E_4=(2,4,6)$, so that the foundation is $\set{1,2}$. Then the proof of Lemma \ref{prop:7} constructs transfer functions $g_1,\ldots,g_6:W^2\to W$ which give assignments to qubits $1$ through $6$, respectively, as follows. Fixing a basis $\{w_1,w_2\}$ of $W$ and unraveling \eqref{eq:sharp} we obtain
\begin{align*}
g_1(v_1,v_2)&=v_1\,;\\
g_2(v_1,v_2)&=v_2\,;\\
g_3(v_1,v_2)&= H_1^*(v_1\otimes v_2\otimes w_2)w_1-H_1^*(v_1\otimes v_2\otimes w_1)w_2\,;\\
g_4(v_1,v_2)&= H_2^*(v_1\otimes v_2\otimes w_2)w_1-H_2^*(v_1\otimes v_2\otimes w_1)w_2\,;\\
g_5(v_1,v_2)&= H_3^*(v_1\otimes g_3(v_1,v_2)\otimes w_2)w_1-H_3^*(v_1\otimes g_3(v_1,v_2)\otimes w_1)w_2\,;\\
g_6(v_1,v_2)&= H_4^*(v_2\otimes g_4(v_1,v_2)\otimes w_2)w_1-H_4^*(v_2\otimes g_4(v_1,v_2)\otimes w_1)w_2\,.
\end{align*}
In particular the matrix of degrees $d_{ij}$ is
\[
\left(
  \begin{array}{cccccc}
    1 & 0 & 1 & 1 & 2 & 1 \\
    0 & 1 & 1 & 1 & 1 & 2
  \end{array}
\right)^T
\]
so that, in accordance to the bound given in Lemma \ref{prop:7}, the $i$-th entry of each column is less or equal than the $i$-th (ordinary) Fibonacci number $F_i^{(2)}$.
\end{example}

\paragraph{The Qualifier Lemma.} Thus far, we have seen how combining the Decoupling and Transfer Lemmas ``blows up'' an input $k$-QSAT system $\Pi$ to a larger ``decoupled'' system $\Pi^+$ which is easier to solve due to its decoupled property. Now we wish to relate the solutions of $\Pi^+$ \emph{back} to $\Pi$. This is accomplished by the next lemma, which introduces a set of ``qualifier'' constraints $\set{h_s}$ with the key property (follows from Lemma~\ref{l:qualifier} and Remark~\ref{rem:after}): Any solution to $\set{h_s}$ can be extended to one for $\Pi^+$, and then \emph{mapped back} to a solution for $\Pi$. Importantly, the qualifier constraints \emph{act only} on the $b$ foundation vertices, as opposed to all $n$ vertices!

%
%
\begin{lemma}[Qualifier Lemma]\label{l:qualifier}
Let $H=(H_1,\ldots,H_m)$ be a $k$-local collection of functions $H_i:W^n\to \K$ whose interaction graph is a $k$-uniform hypergraph of transfer type $b$ such that $m>n-b$. Then there exist non-zero (non-constant) functions, called \emph{qualifiers}, $h_1,\ldots, h_{m-n+b}:W^b\to \K$ and $\pi:W^n\to W^b$ such that
\begin{enumerate}
\item $h_s(\pi(\S_H))=0$ for all $s\in \{1,\ldots,m-n+b\}$;
\item $h_s$ has degree $(d_{s1},\ldots,d_{sb})$ with $d_{sr}\le 2 F^{(b)}_{\rho(G,b)+b+1}$ for all $s\in\{1,\ldots,m+b\}$ and all $r\in\{1,\ldots,b\}$.
\end{enumerate}
\qi\ Condition (1) essentially says that any satisfying product state to the $k$-QSAT instance $H$ will, after application of an appropriate map $\pi$ from $n$ qubits down to the $b$ foundation qubits, yield a common root to all qualifier functions $h_s$.
\end{lemma}

\begin{proof}
We begin by applying the Decoupling Lemma (Lemma \ref{lem:p}) to $G_H$. This yields the larger hypergraph $\widetilde{G_H}$ and surjection $p:V(\widetilde{G_H})\to V(G_H)$, constructed via a transfer filtration $G_0\subseteq \cdots \subseteq G_m=G$ of type $b$. Note that $\widetilde{G_H}$ is the interaction graph of a $k$-local collection $\widetilde H=(\widetilde H_1,\ldots,\widetilde H_m)$ of functions $\widetilde H_i:W^{m+b}\to \K$ such that $\widetilde H^*_i=H^*_i$ for each $i\in \{1,\ldots,m\}$. Let $\Delta:W^n\to W^{m+b}$ be such that $\Delta(v_1,\ldots,v_n)=(\widetilde v_1,\ldots,\widetilde v_{m+b})$, where $\widetilde v_i = v_{p(i)}$ for all $i\in \{1,\ldots,m+b\}$. In other words, $\Delta$ ``blows up'' any assignment on the original $n$ vertices to one on $m+b$ vertices (where by assumption $m+b>n$), such that any vertex $j$ of $\widetilde{G_H}$ obtained by ``decoupling" a vertex $i$ of $G_H$ (formally, $p(j)=i$) is assigned the same vector as $i$. It follows that $(v_1, \ldots ,v_n)\in \S_H$ if and only if $\Delta(v_1,\ldots,v_n)\in \S_{\widetilde H}$.

Having applied the Decoupling Lemma and obtained a hypergraph $\widetilde{G_H}$ satisfying $m=n-b$, we now apply the Transfer Lemma (Lemma \ref{prop:7}) to $\widetilde{G_H}$. This yields that there exist $g_1,\dots,g_{m+b}:W^b\to W$ with the property that $\Delta(v_1,\ldots,v_n)\in \S_{\widetilde H}$ implies that $g_i(v_1,\ldots,v_b)$ is a multiple of $v_{p(i)}$ for all $i\in \{1,\ldots,m+b\}$. Borrowing notation from the proof of Lemma \ref{lem:p}, let $\{i_1,\ldots,i_{m-n+b}\}$ be the subset of all  $i\in\{1,\ldots,m+b\}$  such that $\underline i < i$ (intuitively, these indices correspond to vertices of $\widetilde{G_H}$ obtained by decoupling a vertex of $G_H$). For each $s\in \{1,\ldots, m-n+b\}$, define qualifier $h_s:W^b\to \K$ be such that
\begin{equation}\label{eq:h_s}
h_s(v_1,\ldots,v_b)=(g_{i_s}^\sharp(v_1,\ldots,v_b))(g_{\underline{i_s}}(v_1,\ldots,v_b))
\end{equation}
for all $v_1,\ldots,v_b\in W$. If $(v_1,\ldots,v_n)\in \S_H$, then for every $s\in \{1,\ldots,m-n+b\}$ there exists $\lambda_{i_s}, \lambda_{\underline{i_s}}\in \K$ such that $\lambda_{i_s} v_{p(i_s)} = g_{i_s}(v_1,\ldots,v_b)$ and $\lambda_{\underline{i_s}} v_{p(i_s)}=g_{\underline{i_s}}(v_1,\ldots,v_b)$. Therefore, \[
h_s(v_1,\ldots,v_b)=\lambda_{i_s}\lambda_{\underline i_s}v_{p(i_s)}^\sharp(v_{p(i_s)})=0
\]
for every $s\in \{1,\ldots,m-n+b\}$. Therefore, (1) holds if we define $\pi$ to be the composition of $\Delta$ with the projection onto the first $b$ entries. It follows from Lemma \ref{prop:7} and \eqref{eq:h_s} that $h_s$ has degree $(d_{s1},\ldots,d_{sb})$ with $d_{sr}\le 2F_{i_s}^{(b)}$ for every $s\in \{1,\ldots,m-n+b\}$ and $r\in\{1,\ldots,b\}$. Finally (2) follows from Lemma \ref{prop:7}, by choosing the transfer filtration $G_0\subseteq \cdots \subseteq G_m=G$ of type $b$ to have radius $\rho(G,b)$.
\end{proof}
\begin{rem}[How to use Qualifier Lemma to extend local solutions to global solutions]\label{rem:after}
In the notation of the proof of Lemma \ref{l:qualifier}, suppose $v_1,\ldots,v_b$ is an assignment to the foundation which satisfies the qualifiers, i.e.  $h_s(v_1,\ldots,v_b)=0$ for all $s\in \{1,\ldots,m-n+b\}$, and that $g_{i_s}(v_1,\ldots,v_b)\neq 0$ for all $s\in \{1,\ldots,m-n+b\}$. Then by Remark \ref{rem:7} we have that for each $s\in\{1,\ldots,m-n+b\}$ there exist $\mu_s\in \mathbb K$ such that
\[
\mu_s g_{\underline{i_s}}(v_1,\ldots,v_b)=g_{i_s}(v_1,\ldots,v_b)\,.
\]
In other words, in a solution to the decoupled instance on $\widetilde{G_H}$, it must be the case that all decoupled copies of a vertex $v$ receive the \emph{same} assignment (up to scalars). Indeed, $g_{\underline{i_s}}(v_1,\ldots,v_b)\neq 0$, and hence by Lemma \ref{prop:7} there exist $v_{b+1},\ldots,v_n$ such that $(v_1,\ldots,v_n)\in \S_H$, i.e. a satisfying assignment for $H$ must exist. Thus, to solve the original $k$-QSAT instance $\Pi$, it suffices to: (1) apply the Decoupling Lemma to blow up the instance to a decoupled instance $\Pi^+$, (2) apply the transfer functions from the Transfer Lemma to $v_1,\ldots,v_b$ to obtain a solution on all $m+b$ vertices for $\Pi^+$, and (3) (easily) map this solution back to one on $n$ vertices for $\Pi$.
\end{rem}

\begin{example}[Running example]\label{ex:4cycle7}
Let $H=(H_1,H_2,H_3,H_4)$ be a $3$-local collection of functions $H_i:W^4\to \mathbb K$ whose interaction graph is the $3$-uniform cycle of transfer type $2$ introduced in Example \ref{ex:4cycle1}. If $p$ is chosen as in Example \ref{ex:4cycle4}, then the two qualifier functions are
\[
h_1(v_1,v_2)=(g_5^\sharp(v_1,v_2))(g_4(v_1,v_2))
\]
of degree $(3,2)$ and
\[
h_2(v_1,v_2)=(g_6^\sharp(v_1,v_2))(g_3(v_1,v_2))
\]
of degree $(2,3)$, where $g_3,g_4,g_5,g_6$ so that $d_{sr}\le 3\le 10 =2 F_{5}^{(2)}$ for each $s,r\in\{1,2\}$, in accordance to Lemma \ref{l:qualifier}. Choosing a basis $\{w_1,w_2\}$ of $W$ as in Example \ref{ex:4cycle7}, we obtain via Condition 1 of Lemma~\ref{l:qualifier} that $(v_1,v_2,v_3,v_4)\in \mathcal S_H$ implies
\begin{align*}
0=H_1^*(v_1\otimes v_2\otimes w_1)H_2^*(v_1\otimes v_2\otimes w_1) H_3^*(v_1\otimes w_2\otimes w_2)\\
+ H_1^*(v_1\otimes v_2\otimes w_2)H_2^*(v_1\otimes v_2\otimes w_2) H_3^*(v_1\otimes w_1\otimes w_1)\\
-H_1^*(v_1\otimes v_2\otimes w_1)H_2^*(v_1\otimes v_2\otimes w_2) H_3^*(v_1\otimes w_2\otimes w_1)\\
-H_1^*(v_1\otimes v_2\otimes w_2)H_2^*(v_1\otimes v_2\otimes w_1) H_3^*(v_1\otimes w_1\otimes w_2)
\end{align*}
and
\begin{align*}
0=H_1^*(v_1\otimes v_2\otimes w_1)H_2^*(v_1\otimes v_2\otimes w_1) H_4^*(v_2\otimes w_2\otimes w_2)\\
+ H_1^*(v_1\otimes v_2\otimes w_2)H_2^*(v_1\otimes v_2\otimes w_2) H_4^*(v_2\otimes w_1\otimes w_1)\\
-H_1^*(v_1\otimes v_2\otimes w_1)H_2^*(v_1\otimes v_2\otimes w_2) H_4^*(v_2\otimes w_2\otimes w_1)\\
-H_1^*(v_1\otimes v_2\otimes w_2)H_2^*(v_1\otimes v_2\otimes w_1) H_4^*(v_2\otimes w_1\otimes w_2)\,.
\end{align*}
\end{example}

\subsection{Generic constraints}\label{sscn:generic}

Remark~\ref{rem:after} outlined the high-level strategy for computing a (product-state) solution to an input $k$-QSAT system $\Pi$. For this strategy to work, however, we require an assignment to the foundation of the transfer filtration which (1) satisfies the qualifier functions from the Qualifier Lemma, and (2) causes the transfer functions $g_i$ from the Transfer Lemma to output non-zero vectors. When are (1) and (2) possible? We now answer this question affirmatively for a non-trivial class of $k$-QSAT instances, assuming constraints are chosen generically. Aside: In this section alone, the terminology of projective algebraic geometry (AG) is used briefly to define generic constraints; otherwise, the content is intended to be accessible without a background in AG.

\begin{rem}[Generic constraints]
Let $G$ be a $k$-uniform hypergraph. The set of $k$-local constraints $H$ with interaction graph $G$ is canonically identified with the projective variety $\X_G(\K)=({\mathbb P}^{2^k-1}(\K))^m$. (We remark the same variety was used in~\cite{LLMSS10}.) We say that a property holds for the {\it generic constraint with interaction graph $G$} if it holds for every $k$-local constraint on a Zariski open set of $\X_G(\K)$. In the important case $\mathbb K=\mathbb C$, this implies in particular that such a property holds for almost all choices of constraints (with respect to the natural measure on $\X_G(\C)$ induced by the Fubini-Study metric). Note that we allow different (generic) constraints on each edge (i.e. the Hamiltonian is not necessarily translation-invariant). \qi\ To give some intuition, a generic $2$-qubit constraint $\ket{\psi}\in\complex^2\otimes\complex^2$ would be any entangled $\ket{\psi}$ (since the set of two-qubit states with a Schmidt coefficient of zero forms a zero measure set). In the case of $2$-QSAT, the upshot of this is that if $\ket{\psi}$ is a constraint on qubits $i$ and $j$, then given any assignment $\ket{\phi_1}$ to qubit $i$, the transfer matrix $T_{ij}$ corresponding to $\ket{\psi}$ is full-rank, meaning $T_{ij}\ket{\phi_1}\neq 0$ (and additionally, by definition of a transfer matrix, $\bra{\psi}(I\otimes T_{ij})\ket{\phi}_i\otimes\ket{\phi}_j=0$). In other words, an assignment on qubit $i$ is guaranteed to propagate to a valid assignment on qubit $j$. Note that classical $k$-SAT constraints are not generic, as they represent a finite set in the variety of all constraints (and in particular, correspond to rank-$1$ projections onto product states). Nevertheless, as we prove in Theorem~\ref{thm:general}, the hypergraphs to which our main theorem applies (Theorem~\ref{theorem}) have an SDR; and a satisfying assignment to any classical $k$-SAT instance with an SDR is easy to compute (as mentioned in Section~\ref{scn:intro}).
\end{rem}

\begin{definition}
Let $G_0\subseteq G_1\subseteq \cdots\subseteq G$ be a transfer filtration on a $k$-uniform hypergraph $G$ whose vertices $\{1,\ldots,n\}$ are ordered in such a way that if $i\in V(G_j)\setminus V(G_{j-1})$, then $i-1\in V(G_{j-1})$. Given $i,i'\in \{1,\ldots,n\}$ we say that $i'$ is a {\it successor of $i$} (and $i$ is a {\it predecessor} of $i'$) if there exists sequences $i_1,\ldots,i_l\in V(G)$ and $E_1,\ldots,E_{l+1}\in E(G)$ such that $i\le i_1\le \cdots\le i_l\le i'$, $i\in E_1$, $i'\in E_{l+1}$ and  $i_j\in E_{j}\cap E_{j+1}$ for all $j\in \{1,\ldots,l\}$.
\end{definition}

\begin{example}\label{ex:4cycle8}
Let $G$ be as in Example \ref{ex:4cycle8}. Then $3$ is a predecessor of $4$ and a successor of both $1$ and $2$.
\end{example}

Before showing the main theorem of this section (Theorem~\ref{theorem}), we first require the following lemma, which shows that under certain conditions, the transfer functions $g_i$ of the Transfer Lemma satisfy a surjectivity criterion. \qi\ This is roughly a generalization of the idea from $2$-QSAT that when all constraints all entangled, then all transfer matrices are full rank (and hence are surjective maps).

\begin{lemma}[Surjectivity Lemma]\label{l:surjectivity}
Let $G$ be a $k$-uniform hypergraph of transfer type $b$ with $n$ vertices and $n-b$ edges. Let $H$ be a generic $k$-local constraint in $\X_G(\K)$, let $g_1,\ldots,g_n:W^b\to W$ be transfer functions as in the Transfer Lemma (Lemma~\ref{prop:7}) and choose non-zero $v_1,\ldots,v_{b-1}\in W$. For each $i$, define $\gamma_i:W\to W$ such that
\begin{equation}\label{eq:gamma}
\gamma_i(w):=g_i(v_1,\ldots,v_{b-1},w)
\end{equation}
for all $w\in W$. If $i$ is a successor of $b$, then $\gamma_i$ is surjective.
\end{lemma}
\begin{proof} Up to a permutation of $\{1,\ldots,b\}$, we may assume that $n$ is a successor of $b$. By induction we may assume the claim has been proved for hypergraphs with $n-1$ vertices (the base case $i=b$ follows since we can take $g_b$ to be the identity, see the first line of the proof of Lemma~\ref{prop:7}). Therefore, a generic choice of $H_1,\ldots,H_{n-1-b}$ leads to surjective maps $\gamma_i$ for every successor of $b$. If the edge that contains vertex $n$ is $\{i_1,\ldots,i_{k-1},n\}$, then at least one other vertex, say $i_1$, is also a successor of $b$. Since $\gamma_{i_1}$ is surjective by the induction hypothesis, we can choose $w',w''\in W$ such that $\gamma_{i_1}(w')$ and $\gamma_{i_1}(w'')$ are linearly independent. Suppose that $\gamma_n$ is not surjective. Since it has definite (i.e. well-defined) degree, this implies $\gamma_n(w')$ and $\gamma_n(w'')$ are not linearly independent. By \eqref{eq:prop} this means that $H_{n-b}^*$ vanishes at two prescribed, linearly independent vectors. Since this condition is not generic, we have a contradiction. Thus, generically $\gamma_n$ is surjective.
\end{proof}
\begin{rem}\label{rem:F}
For each $i,j$, consider the function $\Gamma_{ij}:W\to \K$ defined by setting
\[
\Gamma_{ij}(w):=(\gamma_i(w))^\sharp(\gamma_j(w))
\]
for all $x\in \mathbb K$. It is easy to see by induction that if $w',w''\in W$ are linearly independent, then $P_{ij}(x)=\Gamma_{ij}(w'+xw'')$ is a univariate polynomial in $x$ with coefficients in $\mathbb K$. Furthermore, if the constraint $H$ is chosen generically and either $i$ or $j$ is a successor of $b$ then $P_{ij}$ is not a non-zero constant.
\end{rem}
\begin{figure}[t]
\begin{center}
\resizebox{220pt}{!}{%
\begin{tikzpicture}
[inner sep=0pt, minimum size = 0pt]
	
	\coordinate (v1) at (-2,0) {};
	\coordinate (v2) at (-1,0) {};
	\coordinate (v3) at (0,0) {};
	\coordinate (v4) at (1,0) {};
	\coordinate (v5) at (2,0) {};
	\coordinate (v6) at (1.5,-1) {};
	\coordinate (v7) at (1,-2) {};
	\coordinate (v8) at (0.5,-3) {};
	\coordinate (v9) at (0,-4) {};
	\coordinate (v10) at (-0.5,-3) {};
	\coordinate (v11) at (-1,-2) {};
	\coordinate (v12) at (-1.5,-1) {};

	\coordinate (v13) at (-1.5,1) {};
	\coordinate (v14) at (-0.5,1) {};
	\coordinate (v15) at (0.5,1) {};
	\coordinate (v16) at (1.5,1) {};
	\coordinate (v17) at (1,2) {};
	\coordinate (v18) at (0,2) {};
	\coordinate (v19) at (-1,2) {};

	\coordinate (v20) at (-0.5,3) {};
	\coordinate (v21) at (0.5,3) {};

	\coordinate (v22) at (0,4) {};

	\coordinate (v23) at (-2.5,-1) {};
	\coordinate (v24) at (-3,-2) {};
	\coordinate (v25) at (-2,-2) {};
	\coordinate (v26) at (-1.5,-3) {};

	\coordinate (v27) at (-2.5,-3) {};
	\coordinate (v28) at (-3.5,-3) {};
	\coordinate (v29) at (-4,-4) {};
	\coordinate (v30) at (-3,-4) {};		
	\coordinate (v31) at (-2,-4) {};
	\coordinate (v32) at (-1,-4) {};
	
	\coordinate (v33) at (2.5,-1) {};
	\coordinate (v34) at (3,-2) {};	
	\coordinate (v35) at (2,-2) {};		
	
	\coordinate (v36) at (1.5,-3) {};		
	\coordinate (v37) at (2.5,-3) {};
	\coordinate (v38) at (3.5,-3) {};
	\coordinate (v39) at (1,-4) {};	
	\coordinate (v40) at (2,-4) {};	
	\coordinate (v41) at (3,-4) {};	
	\coordinate (v42) at (4,-4) {};

	\foreach \v in {v1,v2,v3,v4,v5,v6,v7,v8,v9,v10,v11,v12,v13,v14,v15,v16,v17,v18,v19,v20,v21,v22,v23,v24,v25,v26,v27,v28,v29,v30,v31,v32,v33,v34,v35,v36,v37,v38,v39,v40,v41,v42}{
      \node [circle, minimum size=0.4cm, line width=0pt] (\v') at (\v) {};
    }
	\filldraw [draw=black, fill=yellow, opacity=0.2]
		(v2'.30) -- (v13'.30) arc (30:150:0.2cm)
		-- (v1'.150) arc (150:270:0.2)
		-- (v2'.270) arc (-90:30:0.2) -- cycle;
	\filldraw [draw=black, fill=yellow, opacity=0.2]
		(v3'.30) -- (v14'.30) arc (30:150:0.2cm)
		-- (v2'.150) arc (150:270:0.2)
		-- (v3'.270) arc (-90:30:0.2) -- cycle;
	\filldraw [draw=black, fill=yellow, opacity=0.2]
		(v4'.30) -- (v15'.30) arc (30:150:0.2cm)
		-- (v3'.150) arc (150:270:0.2)
		-- (v4'.270) arc (-90:30:0.2) -- cycle;
	\filldraw [draw=black, fill=yellow, opacity=0.2]
		(v5'.30) -- (v16'.30) arc (30:150:0.2cm)
		-- (v4'.150) arc (150:270:0.2)
		-- (v5'.270) arc (-90:30:0.2) -- cycle;
	\filldraw [draw=black, fill=yellow, opacity=0.2]
		(v16'.30) -- (v17'.30) arc (30:150:0.2cm)
		-- (v15'.150) arc (150:270:0.2)
		-- (v16'.270) arc (-90:30:0.2) -- cycle;
		-- (v5'.270) arc (-90:30:0.2) -- cycle;
	\filldraw [draw=black, fill=yellow, opacity=0.2]
		(v15'.30) -- (v18'.30) arc (30:150:0.2cm)
		-- (v14'.150) arc (150:270:0.2)
		-- (v15'.270) arc (-90:30:0.2) -- cycle;
	
	\filldraw [draw=black, fill=yellow, opacity=0.2]
		(v14'.30) -- (v19'.30) arc (30:150:0.2cm)
		-- (v13'.150) arc (150:270:0.2)
		-- (v14'.270) arc (-90:30:0.2) -- cycle;
		\filldraw [draw=black, fill=yellow, opacity=0.2]
		(v17'.30) -- (v21'.30) arc (30:150:0.2cm)
		-- (v18'.150) arc (150:270:0.2)
		-- (v17'.270) arc (-90:30:0.2) -- cycle;
	\filldraw [draw=black, fill=yellow, opacity=0.2]
		(v18'.30) -- (v20'.30) arc (30:150:0.2cm)
		-- (v19'.150) arc (150:270:0.2)
		-- (v18'.270) arc (-90:30:0.2) -- cycle;
	\filldraw [draw=black, fill=yellow, opacity=0.2]
		(v21'.30) -- (v22'.30) arc (30:150:0.2cm)
		-- (v20'.150) arc (150:270:0.2)
		-- (v21'.270) arc (-90:30:0.2) -- cycle;
	\filldraw [draw=black, fill=red, opacity=0.2]
		(v12'.30) -- (v1'.30) arc (30:150:0.2cm)
		-- (v23'.150) arc (150:270:0.2)
		-- (v12'.270) arc (-90:30:0.2) -- cycle;		
	\filldraw [draw=black, fill=red, opacity=0.2]
		(v25'.30) -- (v23'.30) arc (30:150:0.2cm)
		-- (v24'.150) arc (150:270:0.2)
		-- (v25'.270) arc (-90:30:0.2) -- cycle;		
	\filldraw [draw=black, fill=red, opacity=0.2]
		(v27'.30) -- (v24'.30) arc (30:150:0.2cm)
		-- (v28'.150) arc (150:270:0.2)
		-- (v27'.270) arc (-90:30:0.2) -- cycle;
	\filldraw [draw=black, fill=red, opacity=0.2]
		(v30'.30) -- (v28'.30) arc (30:150:0.2cm)
		-- (v29'.150) arc (150:270:0.2)
		-- (v30'.270) arc (-90:30:0.2) -- cycle;
	\filldraw [draw=black, fill=red, opacity=0.2]
		(v31'.30) -- (v27'.30) arc (30:150:0.2cm)
		-- (v30'.150) arc (150:270:0.2)
		-- (v31'.270) arc (-90:30:0.2) -- cycle;
	\filldraw [draw=black, fill=red, opacity=0.2]
		(v32'.30) -- (v26'.30) arc (30:150:0.2cm)
		-- (v31'.150) arc (150:270:0.2)
		-- (v32'.270) arc (-90:30:0.2) -- cycle;
	\filldraw [draw=black, fill=red, opacity=0.2]
		(v9'.30) -- (v10'.30) arc (30:150:0.2cm)
		-- (v32'.150) arc (150:270:0.2)
		-- (v9'.270) arc (-90:30:0.2) -- cycle;
	\filldraw [draw=black, fill=red, opacity=0.2]
		(v10'.30) -- (v11'.30) arc (30:150:0.2cm)
		-- (v26'.150) arc (150:270:0.2)
		-- (v10'.270) arc (-90:30:0.2) -- cycle;	
	\filldraw [draw=black, fill=red, opacity=0.2]
		(v26'.30) -- (v25'.30) arc (30:150:0.2cm)
		-- (v27'.150) arc (150:270:0.2)
		-- (v26'.270) arc (-90:30:0.2) -- cycle;	
	\filldraw [draw=black, fill=red, opacity=0.2]
		(v11'.30) -- (v12'.30) arc (30:150:0.2cm)
		-- (v25'.150) arc (150:270:0.2)
		-- (v11'.270) arc (-90:30:0.2) -- cycle;
	
	\filldraw [draw=black, fill=green, opacity=0.2]
		(v33'.30) -- (v5'.30) arc (30:150:0.2cm)
		-- (v6'.150) arc (150:270:0.2)
		-- (v33'.270) arc (-90:30:0.2) -- cycle;
	\filldraw [draw=black, fill=green, opacity=0.2]
		(v34'.30) -- (v33'.30) arc (30:150:0.2cm)
		-- (v35'.150) arc (150:270:0.2)
		-- (v34'.270) arc (-90:30:0.2) -- cycle;	
	\filldraw [draw=black, fill=green, opacity=0.2]
		(v35'.30) -- (v6'.30) arc (30:150:0.2cm)
		-- (v7'.150) arc (150:270:0.2)
		-- (v35'.270) arc (-90:30:0.2) -- cycle;
	\filldraw [draw=black, fill=green, opacity=0.2]
		(v38'.30) -- (v34'.30) arc (30:150:0.2cm)
		-- (v37'.150) arc (150:270:0.2)
		-- (v38'.270) arc (-90:30:0.2) -- cycle;	
	\filldraw [draw=black, fill=green, opacity=0.2]
		(v37'.30) -- (v35'.30) arc (30:150:0.2cm)
		-- (v36'.150) arc (150:270:0.2)
		-- (v37'.270) arc (-90:30:0.2) -- cycle;
	\filldraw [draw=black, fill=green, opacity=0.2]
		(v36'.30) -- (v7'.30) arc (30:150:0.2cm)
		-- (v8'.150) arc (150:270:0.2)
		-- (v36'.270) arc (-90:30:0.2) -- cycle;	
	\filldraw [draw=black, fill=green, opacity=0.2]
		(v42'.30) -- (v38'.30) arc (30:150:0.2cm)
		-- (v41'.150) arc (150:270:0.2)
		-- (v42'.270) arc (-90:30:0.2) -- cycle;
	\filldraw [draw=black, fill=green, opacity=0.2]
		(v41'.30) -- (v37'.30) arc (30:150:0.2cm)
		-- (v40'.150) arc (150:270:0.2)
		-- (v41'.270) arc (-90:30:0.2) -- cycle;
	\filldraw [draw=black, fill=green, opacity=0.2]
		(v40'.30) -- (v36'.30) arc (30:150:0.2cm)
		-- (v39'.150) arc (150:270:0.2)
		-- (v40'.270) arc (-90:30:0.2) -- cycle;
	\filldraw [draw=black, fill=green, opacity=0.2]
		(v39'.30) -- (v8'.30) arc (30:150:0.2cm)
		-- (v9'.150) arc (150:270:0.2)
		-- (v39'.270) arc (-90:30:0.2) -- cycle;	
	
   \filldraw [draw=black,fill=blue, opacity=0.2, line width=0.5pt, even odd rule]
		plot[smooth cycle, tension=1] coordinates { (v22'.90) (-3,1) (v29'.180) (0,-5.2) (v42'.0) (3,1)}
		plot[smooth cycle, tension=1] coordinates { (v22'.270) (-2.8,0.8) (v29'.90) (0,-5) (v42'.90) (2.8,0.8)}
		;

    \foreach \l in {13,...,42}{
      \filldraw [black] (v\l) circle (1pt) node [inner sep=5pt, label=below:$v_{\l}$] {};
    }
    \foreach \l in {1,...,12}{
      \filldraw [blue] (v\l) node [star, fill, star points=5, star point ratio=2.25, draw,inner sep=1pt] {};
	\filldraw [black] (v\l) node [inner sep=5pt ,label=below:$v_{\l}$]{};
    }
\end{tikzpicture}
}
\end{center}
\caption{(Fir tree) A $3$-uniform hypergraph $G$ of transfer type $b$ to which Theorem~\ref{theorem} applies. $G$ has $m=31$ edges, $n=42$ vertices, and foundation size $b=12$ (the vertices marked with a star are foundation vertices). The hypergraph's name was chosen due to its (vague) resemblance to a fir tree.}
\label{fig:firtree}
\end{figure}
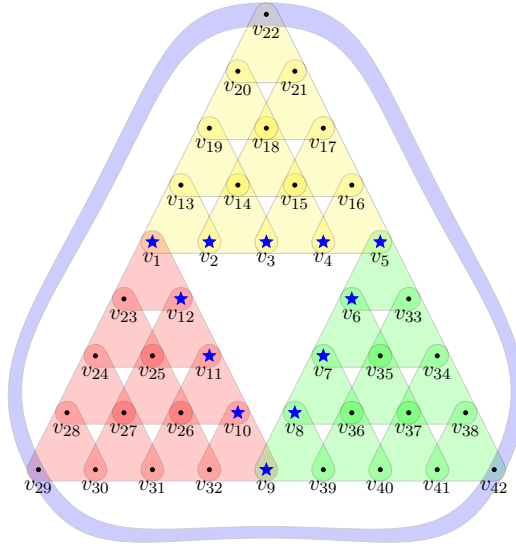
\paragraph{The main theorem.} We now show the main theorem of this section, which applies to $k$-uniform hypergraphs of transfer type $b=n-m+1$. The latter includes non-trivial instances which we discuss in Section~\ref{sscn:runtime} (along with an example of exponential speedup of our parameterized algorithm over brute diagonalization). In words, the theorem says that for any $k$-uniform hypergraph of transfer type $b=n-m+1$ (i.e. there is precisely one qualifier function $h_1$), if the constraints are chosen generically, then any zero of $h_1$ is the image under the map $\pi$ (defined in Qualifier Lemma) of a satisfying assignment to the corresponding $k$-QSAT instance. The key advantage to this approach is simple: To solve the $k$-QSAT instance, instead of solving a system of equations, we are reduced to solving for the roots of just one polynomial --- $h_1$. Moreover, if both the foundation size $b$ and the radius of the transfer filtration of $G$ scale as $O(\log m+\log n)$, then $h_1$ has polynomial degree in $m$ and $n$.

\begin{theorem}\label{theorem}
Let $\K$ be algebraically closed, and let $\spa{F}$ denote the set of $k$-uniform hypergraphs with $n$ vertices, $m$ edges, and transfer type $b=n-m+1$. If $H$ is a generic $k$-local constraint with interaction graph $G\in\spa{F}$ and $h_1$ and $\pi$ are as in the Qualifier Lemma (Lemma \ref{l:qualifier}), then $(h_1\circ \pi)^{-1}(0)\cap \S_H$ is nonempty.
\end{theorem}
\begin{proof}
Let $\widetilde{G}$ and $g_1,\ldots,g_{n+1}$ as in the proof of the Qualifier Lemma (Lemma \ref{l:qualifier}). Since $\K$ is algebraically closed, by the discussion in Remark~\ref{rem:reduction}, there exist $\hat{v}_1,\ldots, \hat{v}_b\in W$ such that $h_1(\hat{v}_1,\ldots,\hat{v}_b)=0$. As shown in the previous section, if additionally $g_{i_1}(\hat{v}_1,\ldots,\hat{v}_b)\neq 0$, then $\pi^{-1}(\hat{v}_1,\ldots,\hat{v}_b)\cap \S_H\neq\emptyset$, as desired. Thus, we are left to show that the same conclusion holds even if $g_{i_1}(\hat{v}_1,\ldots,\hat{v}_b)=0$ (which immediately implies $h_1(\hat{v}_1,\ldots,\hat{v}_b)=0$).

Let $j\in \{b+1,\ldots,n+1\}$ be minimal among those predecessors of $i_1$ such that $g_j(\hat{v}_1,\ldots,\hat{v}_b)=0$. For a sufficiently generic choice of constraint we may assume that $g_l(\hat{v}_1,\ldots,\hat{v}_b)\neq 0$ for any $l$ that is not a successor of $j$. Our strategy is now to \emph{remove} the edge in $\widetilde{G}$ which was added in the same step of the filtration in which vertex $j$ was added, namely edge $\widetilde{E_{j-b}}$. We then instead add vertex $j$ directly to the foundation. Formally, consider hypergraph $F$ obtained from $\widetilde{G}$ by removing edge $\widetilde{E_{j-b}}$, and consider transfer filtration $F_0\subseteq F_1\subseteq \cdots \subseteq F_{m-1}=F$ where: \begin{itemize}
    \item (Add $j$ to the foundation) $V(F_i)=V(\widetilde{G_i})\cup\{j\}$ for all $i\in\set{0,\ldots, j-b-1}$ and $V(F_i)=V(\widetilde{G_{i+1}})$ for all $i\in \set{j-b,\ldots, m-1}$,
    \item (Remove edge $\widetilde{E_{j-b}}$) $E(F_i)=E(\widetilde{G_i})$ for all $i\in \set{0,\ldots,j-b-1}$ and $E(F_i)=E(\widetilde{G_{i+1}})\setminus\{\widetilde{E_{j-b}}\}$ for all $i\in \set{j-b,\ldots, m-1}$.
\end{itemize}
The hypergraph $F$ thus has $n_F:=n$ vertices, $m_F:=m-1=n-b$ edges, and foundation of size $b_F:=b+1=n-m+2$. Thus, $m_F>n_F-b_F$, implying we can apply the Qualifier Lemma (Lemma~\ref{l:qualifier}) and the Surjectivity Lemma (Lemma \ref{l:surjectivity}) to obtain transfer functions $f_1,\ldots,f_{n_F}:W^{b_F}\mapsto W$ and $\gamma_{1,F},\ldots,\gamma_{n_F,F}:W\mapsto W$ (based on $f_i$ instead of $g_i$). Two comments are in order here: First, each $\gamma_{i,F}$ depends on $\hat{v}_1,\ldots, \hat{v}_b$ (which are fixed) and parameter $w$, which corresponds to the vertex $j$ added to the foundation in the construction of $F$. Second, for $i\in\set{0,\ldots,j-b-1}$, we have by construction that for all $w\in W$, $f_i(\hat{v}_1,\ldots, \hat{v}_b, w)=g_i(\hat{v}_1,\ldots,\hat{v}_b)$ --- this is because each corresponding $\widetilde{G_i}$ (i.e. the first $j-b$ ``steps'' of the original transfer filtration) does not depend on vertex $j$.

Now recall from Remark~\ref{rem:F} the polynomial
\[
    P_{i_1\underline{i_1}}(x):=\Gamma_{i_1,\underline{i_1}}(w'+xw'')=(\gamma_{i_1,F}(w'+xw''))^\sharp(\gamma_{\underline{i_1},F}(w'+xw''))\in \K[x].
\]
Since $\K$ is algebraically closed, $P_{i_1\underline{i_1}}(x)$ has a root $\hat{x}$ in $\K$. Thus, by Remark \ref{rem:7}, there exists $\lambda\in \K$ such that $\gamma_{i_1,F}(w'+\hat{x}w'')=\lambda \gamma_{\underline{i_1},F}(w'+\hat{x}w'')$ (i.e. both decoupled vertices $i_1$ and $\underline{i_1}$ can be given the same assignment). Therefore, if we set ${v}_i=\gamma_{{i},F}(w'+\hat{x}w'')$ for all $i\in\{1,\ldots,n\}$ (note we may assume by Remark~\ref{rem:non-zero} that each $v_i\neq 0$), then (1) vertices $i_1$ and $\underline{i_1}$ receive consistent assignments, (2) edge $\widetilde{E_{j-b}}$ is satisfied by the second remark above that $f_i(\hat{v}_1,\ldots, \hat{v}_b, w)=g_i(\hat{v}_1,\ldots,\hat{v}_b)$ for $i\in\set{0,\ldots,j-b-1}$ (here we also use the fact that by Equation~(\ref{eq:sharp}), edge $\widetilde{E_{j-b}}$ was satisfied by foundation assignment $(\hat{v}_1,\ldots, \hat{v}_b)$ regardless of the assignment to vertex $i_1$), and (3) $({v}_1,\ldots,{v}_n)$ satisfies all constraints in hypergraph $F$. It follows that $({v}_1,\ldots,{v}_n)\in \pi^{-1}(\hat{v_1},\ldots,\hat{v_b})\cap \S_H$, as required.
\end{proof}

\begin{rem}[The maps $\gamma_i(w)$ have no zeroes]\label{rem:non-zero}
The maps $\gamma_i(w)$ from the Surjectivity Lemma (Lemma~\ref{l:surjectivity}) can without loss of generality be taken to have no zeroes, i.e. $\nexists w\in W$ such that $\gamma_i(w)=0$. To see this, consider $\gamma_i(w'+xw''):W\mapsto W$ for $x\in \K$ as in Remark~\ref{rem:F}. Since $w'$ and $w''$ are linearly independent, we can write $\gamma_i(w'+xw'')=p(x)w'+q(x)w''$ for polynomials $p,q\in \K[x]$. In order for the right hand size to vanish on $x_0$, since $\K$ is algebraically closed, $p$ and $q$ must share a common factor $(x-x_0)$. Factoring out $(x-x_0)$, we obtain new polynomials $p'$ and $q'$, and redefine $\gamma'_i(w'+xw'')=p'(x)w'+q'(x)w''$. (Note we are implicitly working in projective space, since solutions to the Quantum SAT instance remain solutions under (non-zero) scaling.)
\end{rem}

\begin{example}[Running example] We illustrate the proof of Theorem \ref{theorem} by specializing the construction to the $3$-uniform semicycle $G_3$ from Example \ref{ex:4cycle1}. Then $\widetilde G_3$ is the hypergraph with vertices $\{1,2,3,4,5\}$ and edges $\widetilde E_1=\{1,2,3\}$, $\widetilde E_2=\{1,2,4\}$, $\widetilde E_3=\{1,3,5\}$. Moreover, the transfer functions $g_1,\ldots,g_5:W^2\to W$ can be chosen as in Example \ref{ex:4cycle6}. Let $h_1$ be as in Example \ref{ex:4cycle7} and suppose $v_1,v_2\in W$ are such that $h_1(v_1,v_2)=0$. If none of the $g_i(v_1,v_2)$ are zero, then a solution of the form $(v_1,v_2,v_3,v_4)$ can be found by Remark~\ref{rem:after}. Else, suppose (say) $g_3(v_1,v_2)=0$ (generically, only one $g_i(v_1,v_2)$ will be zero in this case). Then the hypergraph $F$ has vertices $\{1,2,3,4,5\}$ and edges $\widetilde E_2$, $\widetilde E_3$. Furthermore, $F$ has a transfer filtration $F_0\subseteq F_1\subseteq F_2=F$, where $F_0$ has vertices $\{1,2,3\}$ and no edges, while $F_1$ has vertices $\{1,2,3,4\}$ and a single edge $\widetilde E_2$. Unraveling \eqref{eq:sharp} we obtain
\begin{align*}
\gamma_{1,F}(w)&=v_1\,;\\
\gamma_{2,F}(w)&=v_2\,;\\
\gamma_{3,F}(w)&=w\,;\\
\gamma_{4,F}(w)&=H_2^*(v_1\otimes v_2\otimes w'')w'-H_2^*(v_1\otimes v_2\otimes w')w''\,;\\
\gamma_{5,F}(w)&=H_3^*(v_1\otimes w\otimes w'')w'-H_3^*(v_1\otimes w\otimes w')w''
\end{align*}
and thus
\[
\Gamma_{54}(w)=H_3^*(v_1\otimes w \otimes w')H_2^*(v_1\otimes v_2\otimes w'')-H_3^*(v_1\otimes w\otimes w'')H_2^*(v_1\otimes v_2\otimes w')
\]
for all $w\in W$. Hence $P_{54}(x)=P_{54}^0+P_{54}^1x$ where
\[
P_{54}^0=H_3^*(v_1\otimes w' \otimes w')H_2^*(v_1\otimes v_2\otimes w'')-H_3^*(v_1\otimes w'\otimes w'')H_2^*(v_1\otimes v_2\otimes w')
\]
and
\[
P_{54}^1=H_3^*(v_1\otimes w'' \otimes w')H_2^*(v_1\otimes v_2\otimes w'')-H_3^*(v_1\otimes w''\otimes w'')H_2^*(v_1\otimes v_2\otimes w')\,.
\]
In this case it is immediately clear that generically $P_{54}(x)$ is non-constant and, being linear, has a root over any (non-necessarily algebraically closed) field. However, in general, $P_{i\underline i}(x)$ may have arbitrarily high degree and assuming that $\mathbb K$ is algebraically closed becomes necessary.

\end{example}

\begin{rem}[Reduction to univariate polynomials]\label{rem:reduction}
Theorem~\ref{theorem} reduces us to solving a single polynomial equation,
\[
h_1(v_1,\ldots,v_b)=0,
\]
which is multi-variate. In this case, we reduce it further to a univariate polynomial by fixing arbitrary vectors $w_1,\ldots,w_b\in W$ and $w'_b\in W$ linearly independent from $w_b$. Then
\[
P(x)=h_1(w_1,\ldots,w_b+xw'_b)
\]
is a \emph{univariate} polynomial in $\mathbb K[x]$, which has a root $x\in \K$ since $\K$ is algebraically closed. Thus, we can reduce the problem of computing a root of the multivariate polynomial $h_1$ to computing a root of the univariate polynomial $P$.
\end{rem}

Before explicitly stating our algorithm and formally analyzing runtimes in Section~\ref{sscn:runtime}, we finally tie Theorem~\ref{theorem} back to SDRs, and give an AG overview of our approach.

\begin{rem}
     [Connection to SDRs] Theorem~\ref{theorem} gives a \emph{constructive} proof that all hypergraphs in $\spa{F}$ have a satisfying (product state) assignment if the constraints are chosen generically. This implies each hypergraph in $\spa{F}$ has an SDR, as we now observe.
     \begin{corollary}\label{cor:actually}
        If $G$ is a $k$-uniform hypergraph of transfer type $b$ and such that $|E(G)|=|V(G)|-b+1$, then $G$ has an SDR.
    \end{corollary}
    \begin{proof}
        By Theorem \ref{theorem}, $\S_H$ is nonempty for generic $H\in \X_G(\mathbb C)$. Theorem 2 of \cite{LLMSS10} then implies that $G$ has an SDR.
    \end{proof}
     \noindent Thus, Theorem~\ref{theorem} constructively recovers the result of~\cite{LLMSS10} (that any $k$-QSAT instance with an SDR has a (product-state) solution) in the case when the additional conditions of Theorem~\ref{theorem} are met (recall~\cite{LLMSS10} works on {all} graphs with an SDR, but is not constructive). Finally, in Theorem~\ref{thm:general}, we show a generalization of Corollary~\ref{cor:actually} above with a self-contained proof based on Hall's marriage theorem.
\end{rem}

\begin{rem}[AG view]
 As stated at the outset of Section~\ref{scn:parameterized}, we now give a brief sketch of the ideas discussed here in AG terms. It is natural to represent product states using the Segre map $(v_1,\ldots,v_n)\to v_1\otimes \cdots \otimes v_n$ which describes an embedding of the $n$-fold product of projective lines into a projective space of dimension $2^n-1$. Using this representation, the product satisfiability $\S_H$ can be naturally thought of as an algebraic set in $({\mathbb P}^1(\mathbb K))^n$. In the language of algebraic geometry, the Transfer Lemma (Lemma~\ref{prop:7}) says that if the interaction graph is $k$-uniform of transfer type $b=|V(G)|-|E(G)|$, then $\S_H$ can be realized as an iterated blow-up of $({\mathbb P}^1(\mathbb K))^b$, so in particular it is a rational variety. Moreover, the blow-ups are taken along the loci where the transfer functions $g_i$ vanish or equivalently where their projectivizations (which are by construction rational functions in suitable coordinates) are indeterminate. On the other hand, the Qualifier Lemma (Lemma~\ref{l:qualifier}) says that the projection of $\S_H$ along the coordinates corresponding to $G_0$ is contained in the algebraic set $X_H$ cut out by the regular functions $h_s$. In general, we cannot guarantee that the projection of $\S_H$ coincides with $X_H$. However, Theorem~\ref{theorem} says that in the non-trivial case when $X_H$ is a hypersurface, then the projection of $\S_H$ does indeed coincide with $X_H$.
\end{rem}

\subsection{Explicit algorithm statement, runtime, and speedups}\label{sscn:runtime}
Figure~\ref{alg:statement} states our algorithm explicitly. We now give a runtime analysis, as well as families of inputs for which a speedup (polynomial to exponential) is obtained over brute force methods.

\begin{figure}
\begin{adjustbox}{varwidth=\textwidth,margin=2 {\abovecaptionskip} 5 10, frame=1pt }
\begin{algorithmic}[1]
\State \textbf{input:} A $k$-uniform hypergraph $G=(V,E)$ (along with a constraint $H_e$ for each hyperedge $e$), transfer filtration $F$ of $G$ of type $b$ and radius $r$.
\State \textbf{output:} An $n$-qubit tensor product assignment satisfying all constraints, if one exists. Else, the algorithm rejects.
\State
\Procedure{BLOWUP}{$G,F$}
\State Apply Decoupling Lemma (Lemma~\ref{lem:p}) to $G$ to compute $\widetilde{G}$ and $p:V(\widetilde{G})\mapsto V(G)$.
\State Apply Transfer Lemma (Lemma~\ref{prop:7}) to $\widetilde{G}$ to compute $g_1,\ldots,g_n:W^b\mapsto W$.
\State Apply Qualifier Lemma (Lemma~\ref{l:qualifier}) to $g_1,\ldots, g_n$ to compute $h_1,\ldots, h_{m-n+b}:W^b\mapsto\mathbb{K}$.
\State Compute foundation assignment $v_1,\ldots, v_b\subseteq W$ such that for all $s\in[m-n+b]$, $h_s(v_1,\ldots, v_b)=0$ and $g_{i_s}(v_1,\ldots,v_b)\neq 0$ . If no such assignment found, reject.
\State For each qubit $i$ of $G$, return assignment $g_{i}(v_1,\ldots,v_b)$.
\EndProcedure
\end{algorithmic}
\end{adjustbox}
\caption{Parameterized algorithm for Quantum $k$-SAT.}\label{alg:statement}
\end{figure}

\paragraph{Runtime analysis.} We bound the complexity of each step of Figure~\ref{alg:statement}. Note these bounds are rather naive; improved polynomial factors can likely be attained by an appropriate use of data structures and a more careful analysis. Below, we assume $\mathbb{K}=\complex$, use of the Random Access Machine model (e.g. looking up entries in arrays costs constant time), and we count the number of field operations over $\complex$ (as opposed to bit operations).

\begin{itemize}
    \item (Decoupling Lemma, Line $5$) Let $(E_1,\ldots, E_m)$ denote the order in which edges of $G$ are added in the filtration $F$. In order to compute $\widetilde{G}$ and $p$, the first step is to compute $r:\set{1,\ldots, m}\mapsto\set{0,\ldots,m-1}$ from Remark~\ref{rem:r}. Recall this was defined such that $r(i)$ is the smallest integer satisfying $\abs{E_i\setminus V(G_{r(i)})}=1$, i.e. the first point in the filtration at which $E_i$ could be added next via the introduction of precisely one vertex. Since $k\in O(1)$, one can compute all vertex sets $V(G_{i})$ in $O(m)$ time. With all $V(G_i)$ precomputed, computing $r(i)$ now takes $O(i)$ time, and so precomputing all values of $r$ takes $O(m+m^2)\in O(m^2)$ total time.

        Moving on to $p$, since $p(i)=i$ when $i\in[b]$ and $p(i)=E_{i-b}\setminus V(G_r(i-b))$ when $i\in\set{b+1,\ldots, b+m}$, each entry $p(i)$ can now be computed in $O(1)$ time, implying all values of $p$ can be computed in time $O(m+b)$.

        Since the preimage $p^{-1}(j)$ can be of size up to $m$, computing it takes $O(m)$ time, implying computing $\underline{j}$ and  $\widetilde{E_i}$ take $O(m)$ time. Therefore, computing all edges $\widetilde{E_i}$ takes $O(m(m+b))$ time.

        The total runtime for this step is hence $O(m^2)+O(m+b)+O(m(m+b))\in O(m(m+b))$.
    \item (Transfer Lemma, Line $6$) To analyze the cost of computing a transfer function (where by ``computing'', we mean explicitly writing out all coefficients of the corresponding multi-variate polynomial) , $g:W^b\mapsto W$, it is helpful to use the construction of Example~\ref{ex:4cycle6} as a guide. Namely, since $W$ is two-dimensional, it suffices to analyze the maximum cost of outputting a component of $g$'s two-dimensional output vector. This requires bounding the cost of expanding an expression of the form
        \[
            H^*(g_{i_1}(v_1,\ldots,v_b)\otimes\cdots\otimes g_{i_{k-1}}(v_1,\ldots,v_b)\otimes w_l)=: H^*(\ve{u})
        \]
        (for $w_l$ some fixed basis vector for $W$) as a multivariate polynomial in the $2b$ components of $v_1,\ldots, v_b$. For convenience, let $\ve{u}=g_{i_1}(v_1,\ldots,v_b)\otimes\cdots\otimes g_{i_{k-1}}(v_1,\ldots,v_b)\otimes w_l$ above, and set $S=\set{i_1,\ldots, i_{k-1}}$.  Observe now that if $g_{i_t}$ has degree $(d_{i_t,1},\ldots, d_{i_t,b})$, then since $g_{i_t}$ is a multivariate polynomial, it has at most $O(d_{i_t,1}d_{i_t,2}\cdots d_{i_t,b})$ distinct monomials. It follows that computing an arbitrary entry of $\ve{u}$ costs $O(\Pi_{i\in S}\Pi_{j=1}^b d_{i,j})$ time (for this simple analysis, we naively multiply out all monomials and add like terms). Note next that if $H$ denotes the $k$-local Hamiltonian term corresponding to $H^*$, then $H^*(\ve{u})=\ve{u}^\dagger H\ve{u}$. Thus, assuming $k\in O(1)$, computing $H^*(\ve{u})$ costs
        \begin{equation}\label{eqn:cost}
            O\left(\Pi_{i\in S}\;\Pi_{j=1}^b\; d^2_{i,j}\right)
        \end{equation}
        time in the worst case. By the degree bounds of Lemma~\ref{prop:7} and Remark~\ref{rem:fib}, we have $d_{i,j}\leq F^{(b)}_i\in o(2^i)$. Equation~(\ref{eqn:cost}) can hence be bounded by
        \begin{equation}\label{eqn:cost2}
            O\left(\Pi_{i\in S}\; d^{2b}_{i,j}\right)\in O\left(2^{2b\sum_{i\in S}i}\right)\in O\left(2^{2bk\max_{i\in S}i}\right),
        \end{equation}
        where the last statement follows since $\abs{S}\leq k$. A naive analysis would now suggest $\max_{i\in S}\leq n$, as there are $n$ transfer functions $g$. However, this analysis is too loose --- it assumes the worst case that each tensor factor $g_{i_1}(v_1,\ldots, v_b)$ depends recursively on all previous transfer functions $g_1\ldots, g_{i_1-1}$ (hence the degree of $g_{i_1-1}\leq F^b_i$ for $i\in[n]$). By taking the \emph{radius} of the filtration into account, $\rho(G,b)$, we can obtain a much smaller bound when the radius is small. Specifically, similar to the analysis of Lemma~\ref{l:qualifier}, one can instead show $d_{i,j}\leq F^{(b)}_{\rho(G,b)+b+1}\in O(2^{\rho(G,b)+b})$. Thus, we can improve Equation~(\ref{eqn:cost2}) to
        \[
           O\left(2^{2kb(\rho(G,b)+b)}\right).
        \]
        Since there are $n$ transfer functions in total, the cost of Line 6 is hence $O\left(n2^{2kb(\rho(G,b)+b)}\right)$, where recall we assume $k\in O(1)$.

        \item (Qualifier Lemma, Line $7$) We consider the cost of computing an arbitrary qualifier $h_s$. Since in Line 5 we already computed $\underline{i}$ for all $i\in\set{1,\ldots, m+b}$, we may compile set $\set{i_1,\ldots, i_{m-n+b}}$ in time $O(m-n+b)$. Since $W$ is a two-dimensional space, computing $g_{i_s}^\#$ takes time linear in the number of monomials in each amplitude/component of $g_{i_s}$, i.e. $O(d_{i_s,1}\cdots d_{i_s,b})\in O\left(2^{b(\rho(G,b)+b)}\right)$ time. It follows that computing $h_s$ requires $O\left(2^{2b(\rho(G,b)+b)}\right)$ time. Since there are $m-n+b$ qualifier functions, the total cost of this step is
            \[
                O\left((m-n+b)2^{2b(\rho(G,b)+b)}\right).
            \]
\end{itemize}
We may bound the cost of Lines $5$ - $7$ (which apply to any instance of $k$-QSAT) by
\begin{equation}\label{eqn:cost567}
    O\left((m-n+b)2^{2kb(\rho(G,b)+b)}\right).
\end{equation}
For arbitrary $k$-QSAT instances, the runtime of Lines $8$ - $9$ is less clear (e.g. how to find a common root to all $h_s$ while also ensuring $g_{i_s}$ to be non-zero?); however, these steps can indeed be solved for the family of instances with generic constraints and transfer type $b=n-m+1$, as demonstrated in Section~\ref{sscn:generic}. As noted in Remark~\ref{rem:reduction}, in this case we may reduce Lines $8$ - $9$ to solving for the roots of a single univariate polynomial (moreover, the Surjectivity Lemma~\ref{l:surjectivity} ensures $g_{i_s}$ outputs a non-zero value). We now analyze this reduction's runtime.

\begin{itemize}
    \item (Lines $8$ - $9$, generic constraints, $b=n-m+1$) When $b=n-m+1$, there is precisely one qualifier function $h_s(v_1,\ldots, v_b)$, which is a multivariate polynomial in $2b$ complex variables, and has degree $(d_{s1},\ldots,d_{sb})$ for $d_{sr}\le 2^{\rho(G,b)+b+2}$ (by Lemma~\ref{l:qualifier}) for all $r\in[b]$. Following Remark~\ref{rem:reduction}, fix $w_1=\cdots =w_b=\ket{0}\in\complex^2$ and $w'_b=\ket{1}\in\complex^2$ and consider variable $x\in \complex$. Then, $q(x):=h_s(w_1,\ldots, w_{b-1},w_b+x w'_b)$ is a univariate non-constant\footnote{Note that, in principle, $q(x)$ could be constant. However, since $h_s$ is non-constant by the Qualifier Lemma, there must be at least one coordinate $j\in\set{1,\ldots, b}$ such that redefining $q(x)$ to have its $j$th argument to equal $w_j+xw_j'$, and all other arguments $j'\neq j$ equal to $\ket{\psi}$ for some appropriate constant vector $\ket{\psi}$, yields a non-constant polynomial over $x$. (For example, set the entries of $\ket{\psi}$ to be sufficiently large in magnitude to ensure the monomials of different degrees in the expansion of $h_s$ cannot cancel out. Recall here that since we are essentially working in projective space, $\ket{\psi}$ need not be normalized.)} polynomial over $x$ with degree at most
        \begin{equation}\label{eqn:degbound}
            d\leq 2^{\rho(G,b)+b+2}.
        \end{equation}
        Observe that $h_s$ has at most $M=\Pi_{s}d_{sr}\leq 2^{b(\rho(G,b)+b+2)}$ monomials. Thus, computing $q(x)$ takes $O(M)$ time since each of the $2b$ complex variables plugged into $h_s$ is given a value from set $\set{0,1,x}$ (i.e. simplifying monomials with large exponents is trivial). To find a root of $q$, which exists since $\complex$ is algebraically closed, we can now apply an algorithm of Sch\"{o}nhage~\cite{S86} as follows.

        \textbf{Sch\"{o}nhage's univariate polynomial factorization algorithm.} Let us write $q(x)=\sum_{i=1}^d c_i x^i\in\complex[x]$, and define one-norm $\norm{q}_1:=\sum_{i=1}\abs{c_i}$. Then, Equation (3.2) of~\cite{S86} says ``approximate factorization'' of $q$ is possible within error\footnote{For clarity, Equation (3.2) of~\cite{S86} assumes $\norm{q}_1\leq 1$, and hence the error is stated therein as $\epsilon=2^{-N}$. In our setting, however, we first need to rescale our $q$ to ensure $\norm{q}_1\leq 1$. This does not affect the number of field operations required for Sch\"{o}nhage's algorithm, but results in the multiplicative $\norm{q}_1$ term (which is absent in~\cite{S86}) in the resulting bound on accuracy, Equation~(\ref{eqn:almost}) . We hence define the error as $\epsilon=2^{-N}\norm{q}_1$ here for convenience.} $\epsilon=2^{-N}\norm{q}_1$ in $O(d^3\log d + d^2N)$ field operations. Here, ``approximate factorization'' within error $\epsilon$ means computing linear factors $L_j(x)=u_jx+v_j$ for $1\leq j \leq d$ such that
        \begin{equation}\label{eqn:almost}
            \norm{q(x)-L_1(x)\cdots L_d(x)}_1< 2^{-N}\norm{q}_1.
        \end{equation}
        Let $w_j=-v_j/u_j\leq 1$ denote the root of linear factor $L_j(x)$. Given Equation~(\ref{eqn:almost}), the next question is: How closely do the $w_j$ approximate the roots $z_j$ of $q(x)$? As done in Section 3.4 of~\cite{S86}, one can apply the perturbation bound of Theorem 2.7 of~\cite{SCHONHAGE1985118} to obtain the following. If $\norm{q(x)-L_1(x)\cdots L_d(x)}_1\leq \eta \norm{q(x)}_1$ for $\eta\leq 2^{-7d}$, then there exists an ordering of the $w_j$ such that for all $j$,
        \[
            \abs{w_j-z_j}< 9\eta^{1/d}.
        \]
        It follows that in order to compute the roots $z_j$ within additive error $2^{-p(n)}$ for some polynomial $p$, it suffices to approximately factorize $q$ within error $\epsilon=\norm{q}_1/(9^d2^{p(n)d})$, which can be accomplished in time
        \[
            O\left(d^3\log d + d^2 \log\left(9^d2^{p(n)d}\right)\right).
        \]

         \item (Line 9) We now have a root $(v_1:=\ket{0},\ldots, v_{b-1}:=\ket{0},v_b:=\ket{0}+x\ket{1})$ of $h_s$. To compute $g_i(v_1,\ldots, v_b)$ for any $i\in[n]$ involves evaluating a multivariate polynomial of degree $(d_{s1},\ldots,d_{sb})$ for $d_{sr}\le 2^{\rho(G,b)+b+2}$ for all $r\in[b]$. Naively substituting in our root and using square-and-multiply to compute monomials with large powers, evaluating each $g_i$ hence requires at most time
             \[
                O((\rho(G,b)+b)2^{b(\rho(G,b)+b+2)}).
             \]
\end{itemize}
Combining the cost of all steps, we find $k$-QSAT instances with generic constraints and $b=n-m+1$ require total time at most (for $d\leq 2^{\rho(G,b)+b+2}$)
\begin{equation}\label{eqn:runtime}
       O(mn)+ O\left(2^{2kb(\rho(G,b)+b)}\right)+O\left(d^3\log d + d^2 \log\left(9^d2^{p(n)d}\right)\right)+O((\rho(G,b)+b)2^{b(\rho(G,b)+b+2)}).
\end{equation}
Thus, the algorithm is polynomial in $m$ and $n$, and exponential in $k$ (the locality of the constraints), $\rho(G,b)$ (the radius), and $b$ (foundation size).\\

\noindent\emph{Aside: Error propagation.} Note that in Line 8, we compute a root $z$ of $q$ within additive error $2^{-p(n)}$. In Line 9, we then substitute this root into our transfer functions $g_i$  to obtain our final assignment on qubit $i$. This substitution can propagate the additive error from Line 8. However, since the degrees of the transfer functions are bounded by $\exp(n)$ (specifically, $d_{sr}\le 2^{\rho(G,b)+b+2}$), it follows that by choosing $p$ to be a sufficiently large fixed polynomial, one can make the additive error for the output of $g_i$ exponentially small in $n$.

\paragraph*{On exponential speedups via Theorem~\ref{theorem}.}
Recall Theorem~\ref{theorem} applies to $k$-uniform hypergraphs of transfer type $b=n-m+1$, such as the semicycle. From a parameterized complexity perspective, however, most interesting are hypergraphs for which the foundation size $b$ and filtration radius $r$ satisfy $b,r\in o(n+m)$, for which we might obtain an asymptotic speedup over brute force diagonalization of the Quantum SAT system. We now discuss various hypergraph families and analyze their parameters $m$ and $n$ versus $b$ and $r$. In particular, we obtain quadratic (tiling of torus, fir tree) to \emph{exponential} (crash) separations between these parameters. Note that for the runtime of Equation~(\ref{eqn:runtime}), a quadratic separation is unfortunately not enough for an asymptotic speedup. However, an \emph{exponential} separation in parameters, as for the crash family of hypergraphs,  implies our parameterized algorithm runs in polynomial time, whereas brute force diagonalization would require time exponential in $m$ and $n$.

\begin{example}[Semicycles (Figure~\ref{fig:chain}, page $9$)]
We begin with the study of semicycles, for which there is no separation in parameters. Namely, let $S_{t,k}$ denote the $k$-uniform hypergraph with vertices $V(S_{t,k})=\mathbb Z/t\mathbb Z$ and edges $E(S_{t,k})=\{E_0,\ldots,E_{t-k+1}\}$ such that $E_{i}=\{i,i+1,\ldots,i+k-1\}$ for every $i\in \{0,\ldots,t-k+1\}$. By construction, $|V(S_{t,k})|-|E(S_{t,k})|+1=k-1$ for every $t$. On the other hand, $S_{t,k}$ is line graph connected for every $t$ and thus has transfer type $k-1$ by Example \ref{ex:L}. Depending on how the transfer filtration is chosen, the transfer radius is at least $|E(S_{t,k})|/2$ and at most $|E(S_{t,k})|$. Thus, the radius $r$ satisfies $r\in \Theta (|E(S_{t,k})|)$, i.e. no separation in parameters exists.
\end{example}

\begin{figure}[t]
	\begin{center}
\resizebox{150pt}{!}{%
\begin{tikzpicture}
	
    \coordinate (v1) at (0,6) {};
    \coordinate (v2) at (0,0) {};
    \coordinate (v3) at (6,0) {};
    \coordinate (v4) at (0,4) {};
    \coordinate (v5) at (6,4) {};
    \coordinate (v6) at (0,2) {};
    \coordinate (v7) at (6,2) {};
    \coordinate (v8) at (2,6) {};
    \coordinate (v9) at (2,0) {};
    \coordinate (v10) at (4,6) {};
    \coordinate (v11) at (4,0) {};
    \coordinate (v12) at (2,4) {};
    \coordinate (v13) at (4,4) {};
    \coordinate (v14) at (2,2) {};
    \coordinate (v15) at (4,2) {};

    \foreach \v in {v1,v2,v3,v4,v5,v6,v7,v8,v9,v10,v11,v12,v13,v14,v15}{
      \node [circle, minimum size=0.4cm, line width=0pt] (\v') at (\v) {};
    }

    \filldraw [draw=black, fill=yellow, opacity=0.2]
		(v1'.180) -- (v4'.180) arc (180:270:0.2cm)
		-- (v12'.-90) arc (-90:45:0.2)
		-- (v1'.45) arc (45:180:0.2) -- cycle;
    \filldraw [draw=black, fill=violet, opacity=0.2]
		(v8'.180) -- (v12'.180) arc (180:270:0.2cm)
		-- (v13'.-90) arc (-90:45:0.2)
		-- (v8'.45) arc (45:180:0.2) -- cycle;
    \filldraw [draw=black, fill=yellow, opacity=0.2]
		(v10'.180) -- (v13'.180) arc (180:270:0.2cm)
		-- (v5'.-90) arc (-90:45:0.2)
		-- (v10'.45) arc (45:180:0.2) -- cycle;
    \filldraw [draw=black, fill=violet, opacity=0.2]
		(v4'.180) -- (v6'.180) arc (180:270:0.2cm)
		-- (v14'.-90) arc (-90:45:0.2)
		-- (v4'.45) arc (45:180:0.2) -- cycle;
    \filldraw [draw=black, fill=yellow, opacity=0.2]
		(v12'.180) -- (v14'.180) arc (180:270:0.2cm)
		-- (v15'.-90) arc (-90:45:0.2)
		-- (v12'.45) arc (45:180:0.2) -- cycle;
    \filldraw [draw=black, fill=violet, opacity=0.2]
		(v13'.180) -- (v15'.180) arc (180:270:0.2cm)
		-- (v7'.-90) arc (-90:45:0.2)
		-- (v13'.45) arc (45:180:0.2) -- cycle;
    \filldraw [draw=black, fill=yellow, opacity=0.2]
		(v6'.180) -- (v2'.180) arc (180:270:0.2cm)
		-- (v9'.-90) arc (-90:45:0.2)
		-- (v6'.45) arc (45:180:0.2) -- cycle;
    \filldraw [draw=black, fill=violet, opacity=0.2]
		(v14'.180) -- (v9'.180) arc (180:270:0.2cm)
		-- (v11'.-90) arc (-90:45:0.2)
		-- (v14'.45) arc (45:180:0.2) -- cycle;
    \filldraw [draw=black, fill=yellow, opacity=0.2]
		(v15'.180) -- (v11'.180) arc (180:270:0.2cm)
		-- (v3'.-90) arc (-90:45:0.2)
		-- (v15'.45) arc (45:180:0.2) -- cycle;

    \foreach \l in {1,...,3}{
      \filldraw [black] (v\l) circle (2pt) node [inner sep=5pt, label=below:$v_{1}$] {};
	}
    \foreach \l in {4,5}{
      \filldraw [black] (v\l) circle (2pt) node [inner sep=5pt, label=below:$v_{2}$] {};
    }
    \foreach \l in {6,7}{
      \filldraw [black] (v\l) circle (2pt) node [inner sep=5pt, label=below:$v_{3}$] {};
    }
    \foreach \l in {8,9}{
      \filldraw [black] (v\l) circle (2pt) node [inner sep=5pt, label=below:$v_{4}$] {};
    }
     \foreach \l in {10,11}{
      \filldraw [black] (v\l) circle (2pt) node [inner sep=5pt, label=below:$v_{5}$] {};
     }
     \foreach \l in {12}{
      \filldraw [black] (v\l) circle (2pt) node [inner sep=5pt, label=below:$v_{6}$] {};

          }
     \foreach \l in {13}{
      \filldraw [black] (v\l) circle (2pt) node [inner sep=5pt, label=below:$v_{7}$] {};

          }
     \foreach \l in {14}{
      \filldraw [black] (v\l) circle (2pt) node [inner sep=5pt, label=below:$v_{8}$] {};

          }
     \foreach \l in {15}{
      \filldraw [black] (v\l) circle (2pt) node [inner sep=5pt, label=below:$v_{9}$] {};
     }
\end{tikzpicture}
}

\end{center}
\caption{A $3\times 3$ tiling of the torus. Note the closed boundary conditions, i.e. we identify two vertices with the same label as being the same vertex.}
\label{fig:torus}
\end{figure}

\begin{example}[Modified tiling of the torus (Figure~\ref{fig:torus}, page $34$)] \label{ex:triangle}
We next discuss a slight modification of the tiling of torus which yields a quadratic separation in parameters. Let $k\ge 3$ and let $t$ be a positive integer. Let $V_{t,k}$ be the set of pairs $(p,q)$ of integers such that $0\le (k-2) q \le p \le (k-2) t$. Let $G_{t,k}$ be the $k$-uniform hypergraph with vertices $V(G_{t,k})=V_{t,k}$ and edges $E_{r,s}=\{(r,s),(r+1,s),\ldots,(r+k-2,s),(r+k-2,s+1)\}$ for each pair $(r,s)\in V_{t-1,k}$. Then $G_{t,k}$ has $\frac{(t+1)(2+t(k-2))}{2}$ vertices and $\frac{t(2+(t-1)(k-2))}{2}$ edges. All vertices except for $(0,0)$, $(t(k-2),0)$ and $(t(k-2),t)$ have degree at least $2$. There exists a transfer filtration with foundation $\{(0,0),\ldots,(t(k-2),0)\}$ obtained by adding the edges in the following order
\[
E_{0,0}, E_{1,0},\ldots,E_{(t-1)(k-2),0},E_{k-2,1},E_{k-1,1},\ldots,E_{(t-1)(k-2),1},E_{2(k-2),2},\ldots,E_{(t-1)(k-2),t-1}\,.
\]
In particular, this filtration has radius $t$.

Now, consider the $k$-uniform hypergraph $T_{t,k}$ obtained by identifying the degree $1$ vertices of $G_{t,k}$. (This can be visualized similar to the tiling of the torus as in Figure~\ref{fig:torus}, except one keeps only the ``lower triangular'' portion of Figure~\ref{fig:torus}, i.e. vertices $v_2,v_3,v_4,v_5,v_7$ from Figure~\ref{fig:torus} are discarded. Unlike the tiling of the torus, however, the example here has a transfer filtration type which satisfies the preconditions of Theorem~\ref{theorem}, as stated below.) Then $T_{t,k}$ has a transfer filtration of type
\begin{align*}
t(k-2) &= \frac{(t+1)(2+t(k-2))}{2} -2 + \frac{t(2+(t-1)(k-2))}{2} +1  \\
& = |V(T_{t,k})|-|E(T_{t,k})|+1,
\end{align*}
satisfying the preconditions for Theorem~\ref{theorem}. Since $T_{t,k}$ has radius $t$, it follows that the radius and foundation size are $\Theta(t)$ (for $k\in O(1)$, which is the typical assumption), whereas $\abs{V(T_{t,k})},\abs{E(T_{t,k})}\in \Theta(t^2)$, i.e. there is a quadratic gap between these parameters.
\end{example}

\begin{example}[Fir Tree (Figure~\ref{fig:firtree}, page $27$)]
We next give another example of a hypergraph family (Figure~\ref{fig:firtree}) with a quadratic separation in parameters. For each $k\ge 3$, $t\ge 1$ and $i\in \mathbb Z/k\mathbb Z$ let $G_{t,k,i}$ be a copy of the $k$-uniform hypergraph $G_{t,k}$ introduced in Example \ref{ex:triangle}. Let $H_{t,k}'$ be obtained by identifying the vertex $(t(k-2),0)$ of $G_{t,k,i}$ with the vertex $(0,0)$ of $G_{t,k,i+1}$ for each $i\in \mathbb Z/k\mathbb Z$ and let $H_{t,k}$ be the graph obtained from $H'_{t,k}$ by adding the edge whose elements are the vertices $(t(k-2),t)$ for each of the $G_{t,k,i}$. For example, $H_{4,3}$ is illustrated in Figure \ref{fig:firtree}. Then $H_{t,k}$ has a filtration of type
\begin{align*}
tk(k-2) &= k\left(\frac{(t+1)(2+t(k-2))}{2}-1\right) - k\frac{t(2+(t-1)(k-2))}{2} - 1 + 1\\
&=|V(H_{t,k})|-|E(H_{t,k})|+1
\end{align*}
and radius $t+1$. Thus, we have another family of hypergraphs which satisfy the preconditions of Theorem~\ref{theorem}, and have a quadratic gap between the radius and foundation size (both $\Theta(t)$) versus the number of vertices and edges (both $\Theta(t^2)$).
\end{example}


\begin{figure}[t]
\begin{center}
\resizebox{380pt}{!}{%
\begin{tikzpicture}
[inner sep=0pt, minimum size = 0pt]
	\coordinate (v0) at (0,0) {};
	\coordinate (v1) at (-4,5) {};
	\coordinate (v2) at (4,5) {};
	\coordinate (v3) at (-4,2) {};
	\coordinate (v4) at (4,2) {};
	\coordinate (v5) at (-7,0) {};
	\coordinate (v6) at (-4,0) {};
	\coordinate (v7) at (4,0) {};
	\coordinate (v8) at (7,0) {};
	\coordinate (v9) at (-9,-2) {};
	\coordinate (v10) at (-7,-2) {};
	\coordinate (v11) at (-4,-2) {};
	\coordinate (v12) at (-1,-2) {};
	\coordinate (v13) at (1,-2) {};
	\coordinate (v14) at (4,-2) {};
	\coordinate (v15) at (7,-2) {};
	\coordinate (v16) at (9,-2) {};

	\foreach \v in {v1,v2,v3,v4,v5,v6,v7,v8,v9,v10,v11,v12,v13,v14,v15,v16}{
\node [circle, minimum size=0.4cm, line width=0pt] (\v') at (\v) {};
}

\filldraw [draw=black, fill=cyan, opacity=0.4]
		(v1'.180) -- (v3'.180) arc (180:270:0.2)
		-- (v4'.-90) arc (-90:90:0.2)
		-- (v3'.45) -- (v1'.0 )arc (0:180:0.2) -- cycle;
		
\filldraw [draw=black, fill=yellow, opacity=0.6]
		(v3'.0) -- (v6'.0) arc (0:-90:0.2)
		-- (v5'.-90) arc (-90:-270:0.2)
		-- (v6'.135) -- (v3'.180 )arc (180:0:0.2) -- cycle;				
\filldraw [draw=black, fill=yellow, opacity=0.6]
		(v4'.180) -- (v7'.180) arc (180:270:0.2)
		-- (v8'.-90) arc (-90:90:0.2)
		-- (v7'.45) -- (v4'.0 )arc (0:180:0.2) -- cycle;		

\filldraw [draw=black, fill=Plum, opacity=0.6]
		(v6'.180) -- (v11'.180) arc (180:270:0.2)
		-- (v12'.-90) arc (-90:90:0.2)
		-- (v11'.45) -- (v6'.0 )arc (0:180:0.2) -- cycle;

\filldraw [draw=black, fill=Plum, opacity=0.6]
		(v8'.180) -- (v15'.180) arc (180:270:0.2)
		-- (v16'.-90) arc (-90:90:0.2)
		-- (v15'.45) -- (v8'.0 )arc (0:180:0.2) -- cycle;	
\filldraw [draw=black, fill=Plum, opacity=0.6]
		(v5'.0) -- (v10'.0) arc (0:-90:0.2)
		-- (v9'.-90) arc (-90:-270:0.2)
		-- (v10'.135) -- (v5'.180 )arc (180:0:0.2) -- cycle;
\filldraw [draw=black, fill=Plum, opacity=0.6]
		(v7'.0) -- (v14'.0) arc (0:-90:0.2)
		-- (v13'.-90) arc (-90:-270:0.2)
		-- (v14'.135) -- (v7'.180 )arc (180:0:0.2) -- cycle;	

\filldraw [draw=black, fill=Maroon, opacity=0.5, line width=0.5pt]
		plot[smooth, tension=1] coordinates { (v1'.210) (0.5,0) (v14'.-90)} arc(-90:90:0.2) --
		plot[smooth, tension=1] coordinates
{ (v14'.90) (2,0) (v2'.-45)} arc(-45:150:0.2)
		plot[smooth, tension=1] coordinates
{ (v2'.150) (0,1.75) (v1'.45)}
		arc(45:210:0.2);
\filldraw [draw=black, fill=BurntOrange, opacity=0.6, line width=0.5pt]
		plot[smooth, tension=1] coordinates { (v1'.210) (-0.5,0) (v13'.-150)} arc(-150:70:0.2) --
		plot[smooth, tension=1] coordinates
{ (v13'.70) (1,0) (v2'.-45)} arc(-45:150:0.2)
		plot[smooth, tension=1] coordinates
{ (v2'.150) (0,1.25) (v1'.45)}
		arc(45:210:0.2);
\filldraw [draw=black, fill=olive, opacity=0.5, line width=0.5pt]
		plot[smooth, tension=1] coordinates { (v1'.210) (-1.75,0) (v12'.-235)} arc(-235:60:0.2) --
		plot[smooth, tension=1] coordinates
{ (v12'.60) (0,0) (v2'.-45)} arc(-45:150:0.2)
		plot[smooth, tension=1] coordinates
{ (v2'.150) (-1,1.5) (v1'.45)}
		arc(45:210:0.2);
\filldraw [draw=black, fill=lime, opacity=0.5, line width=0.5pt]
		plot[smooth, tension=1] coordinates { (v1'.210) (-2.5,0) (v11'.-270)} arc(-270:-60:0.2) --
		plot[smooth, tension=1] coordinates
{ (v11'.-60) (-0.5,0) (v2'.-45)} arc(-45:150:0.2)
		plot[smooth, tension=1] coordinates
{ (v2'.150) (-1,1) (v1'.45)}
		arc(45:210:0.2);

	\filldraw [draw=black, fill=TealBlue, opacity=0.6, line width=0.5pt]
		plot[smooth, tension=1] coordinates { (v2'.-60) (-1,3) (v1'.-90) (-9,2)  (v9'.50)}
		arc(50:-150:0.2) --
		plot[smooth, tension=1] coordinates {(v9'.-150) (-9.5,2.5) (v1'.90) (-1,3.5) (v2'.130) }
		arc(130:-60:0.2)
		;
\filldraw [draw=black, fill=pink, opacity=0.6, line width=0.5pt]
		plot[smooth, tension=1] coordinates { (v2'.-90) (-1,4) (v1'.-90) (-8,1.5)  (v10'.50)}
		arc(50:-150:0.2) --
		plot[smooth, tension=1] coordinates {(v10'.-150) (-8.5,2) (v1'.90) (-1,4.5) (v2'.110) }
		arc(110:-90:0.2)
		;
	
\filldraw [draw=black, fill=yellow, opacity=0.6, line width=0.5pt]
		plot[smooth, tension=1] coordinates { (v1'.-60) (1,6) (v2'.-90) (9,2)  (v16'.-220)}
		arc(-220:-30:0.2) --
		plot[smooth, tension=1] coordinates {(v16'.-30) (9.5,2.5) (v2'.90) (1,6.5) (v1'.120) }
		arc(120:300:0.2)
		;
\filldraw [draw=black, fill=Violet, opacity=0.6, line width=0.5pt]
		plot[smooth, tension=1] coordinates { (v1'.-60) (1,7) (v2'.-90) (8,2)  (v15'.-220)}
		arc(-220:-30:0.2) --
		plot[smooth, tension=1] coordinates {(v15'.-30) (8.5,2.2) (v2'.90) (1,7.5) (v1'.120) }
		arc(120:300:0.2)
		;
		
	\filldraw [black] (v16) circle (2pt) node [inner sep=5pt, label=below:{\Large $222$}] {};	
	\filldraw [black] (v15) circle (2pt) node [inner sep=5pt, label=below:{\Large $221$}] {};
	\filldraw [black] (v14) circle (2pt) node [inner sep=5pt, label=below:{\Large $212$}] {};
	\filldraw [black] (v13) circle (2pt) node [inner sep=5pt, label=below:{\Large $211$}] {};
	\filldraw [black] (v12) circle (2pt) node [inner sep=5pt, label=below:{\Large $122$}] {};
	\filldraw [black] (v11) circle (2pt) node [inner sep=5pt, label=below:{\Large $121$}] {};
	\filldraw [black] (v10) circle (2pt) node [inner sep=5pt, label=below:{\Large $112$}] {};
	\filldraw [black] (v9) circle (2pt) node [inner sep=5pt, label=below:{\Large $111$}] {};
	
	\filldraw [black] (v8) circle (2pt) node [inner sep=5pt, label=below:{\Large $22$}] {};
	\filldraw [black] (v7) circle (2pt) node [inner sep=5pt, label=below:{\Large $21$}] {};
	\filldraw [black] (v6) circle (2pt) node [inner sep=5pt, label=below:{\Large $12$}] {};
	\filldraw [black] (v5) circle (2pt) node [inner sep=5pt, label=below:{\Large $11$}] {};
	
	\filldraw [black] (v4) circle (2pt) node [inner sep=5pt, label=below:{\Large $2$}] {};
	\filldraw [black] (v3) circle (2pt) node [inner sep=5pt, label=below:{\Large $1$}] {};
	
	\filldraw [black] (v2) circle (2pt) node [inner sep=5pt, label=above:{\Large ${(0,2)}$}] {};
	\filldraw [black] (v1) circle (2pt) node [inner sep=5pt, label=above:{\Large ${(0,1)}$}] {};
\end{tikzpicture}
}
\end{center}
\caption{Depiction of $3$-uniform crash hypergraph $C_{3,3}$. Generally, $C_{t,k}$ has an exponential separation between the filtration radius and foundation size versus number of vertices and edges.}
\label{fig:crash}
\end{figure}

\begin{example}[Crash (Figure~\ref{fig:crash}, page $36$)]
 Finally, we give a $k$-uniform hypergraph family with an exponential separation in parameters (Figure~\ref{fig:crash}). For $k \ge 2$ an integer, let $\Sigma = \{1,2,\ldots,k-1\}$ be an alphabet of size $k-1$.
  For $t \ge 1$ an integer, consider the hypergraph $C_{t,k}$ with vertices $V(C_{t,k})=\bigcup_{j=0}^{t} V_j$ where $V_j = \Sigma^{t-j+1}$ for all $1 \le j \le t$, and $V_{0} = \{(0,x)\,|\, x \in \Sigma\}$. The edge set of $C_{t,k}$ is the union of all edges of the following three forms:
\begin{enumerate}
  \item For every $x\in V_1$, $E_x = \{x\} \cup V_0$;
  \item for every $2\le j \le t$ and every $x\in V_j$, $E_x=\{x\}\cup\{xa\,|,a\in \Sigma\}$;
    \item $E_0=\{(0,1)\}\cup V_t$.
    \end{enumerate}
Then $C_{t,k}$ has a transfer filtration with foundation $V_0$ obtained by first adding all the edges $E_x$ with $x\in V_1$, then adding all the edges $E_x$ with $x\in V_2$ etc.\ with $E_0$ added last. Clearly this transfer filtration has radius $t$ and type
\begin{align*}
k-1 & = (k-1)\left(1+\frac{(k-1)^t-1}{k-2}\right) - \frac{(k-1)((k-1)^t-1)}{k-2} -1+1\\
& = |V(C_{t,k})|-|E(C_{t,k})|+1\,.
\end{align*}
In particular, this family satisfies the preconditions for Theorem~\ref{theorem} and yields an {exponential} separation between radius and foundation size ($\Theta(t)$ and $O(k)$, respectively) versus the number of vertices and edges ($\Theta[(k-1)^t]$ for constant $k\geq 3$). Thus, via the runtime of Equation~(\ref{eqn:runtime}), our algorithm runs in polynomial time on crash (i.e. time $\operatorname{poly}(|V(C_{t,k})|,|E(C_{t,k})|)$), whereas brute force diagonalization would require exponential time in $|V(C_{t,k})|$. Intuitively, one may think of crash hypergraphs as a hypergraph-version of trees; a crash hypergraph ``starts small'' (has a small foundation, corresponding to a tree root), and then rapidly expands ``horizontally'', resulting a hypergraph whose ``depth'' is short relative to the number of vertices (just as the depth of a complete binary tree is logarithmic in the size of the tree). Whether such hypergraphs are physically relevant is not clear; here, this family is presented primarily as a proof of concept demonstrating that exponential separations (versus brute force) are indeed possible for our algorithm.
\end{example}


\section{On the structure of 3-uniform hypergraphs}\label{scn:hypergraphs}

In this section, we take a structural graph theory perspective, and make steps towards understanding the set of $3$-uniform hypergraphs $G=(V,E)$ in which each edge is matched to a unique vertex. In the combinatorics literature, such a matching for a set system is called a \emph{System of Distinct Representatives (SDR)}~\cite{SJ011}. Formally, an SDR for a sequence of sets $S_1,\ldots, S_m$ is a sequence of distinct elements $x_1,\ldots,x_m$ such that $x_i\in S_i$ for all $i\in[m]$. Hall's well-known Marriage Theorem says that a set system $\set{S_i}_{i=1}^m$ has an SDR if and only if for any $I\subseteq [m]$, $\abs{\bigcup_{i\in I}S_i}\geq \abs{I}$.

To begin, let us state a known corollary of the Marriage Theorem, which will help foreshadow the complexity involved in attempting to characterize the set of $3$-uniform hypergraphs with SDRs. Let $m=\abs{E}$ and $n=\abs{V}$ for brevity.

\begin{corollary}\label{cor:jukna}[see, e.g.,~Corollary 5.2 of~\cite{SJ011}]
    Let $G=(V,E)$ be a $d$-regular $r$-hypergraph. If $m\leq n$, then the hypergraph has an SDR.
\end{corollary}
\noindent Thus, the class of hypergraphs with SDRs is, in the sense above, not necessarily small. In this section, we focus our attention on the edge case $m=n$.
In this case, any $2$-uniform hypergraph (i.e. a \emph{graph}) must be a cycle (up to attaching disjoint paths to each vertex of the cycle). In contrast, we will see that the set of $3$-uniform hypergraphs with $m=n$ and with SDRs is more complex.

\textbf{Aside.} In Corollary~\ref{cor:jukna}, note via doublecounting that we must have $nd=mr$ (i.e. add up all degrees on the left, and take the union of all edges on the right); thus, if $r=3$ and $m=n$, we must have regularity $d=3$. This yields that {any} $3$-regular $3$-uniform hypergraph with $m=n$ has an SDR.

Finally, using the Marriage Theorem we obtain an elementary proof of a generalization of Corollary \ref{cor:actually}.

\begin{theorem}\label{thm:general}
Let $G$ be a $k$-uniform hypergraph of transfer type $b\le |V(G)|-|E(G)|+k-1$. Then $G$ has an SDR. Moreover, this bound is tight.
\end{theorem}

\begin{proof}
  Consider a transfer filtration $G_0\subseteq G_1\subseteq \cdots \subseteq G_m=G$ so that $G$ has $m$ edges and transfer type $b=|V(G_0)|$. Suppose the edges $E(G)=\{E_1,\ldots E_m\}$ are ordered in such a way that $E(G_i)\setminus E(G_{i-1})=\{E_i\}$ for all $i\in\{1,\ldots,m\}$. By the Marriage Theorem, it suffices to prove that $|E_{i_1}\cup \cdots \cup E_{i_l}|\ge l$ for any increasing sequence $1\le i_1<i_2<\cdots<i_l\le m$. For each $i\in \{1,\ldots,m\}$, let $\mu_i=|V(G_{i-1})|-|V(G_i)|+1$. By definition of transfer filtration, $\mu_i\in \{0,1\}$ for each $i\in \{1,\ldots,m\}$. More precisely, $\mu_i=0$ if and only if $E_i$ contains a vertex that does not belong to any $E_j$ with $j<i$. Hence, given any increasing sequence $1\le i_1<i_2<\cdots<i_l\le m$, we have $|E_{i_1}|=k$, $|E_{i_1}\cup E_{i_2}|\ge k+1-\mu_{i_2}$ and, iteratively adding one edge at the time,
\begin{equation}\label{eq:estimate}
|E_{i_1}\cup E_{i_2}\cup \cdots\cup E_{i_l}|\ge k+l-1-(\mu_{i_2}+\cdots+\mu_{i_l}).
\end{equation}
Since by assumption we have
\begin{equation}\label{eq:bound}
\mu_{i_2}+\cdots+\mu_{i_l}\le \mu_1+\mu_2+\cdots+\mu_m=b-|V(G)|+|E(G)|\le k-1\,.
\end{equation}
Substituting \eqref{eq:bound} into \eqref{eq:estimate} we obtain
\[
|E_{i_1}\cup E_{i_2}\cup \cdots\cup E_{i_l}|\ge k+l-1-(k-1)\ge l
\]
which, by the Marriage Theorem, proves the bound. to prove that the bound is tight, consider a $k$-uniform hypergraph $G$ with $|V(G)|=k$ and $|E(G)|=k+1$ so that the same edge is repeated $k+1$ times. Clearly there exist a transfer filtration of type $k-1>k-2=|V(G)|-|E(G)|+k-1$ and (since $|V(G)|<|E(G)|)$, $G$ has no SDR. Hence the bound is tight.

\end{proof}

\subsection{Intersecting families}\label{sscn:intersect}

We begin with the study of $3$-uniform hypergraphs which are \emph{intersecting families}. Here, a set system $S$ is an \emph{intersecting family} if any pair of sets has non-empty intersection. We say the system is \emph{$k$-intersecting} if any distinct pair of sets $A,B\in S$ satisfies $\abs{A\cap B}=k$.

\subsubsection{$2$-intersecting families}\label{ssscn:twointersect}

For $2$-intersecting families, we actually characterize a \emph{larger} class: The set of $3$-uniform hypergraphs with the weaker condition that \emph{if} a pair of distinct edges $e_i$ and $e_j$ intersect, then either $\abs{e_i\cap e_j}=2$ (Lemma~\ref{lm1}) or $\abs{e_i\cap e_j}\geq 2$ (Lemma~\ref{lm2}). The characterization for $2$-intersecting families follows as an immediate corollary, and is given in Corollary~\ref{cor:2intersect}.

\begin{lemma} \label{lm1}
Let $G=(V,E)$ be a $3$-uniform hypergraph such that:
\begin{enumerate}
	\item \label{l1} $m=n$ and $G$ has an SDR, and
	\item \label{l2} If a pair of distinct edges $e_i$ and $e_j$ intersects, then $\abs{e_i \cap e_j}=2$. 	
\end{enumerate}
Then $G$ is a cycle of length $4$.
\end{lemma}
\begin{proof}
	We proceed by case analysis. We begin with a pair of intersecting edges (WLOG we may assume one such pair exists) as in Figure~\ref{fig:twoedges}:
\begin{figure}[h]
\begin{center}
\vspace{2mm}
\resizebox{150pt}{!}{%
\begin{tikzpicture}
    \coordinate (v1) at (0,0) {};
    \coordinate (v2) at (2,0) {};
    \coordinate (v3) at (4,0){};
    \coordinate (v4) at (6,0) {};
   \foreach \v in {v1,v2,v3,v4,v4}{
      \node [circle, minimum size=0.4cm, line width=0pt] (\v') at (\v) {};
    }
    \filldraw [draw=black, fill=red, minimum size=0pt, opacity=0.2]
            (v1'.270) -- (v3'.270) arc(-90:90:0.2cm) -- (v1'.90) arc (90:270:0.2cm)-- cycle;
	\filldraw [draw=black, fill=blue, minimum size=0pt, opacity=0.2]
            (v2'.270) -- (v4'.270) arc(-90:90:0.2cm) -- (v2'.90) arc (90:270:0.2cm);
		
    \foreach \l in {1,...,4}{
      \filldraw [black] (v\l) circle (2pt) node [label=below:$v_\l$] {};
    }
\end{tikzpicture}
}
\end{center}
\caption{A pair of edges with intersection size $2$.}
\label{fig:twoedges}
\end{figure}

\noindent We now analyze how one can ``grow'' the graph by adding hyperedges. Consider first the formation of a tight star as in Figure~\ref{fig:hyperstar}.
\begin{figure}[h]
\begin{center}
\resizebox{150pt}{!}{%
\begin{tikzpicture}
    \coordinate (v1) at (0,0) {};
    \coordinate (v2) at (2,0) {};
    \coordinate (v3) at (4,0) {};
    \coordinate (v4) at (6,0) {};
    \coordinate (v5) at (3,-2){};
    \foreach \v in {v1,v2,v3,v4,v5}{
      \node [circle, minimum size=0.4cm, line width=0pt] (\v') at (\v) {};
    }
    \filldraw [draw=black, fill=red, opacity=0.2]
		(v1'.270) -- (v3'.270) arc (-90:90:0.2)
		-- (v2'.90) arc (90:90:0.2)
		-- (v1'.90) arc (90:270:0.2) -- cycle;
		
    \filldraw [draw=black, fill=blue, opacity=0.2]
		(v2'.90) -- (v4'.90) arc (90:-90:0.2)
		-- (v3'.-90) arc (-90:-90:0.2)
		-- (v2'.-90) arc (-90:-270:0.2) -- cycle;
		
    \filldraw [draw=black, fill=green, opacity=0.2]
		(v3'.-30) -- (v5'.-30) arc (-30:-150:0.2)
		-- (v2'.210) arc (210:90:0.2)
		-- (v3'.90) arc (90:-30:0.2) -- cycle;
		
    \foreach \l in {1,...,5}{
      \filldraw [black] (v\l) circle (2pt) node [label=below:$v_\l$] {};
    }
\end{tikzpicture}
}
\caption{A $3$-uniform hypergraph which is a tight star of size $3$.}
\label{fig:hyperstar}
\end{center}
\end{figure}
We argue that this structure is impossible, given conditions $\ref{l1}$ and $\ref{l2}$. To see why, denote vertices $\{v_1,v_4,v_5\}$ as ``outer'' vertices and $\{v_2,v_3\}$ as ``inner'' vertices. By symmetry, there are seven options for adding a new edge to the tight star:
\begin{enumerate}[label=(\roman*)]
	\item Add an edge containing a {new} vertex $v_6$, an {outer} vertex, and an {inner} vertex, e.g. $\{v_1,v_2,v_6\}$. This violates property $\ref{l2}$ since $\abs{\{v_1,v_2,v_6\}\cap \{v_2,v_3,v_4\}}=1$.

\begin{center}
\resizebox{150pt}{!}{%
\begin{tikzpicture}
    \coordinate (v1) at (0,0) {};
    \coordinate (v2) at (2,0) {};
    \coordinate (v3) at (4,0) {};
    \coordinate (v4) at (6,0) {};
    \coordinate (v5) at (3,-2){};
	\coordinate (v6) at (1,-2){};
    \foreach \v in {v1,v2,v3,v4,v5,v6}{
      \node [circle, minimum size=0.4cm, line width=0pt] (\v') at (\v) {};
    }
    \filldraw [draw=black, fill=red, opacity=0.2]
		(v1'.270) -- (v3'.270) arc (-90:90:0.2)
		-- (v2'.90) arc (90:90:0.2)
		-- (v1'.90) arc (90:270:0.2) -- cycle;
		
    \filldraw [draw=black, fill=blue, opacity=0.2]
		(v2'.90) -- (v4'.90) arc (90:-90:0.2)
		-- (v3'.-90) arc (-90:-90:0.2)
		-- (v2'.-90) arc (-90:-270:0.2) -- cycle;
		
    \filldraw [draw=black, fill=green, opacity=0.2]
		(v3'.-30) -- (v5'.-30) arc (-30:-150:0.2)
		-- (v2'.210) arc (210:90:0.2)
		-- (v3'.90) arc (90:-30:0.2) -- cycle;
	
	\filldraw [draw=black, fill=yellow, opacity=0.2]
		(v2'.-30) -- (v6'.-30) arc (-30:-150:0.2)
		-- (v1'.210) arc (210:90:0.2)
		-- (v2'.90) arc (90:-30:0.2) -- cycle;
		
    \foreach \l in {1,...,6}{
      \filldraw [black] (v\l) circle (2pt) node [label=below:$v_\l$] {};
    }
\end{tikzpicture}
}
\end{center}
	\vspace{-2mm}
	\item Add an edge containing one {new} vertex $v_6$, and two {outer} vertices, e.g. $\{v_1,v_5,v_6\}$. This violates property $\ref{l2}$, since $\abs{\{v_1,v_5,v_6\}\cap \{v_2,v_3,v_5\}}=1$.
	\begin{center}
\resizebox{150pt}{!}{%
	\begin{tikzpicture}
    \coordinate (v1) at (0,0) {};
    \coordinate (v2) at (2,0) {};
    \coordinate (v3) at (4,0) {};
    \coordinate (v4) at (6,0) {};
    \coordinate (v5) at (3,-2) {};
	\coordinate (v6) at (1,-2) {};
    \foreach \i in {1,...,6}{
      \node [circle, minimum size=0.4cm, line width=0pt] (v\i') at (v\i) {};
    }
    \filldraw [draw=black, fill=red, opacity=0.2]
		(v1'.270) -- (v3'.270) arc (-90:90:0.2)
		-- (v1'.90) arc (90:270:0.2) -- cycle;
		
    \filldraw [draw=black, fill=blue, opacity=0.2]
		(v2'.90) -- (v4'.90) arc (90:-90:0.2)
		-- (v2'.-90) arc (-90:-270:0.2) -- cycle;
		
    \filldraw [draw=black, fill=green, opacity=0.2]
		(v3'.-30) -- (v5'.-30) arc (-30:-150:0.2)
		-- (v2'.210) arc (210:90:0.2)
		-- (v3'.90) arc (90:-30:0.2) -- cycle;
	
	\filldraw [draw=black, fill=yellow, opacity=0.2]
		(v1'.60) -- (v5'.60) arc (60:-90:0.2)
		-- (v6'.-90) arc (270:210:0.2)
		-- (v1'.210) arc (210:60:0.2);
		
    \foreach \l in {1,...,6}{
      \filldraw [black] (v\l) circle (2pt) node [label=below:$v_\l$] {};
    }
\end{tikzpicture}
}
\end{center}
	\item \label{c3} Add an edge containing one {new} vertex $v_6$, two {inner} vertices, e.g. $\set{v_2,v_3,v_6}$, which is the only case not violating property $\ref{l2}$.
\begin{center}
\resizebox{150pt}{!}{%
\begin{tikzpicture}
    \coordinate (v1) at (0,0) {};
    \coordinate (v2) at (2,0) {};
    \coordinate (v3) at (4,0) {};
    \coordinate (v4) at (6,0) {};
    \coordinate (v5) at (3,-2){};
	\coordinate (v6) at (3,2) {};
    \foreach \v in {v1,v2,v3,v4,v5,v6}{
      \node [circle, minimum size=0.4cm, line width=0pt] (\v') at (\v) {};
    }
    \filldraw [draw=black, fill=red, opacity=0.2]
		(v1'.270) -- (v3'.270) arc (-90:90:0.2)
		-- (v2'.90) arc (90:90:0.2)
		-- (v1'.90) arc (90:270:0.2) -- cycle;
		
    \filldraw [draw=black, fill=blue, opacity=0.2]
		(v2'.90) -- (v4'.90) arc (90:-90:0.2)
		-- (v3'.-90) arc (-90:-90:0.2)
		-- (v2'.-90) arc (-90:-270:0.2) -- cycle;
		
    \filldraw [draw=black, fill=green, opacity=0.2]
		(v3'.-30) -- (v5'.-30) arc (-30:-150:0.2)
		-- (v2'.210) arc (210:90:0.2)
		-- (v3'.90) arc (90:-30:0.2) -- cycle;
	
	\filldraw [draw=black, fill=yellow, opacity=0.2]
		(v3'.30) -- (v6'.30) arc (30:150:0.2)
		-- (v2'.150) arc (150:270:0.2)
		-- (v3'.270) arc (-90:30:0.2) -- cycle;
		
    \foreach \l in {1,...,6}{
      \filldraw [black] (v\l) circle (2pt) node [label=below:$v_\l$] {};
    }
\end{tikzpicture}
}
\end{center}
\vspace{-2mm}

		\item Add an edge containing two {outer} vertices, and one {inner} vertex, e.g. $\set{v_1,v_2,v_5}$, which violates property $\ref{l2}$, since $\abs{\{v_1,v_2,v_5\}\cap \{v_2,v_3,v_4\}}=1$.
		\begin{center}
\resizebox{150pt}{!}{%
		\begin{tikzpicture}
	    \coordinate (v1) at (0,0) {};
	    \coordinate (v2) at (2,0) {};
	    \coordinate (v3) at (4,0) {};
	    \coordinate (v4) at (6,0) {};
	    \coordinate (v5) at (3,-2) {};
	    \foreach \i in {1,...,5}{
	      \node [circle, minimum size=0.4cm, line width=0pt] (v\i') at (v\i) {};
	    }
	    \filldraw [draw=black, fill=red, opacity=0.2]
			(v1'.270) -- (v3'.270) arc (-90:90:0.2)
			-- (v1'.90) arc (90:270:0.2) -- cycle;
			
	    \filldraw [draw=black, fill=blue, opacity=0.2]
			(v2'.90) -- (v4'.90) arc (90:-90:0.2)
			-- (v2'.-90) arc (-90:-270:0.2) -- cycle;
			
	    \filldraw [draw=black, fill=green, opacity=0.2]
			(v3'.-30) -- (v5'.-30) arc (-30:-150:0.2)
			-- (v2'.210) arc (210:90:0.2)
			-- (v3'.90) arc (90:-30:0.2) -- cycle;
		
		\filldraw [draw=black, fill=yellow, opacity=0.2]
			(v1'.90) -- (v2'.90) arc (90:30:0.2)
			-- (v5'.30) arc (30:-120:0.2)
			-- (v1'.-120) arc (240:90:0.2);
			
	    \foreach \l in {1,...,5}{
	      \filldraw [black] (v\l) circle (2pt) node [label=below:$v_\l$] {};
	    }
	\end{tikzpicture}
}
	\end{center}
	
	\item \label{c5} Add an edge containing two {inner} vertices and one {outer} vertex, e.g. $\set{v_2,v_3,v_4}$, which clearly violates property $\ref{l2}$.
	\begin{center}
\resizebox{150pt}{!}{%
	\begin{tikzpicture}
    \coordinate (v1) at (0,0) {};
    \coordinate (v2) at (2,0) {};
    \coordinate (v3) at (4,0) {};
    \coordinate (v4) at (6,0) {};
    \coordinate (v5) at (3,-2){};
    \foreach \v in {v1,v2,v3,v4,v5}{
      \node [circle, minimum size=0.4cm, line width=0pt] (\v') at (\v) {};
    }
	\foreach \v in {v1,v2,v3,v4,v5}{
      \node [circle, minimum width=1cm] (\v'') at (\v) {};
    }
	
	\foreach \v in {v1,v2,v3,v4,v5}{
      \node [circle, minimum width=.74cm] (\v''') at (\v) {};
    }
	
	\filldraw [draw=black, fill=yellow, opacity=0.2]
		(v1'''.270) -- (v3'''.270) arc (-90:90:0.37)
		-- (v1'''.90) arc (90:270:0.37) -- cycle;
		
    \filldraw [draw=black, fill=blue, opacity=0.5]
		(v2'.90) -- (v4'.90) arc (90:-90:0.2)
		-- (v2'.-90) arc (-90:-270:0.2) -- cycle;
	
	\filldraw [draw=black, fill=pink, opacity=0.5]
		(v2''.90) -- (v4''.90) arc (90:-90:0.5)
		-- (v2''.270) arc (270:90:0.5) -- cycle;	
		
    \filldraw [draw=black, fill=green, opacity=0.2]
		(v3'.-30) -- (v5'.-30) arc (-30:-150:0.2)
		-- (v2'.210) arc (210:90:0.2)
		-- (v3'.90) arc (90:-30:0.2) -- cycle;
		
    \foreach \l in {1,...,5}{
      \filldraw [black] (v\l) circle (2pt) node [label=below:$v_\l$] {};
    }
	\end{tikzpicture}
}
	\end{center}
	\item Add an edge containing three {outer} vertices, e.g. $\set{v_1,v_5,v_4}$, which violates property $\ref{l2}$ since  $\abs{\{v_1,v_5,v_4\}\cap \{v_2,v_3,v_4\}}=1$.
\begin{center}
\resizebox{150pt}{!}{%
\begin{tikzpicture}
    \coordinate (v1) at (0,0) {};
    \coordinate (v2) at (2,0) {};
    \coordinate (v3) at (4,0) {};
    \coordinate (v4) at (6,0) {};
    \coordinate (v5) at (3,-2){};
    \foreach \v in {v1,v2,v3,v4,v5}{
      \node [circle, minimum size=0.4cm, line width=0pt] (\v') at (\v) {};
    }
    \filldraw [draw=black, fill=red, opacity=0.2]
		(v1'.270) -- (v3'.270) arc (-90:90:0.2)
		-- (v2'.90) arc (90:90:0.2)
		-- (v1'.90) arc (90:270:0.2) -- cycle;
		
    \filldraw [draw=black, fill=blue, opacity=0.2]
		(v2'.90) -- (v4'.90) arc (90:-90:0.2)
		-- (v3'.-90) arc (-90:-90:0.2)
		-- (v2'.-90) arc (-90:-270:0.2) -- cycle;
		
    \filldraw [draw=black, fill=green, opacity=0.2]
		(v3'.-30) -- (v5'.-30) arc (-30:-150:0.2)
		-- (v2'.210) arc (210:90:0.2)
		-- (v3'.90) arc (90:-30:0.2) -- cycle;
	
	\filldraw [draw=black, fill=yellow, opacity=0.2]
		(v1'.45) -- (v5'.120) arc(-120:-60:0.2)
		-- (v4'.120) arc(120:-60:0.2)
		-- (v5'.300) arc(300:250:0.2)
		-- (v1'.250) arc(250:45:0.2);

    \foreach \l in {1,...,5}{
      \filldraw [black] (v\l) circle (2pt) node [label=below:$v_\l$] {};
    }
\end{tikzpicture}
}
\end{center}
    \item Add two new vertices, and one {inner} or {outer} vertex. This violates property \ref{l2}.
\end{enumerate}
We conclude that the only possibility is case \ref{c3}, which grows the tight star. In other words, once we have a tight star, our only option is to continue to build a larger tight star (with distinct edges). But a tight star cannot satisfy property~\ref{l1}, since it has $n=m+2$. Thus, $G$ cannot contain a tight star (consisting of more than $2$ edges).

Let us continue our case analysis of how Figure~\ref{fig:twoedges} can be extended:
\begin{enumerate}[label=(\alph*)]
	\item Add a {new} vertex and two {outer} vertices, which violates property $\ref{l2}$.
\begin{center}
\resizebox{150pt}{!}{%
\begin{tikzpicture}
    \coordinate (v1) at (0,0);
    \coordinate (v2) at (2,0) {};
    \coordinate (v3) at (4,0);
    \coordinate (v4) at (6,0) {};
    \coordinate (v5) at (3,-2) {};
    \foreach \v in {v1,v2,v3,v4,v5}{
      \node [circle, minimum size=0.4cm, line width=0pt] (\v') at (\v) {};
    }
    \filldraw [draw=black, fill=red, opacity=0.2]
		(v1'.270) -- (v3'.270) arc (-90:90:0.2)
		-- (v2'.90) arc (90:90:0.2)
		-- (v1'.90) arc (90:270:0.2) -- cycle;
		
    \filldraw [draw=black, fill=blue, opacity=0.2]
		(v2'.90) -- (v4'.90) arc (90:-90:0.2)
		-- (v3'.-90) arc (-90:-90:0.2)
		-- (v2'.-90) arc (-90:-270:0.2) -- cycle;
	
	\filldraw [draw=black, fill=yellow, opacity=0.2]
		(v1'.45) -- (v5'.120) arc(-120:-60:0.2)
		-- (v4'.120) arc(120:-60:0.2)
		-- (v5'.300) arc(300:250:0.2)
		-- (v1'.250) arc(250:45:0.2);
    \foreach \l in {1,...,5}{
      \filldraw [black] (v\l) circle (2pt) node [label=below:$v_\l$] {};
    }
\end{tikzpicture}
}
\end{center}

	\item Add a {new} vertex, one {inner} vertex, and one {outer} vertex, which violates property $\ref{l2}$, e.g. in:
	\begin{center}
\resizebox{150pt}{!}{%
\begin{tikzpicture}
    \coordinate (v1) at (0,0) {};
    \coordinate (v2) at (2,0) {};
    \coordinate (v3) at (4,0) {};
    \coordinate (v4) at (6,0) {};
    \coordinate (v5) at (8,0) {};
    \foreach \v in {v1,v2,v3,v4,v5}{
      \node [circle, minimum size=0.4cm, line width=0pt] (\v') at (\v) {};
    }
    \filldraw [draw=black, fill=red, opacity=0.2]
		(v1'.270) -- (v3'.270) arc (-90:90:0.2)
		-- (v2'.90) arc (90:90:0.2)
		-- (v1'.90) arc (90:270:0.2) -- cycle;
		
    \filldraw [draw=black, fill=blue, opacity=0.2]
		(v2'.90) -- (v4'.90) arc (90:-90:0.2)
		-- (v3'.-90) arc (-90:-90:0.2)
		-- (v2'.-90) arc (-90:-270:0.2) -- cycle;
	\filldraw [draw=black, fill=green, opacity=0.2]
		(v3'.90) -- (v5'.90) arc (90:-90:0.2)
		-- (v4'.-90) arc (-90:-90:0.2)
		-- (v3'.-90) arc (-90:-270:0.2) -- cycle;
		
    \foreach \l in {1,...,5}{
      \filldraw [black] (v\l) circle (2pt) node [label=below:$v_\l$] {};
    }
\end{tikzpicture}
}
\end{center}

	\item The only remaining option is to add an edge such as the blue one below. This temporarily violates property $\ref{l1}$, since $m\neq n$.
\begin{center}
\resizebox{150pt}{!}{%
\begin{tikzpicture}
    \coordinate (v1) at (0,0) {};
    \coordinate (v2) at (2,0) {};
    \coordinate (v3) at (4,0) {};
    \coordinate (v4) at (6,0) {};    \foreach \v in {v1,v2,v3,v4}{
      \node [circle, minimum size=0.4cm, line width=0pt] (\v') at (\v) {};
    }
    \filldraw [draw=black, fill=red, opacity=0.2]
		(v1'.270) -- (v3'.270) arc (-90:90:0.2)
		-- (v2'.90) arc (90:90:0.2)
		-- (v1'.90) arc (90:270:0.2) -- cycle;
		
    \filldraw [draw=black, fill=blue, opacity=0.2]
		(v2'.90) -- (v4'.90) arc (90:-90:0.2)
		-- (v3'.-90) arc (-90:-90:0.2)
		-- (v2'.-90) arc (-90:-270:0.2) -- cycle;
	
	\filldraw [draw=black, fill=cyan, opacity=0.2, xshift=0cm]
	plot [smooth, tension=2] coordinates {(v1'.180) (3,1.2) (v4'.0)}
	arc(0: -90:0.2) -- (v3'.270) arc(270:180:0.2)
	plot [smooth, tension=2] coordinates {(v3'.180) (2,0.8) (v1'.0)} arc(0:-180:0.2);
    \foreach \l in {1,...,4}{
      \filldraw [black] (v\l) circle (2pt) node [label=below:$v_\l$] {};
    }
\end{tikzpicture}
}
\end{center}
	An analogous case analysis shows that the only way to proceed is to add the yellow edge below, which yields a cycle of length $4$, and which satisfies properties \ref{l1} and \ref{l2}. No further edges can be added, completing the proof.
\begin{center}
\resizebox{150pt}{!}{%
\begin{tikzpicture}
    \coordinate (v1) at (0,0) {};
    \coordinate (v2) at (2,0) {};
    \coordinate (v3) at (4,0) {};
    \coordinate (v4) at (6,0) {};    \foreach \v in {v1,v2,v3,v4}{
      \node [circle, minimum size=0.4cm, line width=0pt] (\v') at (\v) {};
    }
    \filldraw [draw=black, fill=red, opacity=0.2]
		(v1'.270) -- (v3'.270) arc (-90:90:0.2)
		-- (v2'.90) arc (90:90:0.2)
		-- (v1'.90) arc (90:270:0.2) -- cycle;
		
    \filldraw [draw=black, fill=blue, opacity=0.2]
		(v2'.90) -- (v4'.90) arc (90:-90:0.2)
		-- (v3'.-90) arc (-90:-90:0.2)
		-- (v2'.-90) arc (-90:-270:0.2) -- cycle;
	
	\filldraw [draw=black, fill=cyan, opacity=0.2]
	plot [smooth, tension=2] coordinates {(v1'.180) (3,1.2) (v4'.0)}
	arc(0: -90:0.2) -- (v3'.270) arc(270:180:0.2)
	plot [smooth, tension=2] coordinates {(v3'.180) (2,0.8) (v1'.0)} arc(0:-180:0.2);
	
	\filldraw [draw=black, fill=yellow, opacity=0.2]
	plot [smooth, tension=2] coordinates {(v1'.180) (3,-1.2) (v4'.0)}
	arc(0:180:0.2)
	plot [smooth, tension=2] coordinates {(v4'.180) (4,-0.8) (v2'.0)}
	arc(0:90:0.2) -- (v1'.90) arc(90:180:0.2)
	;
    \foreach \l in {1,...,4}{
      \filldraw [black] (v\l) circle (2pt) node [label=below:$v_\l$] {};
    }
\end{tikzpicture}
}
\end{center}

\end{enumerate}

\end{proof}

Lemma~\ref{lm1} yields the following immediate corollary, obtained by observing that the cycle of length $4$ is in fact a $2$-intersecting family.

\begin{corollary} \label{cor:2intersect}
The unique $3$-uniform hypergraph which (1) is a $2$-intersecting family, (2) has $m=n$, and (3) has an SDR is the cycle of length $4$.
\end{corollary}

More generally, we can relax the statement of Lemma~\ref{lm1} to enforce only $\abs{e_i\cap e_j}\geq 2$ (as opposed to $\abs{e_i\cap e_j}= 2$) and still obtain a characterization. Note that although the statement of Lemma~\ref{lm2} below does not require $G$ to \emph{a priori} be an intersecting family, the graphs obtained in the characterization are indeed intersecting families (but not necessarily $k$-intersecting for some fixed $k$).

\begin{lemma}\label{lm2}
Let $G=(V,E)$ be a $3$-uniform hypergraph with property $\ref{l1}$ in Lemma \ref{lm1}, and assume each pair of edges $e_i$ and $e_j$ which intersects satisfies $\abs{e_i \cap e_j}\geq 2$. Then $G$ is one of the following:
\begin{enumerate}[label=(\roman*)]
	\item A 3-stacked set, or
	\item a tight star with two 2-stacked sets, or
	\item a tight star with one 3-stacked set, or
	\item a cycle of length $4$.
\end{enumerate}

\end{lemma}
\begin{proof}
Case (i) trivially holds, and case (iv) is given by Lemma~\ref{lm2}. For cases (ii) and (iii), we return to the case analysis in the proof of Lemma~\ref{lm1} involving Figure~\ref{fig:hyperstar}, i.e. the tight star. Following the proof of Lemma~\ref{lm1}, once we have a tight star, we only have two choices each time we add an edge: (1) Grow the tight star (case (iii) in Lemma~\ref{lm1}), or (2) create a parallel edge (case (v) in Lemma~\ref{lm1}); note this latter case is now possible since we can have $\abs{e_i\cap e_j}=3$. The former choice adds a new vertex and new edge, preserving the invariant $n-m$. The latter choice keeps $n$ fixed but increments $m$ by $1$, implying $n-m$ decreases by $1$. It follows that one must add precisely two parallel edges, yielding cases (ii) and (iii). Finally, analyzing cases (a), (b), and (c) in Lemma~\ref{lm1} yields that no other choices of $G$ are possible.
\end{proof}

\subsubsection{$1$-intersecting families}

We next consider the set of $3$-uniform hypergraphs in which each pair of edges intersects in precisely one vertex. In this case, it turns out there is a unique $3$-uniform hypergraph with $m=n$ and an SDR --- the well-known \emph{Fano plane}.

\begin{definition}[Fano plane (see, e.g.,~\cite{SJ011})]
The Fano plane is the unique projective plane of order $2$, and can be represented by the \emph{Fano hypergraph} $G=(V,E)$ (see Figure~\ref{fig:fano}) with ${V}=\set{v_1,v_2,v_3,v_4,v_5,v_6,v_7}$ and
\[
    E=\set{\{v_1,v_2,v_3\},\{v_1,v_4,v_5\},\{v_1,v_6,v_7\},\{v_2,v_4,v_6\},\{v_2,v_5,v_7\},\{v_3,v_4,v_7\}, \{v_3,v_5,v_6\}}.
\]
\end{definition}
\begin{figure}[t]
\begin{center}
\resizebox{90pt}{!}{%
	\begin{tikzpicture}
        [inner sep=0pt]
        \coordinate (v1) at (210:2cm);
        \coordinate (v2) at (270:1cm);
        \coordinate (v3) at (330:2cm);
        \coordinate (v4) at (150:1cm);
        \coordinate (v5) at (90:2cm);
        \coordinate (v6) at (30:1cm);
        \coordinate (v7) at (0,0);
        \foreach \i in {1,...,7}{
            \node [circle, font=\tiny, minimum size=0.4cm, line width=0mm] (v\i') at (v\i) {};
        }
        \filldraw [draw=black, fill=yellow, minimum size=0pt, opacity=0.2]
            (v5'.150) -- (v1'.150) arc (-210:-30:0.2cm) -- (v5'.330) arc (-30:150:0.2cm);

        \filldraw [draw=black, fill=blue, minimum size=0pt, opacity=0.2]
            (v1'.270) -- (v3'.270) arc(-90:90:0.2cm) -- (v1'.90) arc (90:270:0.2cm);

        \filldraw [draw=black, fill=red, minimum size=0pt, opacity=0.2]
            (v3'.30) -- (v5'.30) arc(30:210:0.2cm) -- (v3'.210) arc(-150:30:0.2cm);

        \filldraw [draw=black, fill=brown, minimum size=0pt, opacity=0.2]
            (v5'.180) -- (v2'.180) arc(-180:0:0.2cm) -- (v5'.0) arc(0:180:0.2cm);
        \filldraw [draw=black, fill=cyan, minimum size=0pt, opacity=0.2]
            (v4'.240) -- (v3'.240) arc(-120:60:0.2cm) -- (v4'.60) arc(60:240:0.2cm);

        \filldraw [draw=black, fill=orange, minimum size=0pt, opacity=0.2]
            (v6'.120) -- (v1'.120) arc(120:300:0.2cm) -- (v6'.300) arc(-60:120:0.2cm) ;

        \filldraw [draw=black, fill=green, minimum size=0pt, opacity=0.2, even odd rule]
            circle (1.2cm) (0:0.8) arc (0:360:0.8cm);
       \foreach \l in {1,...,7}{
      \filldraw [black] (v\l) circle (2pt) node [inner sep=5pt, label=below:$v_{\l}$] {};
    }
    \end{tikzpicture}
    }
    \caption{The Fano hypergraph.}
    \label{fig:fano}
\end{center}
\end{figure}

\begin{lemma}\label{l:1intersect}
The unique $3$-uniform hypergraph which is (1) a $1$-intersecting family, (2) has $m=n$, and (3) has an SDR is the Fano hypergraph.	
\end{lemma}
\begin{proof}
	Let $G=(V,E)$ satisfy the properties of the claim. WLOG, we begin with an arbitrary pair of intersecting edges, $e_1=\set{v_1,v_2,v_3}$ and $e_2=\set{v_1,v_4,v_5}$, pictured below:
\begin{center}
\resizebox{90pt}{!}{%
	\begin{tikzpicture}
        [inner sep=0pt]
        \coordinate (v1) at (210:2cm);
        \coordinate (v2) at (270:1cm);
        \coordinate (v3) at (330:2cm);
        \coordinate (v4) at (150:1cm);
        \coordinate (v5) at (90:2cm);
        \foreach \i in {1,...,5}{
            \node [circle, font=\tiny, minimum size=0.4cm, line width=0mm] (v\i') at (v\i) {};
        }

        \filldraw [draw=black, fill=yellow, minimum size=0pt, opacity=0.2]
            (v5'.150) -- (v1'.150) arc (-210:-30:0.2cm) -- (v5'.330) arc (-30:150:0.2cm);

        \filldraw [draw=black, fill=blue, minimum size=0pt, opacity=0.2]
            (v1'.270) -- (v3'.270) arc(-90:90:0.2cm) -- (v1'.90) arc (90:270:0.2cm);
       \foreach \l in {1,...,5}{
      \filldraw [black] (v\l) circle (2pt) node [inner sep=5pt, label=below:$v_{\l}$] {};
    }
    \end{tikzpicture}
    }
\end{center}
We proceed by adding edges one at a time while maintaining property $(1)$. To add the third edge $e_3$, since the graph is $1$-intersecting, there are two cases:
\begin{itemize}
    \item [(a)]$e_3=\set{v_1,v_6,v_7}$ intersects both $e_1$ and $e_2$ at vertex $v_1$.
    	\begin{center}
    \resizebox{90pt}{!}{%
	\begin{tikzpicture}
        [inner sep=0pt]
        \coordinate (v1) at (210:2cm);
        \coordinate (v2) at (270:1cm);
        \coordinate (v3) at (330:2cm);
        \coordinate (v4) at (150:1cm);
        \coordinate (v5) at (90:2cm);
        \coordinate (v6) at (30:1cm);
        \coordinate (v7) at (0,0);
        \foreach \i in {1,...,7}{
            \node [circle, font=\tiny, minimum size=0.4cm, line width=0mm] (v\i') at (v\i) {};
        }

        \filldraw [draw=black, fill=yellow, minimum size=0pt, opacity=0.2]
            (v5'.150) -- (v1'.150) arc (-210:-30:0.2cm) -- (v5'.330) arc (-30:150:0.2cm);

        \filldraw [draw=black, fill=blue, minimum size=0pt, opacity=0.2]
            (v1'.270) -- (v3'.270) arc(-90:90:0.2cm) -- (v1'.90) arc (90:270:0.2cm);

        \filldraw [draw=black, fill=orange, minimum size=0pt, opacity=0.2]
            (v6'.120) -- (v1'.120) arc(120:300:0.2cm) -- (v6'.300) arc(-60:120:0.2cm) ;
       \foreach \l in {1,...,7}{
      \filldraw [black] (v\l) circle (2pt) node [inner sep=5pt, label=below:$v_{\l}$] {};
    }

    \end{tikzpicture}
    }
\end{center}

		\item [(b)] $e_3=\set{v_3,v_5,v_6}$ does not intersect $e_1$ and $e_2$ at $v_1$. By symmetry, we may represent this as below.
\begin{figure}[h]
		\begin{center}
\resizebox{90pt}{!}{%
	\begin{tikzpicture}
        [inner sep=0pt]
        \coordinate (v1) at (210:2cm);
        \coordinate (v2) at (270:1cm);
        \coordinate (v3) at (330:2cm);
        \coordinate (v4) at (150:1cm);
        \coordinate (v5) at (90:2cm);
        \coordinate (v6) at (30:1cm);
        \foreach \i in {1,...,6}{
            \node [circle, font=\tiny, minimum size=0.4cm, line width=0mm] (v\i') at (v\i) {};
        }

        \filldraw [draw=black, fill=yellow, minimum size=0pt, opacity=0.2]
            (v5'.150) -- (v1'.150) arc (-210:-30:0.2cm) -- (v5'.330) arc (-30:150:0.2cm);

        \filldraw [draw=black, fill=blue, minimum size=0pt, opacity=0.2]
            (v1'.270) -- (v3'.270) arc(-90:90:0.2cm) -- (v1'.90) arc (90:270:0.2cm);

        \filldraw [draw=black, fill=red, minimum size=0pt, opacity=0.2]
            (v3'.30) -- (v5'.30) arc(30:210:0.2cm) -- (v3'.210) arc(-150:30:0.2cm);
       \foreach \l in {1,...,6}{
      \filldraw [black] (v\l) circle (2pt) node [inner sep=5pt, label=below:$v_{\l}$] {};
    }
    \end{tikzpicture}
    }
\end{center}
\caption{\vspace{-4mm}A (linear) $3$-cycle.}
\label{fig:nohelly}
\end{figure}
\end{itemize}
Let us now add the fourth edge, which can be done in one of two ways: Either we add a new vertex, or we do not add a new vertex.
	\begin{itemize}

	\item If we add a new vertex: For case (a), we obtain a graph isomorphic to
\begin{center}
\resizebox{90pt}{!}{%
\begin{tikzpicture}
        [inner sep=0pt]
        \coordinate (v1) at (210:2cm);
        \coordinate (v2) at (270:1cm);
        \coordinate (v3) at (330:2cm);
        \coordinate (v4) at (150:1cm);
        \coordinate (v5) at (90:2cm);
        \coordinate (v6) at (30:1cm);
        \coordinate (v7) at (0,0);
        \path (v1) ++(90:1.73cm) node (v8) {};
        \path (v1) ++(90:3.46cm) node (v9) {};
        \foreach \i in {1,...,9}{
            \node [circle, font=\tiny, minimum size=0.4cm, line width=0mm] (v\i') at (v\i) {};
        }

        \filldraw [draw=black, fill=yellow, minimum size=0pt, opacity=0.2]
            (v5'.150) -- (v1'.150) arc (-210:-30:0.2cm) -- (v5'.330) arc (-30:150:0.2cm);

        \filldraw [draw=black, fill=blue, minimum size=0pt, opacity=0.2]
            (v1'.270) -- (v3'.270) arc(-90:90:0.2cm) -- (v1'.90) arc (90:270:0.2cm);

        \filldraw [draw=black, fill=orange, minimum size=0pt, opacity=0.2]
            (v6'.120) -- (v1'.120) arc(120:300:0.2cm) -- (v6'.300) arc(-60:120:0.2cm) ;

		\filldraw [draw=black, fill=red, minimum size=0pt, opacity=0.2]
            (v1'.0) -- (v9'.0) arc(0:180:0.2cm) -- (v1'.180) arc (-180:0:0.2cm);
       \foreach \l in {1,...,9}{
      \filldraw [black] (v\l) circle (2pt) node [inner sep=5pt, label=below:$v_{\l}$] {};
    }
    \end{tikzpicture}
    }
\end{center}
    This does not have $m=n$ and cannot be extended into a $1$-intersecting family since any further edge (of size $3$) would have to intersect $4$ existing edges. For (b), a case analysis reveals that the only way to extend by adding a new vertex while maintaining the $1$-intersection property is via a graph isomorphic to:
    \begin{figure}[h]
    \begin{center}
\resizebox{90pt}{!}{%
	
\begin{tikzpicture}
        [inner sep=0pt]
        \coordinate (v1) at (210:2cm);
        \coordinate (v2) at (270:1cm);
        \coordinate (v3) at (330:2cm);
        \coordinate (v4) at (150:1cm);
        \coordinate (v5) at (90:2cm);
        \coordinate (v6) at (30:1cm);
        \coordinate (v7) at (0,0);
        \foreach \i in {1,...,7}{
            \node [circle, font=\tiny, minimum size=0.4cm, line width=0mm] (v\i') at (v\i) {};
        }

        \filldraw [draw=black, fill=yellow, minimum size=0pt, opacity=0.2]
            (v5'.150) -- (v1'.150) arc (-210:-30:0.2cm) -- (v5'.330) arc (-30:150:0.2cm);

        \filldraw [draw=black, fill=blue, minimum size=0pt, opacity=0.2]
            (v1'.270) -- (v3'.270) arc(-90:90:0.2cm) -- (v1'.90) arc (90:270:0.2cm);

        \filldraw [draw=black, fill=red, minimum size=0pt, opacity=0.2]
            (v3'.30) -- (v5'.30) arc(30:210:0.2cm) -- (v3'.210) arc(-150:30:0.2cm);

        \filldraw [draw=black, fill=orange, minimum size=0pt, opacity=0.2]
            (v6'.120) -- (v1'.120) arc(120:300:0.2cm) -- (v6'.300) arc(-60:120:0.2cm) ;
       \foreach \l in {1,...,7}{
      \filldraw [black] (v\l) circle (2pt) node [inner sep=5pt, label=below:$v_{\l}$] {};
    }

    \end{tikzpicture}
}
\end{center}
\caption{\vspace{-4mm}How to add a fourth edge to extend case (b) via a new vertex.}
\label{fig:f1}
\end{figure}
\item  If we do not add a new vertex:
		\begin{enumerate}
            \item Case (a) can only be extended while maintaining the $1$-intersecting property via a graph isomorphic to that in Figure~\ref{fig:f1}.
			\item For case (b), either we obtain a graph isomorphic to Figure~\ref{fig:f1}, or we add $e_4=\{v_2,v_4,v_6\}$ to obtain Figure~\ref{fig:f2}.\\
\begin{figure}[h]			
\begin{center}
\resizebox{90pt}{!}{%

	    		\begin{tikzpicture}
        [inner sep=0pt]
        \coordinate (v1) at (210:2cm);
        \coordinate (v2) at (270:1cm);
        \coordinate (v3) at (330:2cm);
        \coordinate (v4) at (150:1cm);
        \coordinate (v5) at (90:2cm);
        \coordinate (v6) at (30:1cm);
        \foreach \i in {1,...,6}{
            \node [circle, font=\tiny, minimum size=0.4cm, line width=0mm] (v\i') at (v\i) {};
        }

        \filldraw [draw=black, fill=yellow, minimum size=0pt, opacity=0.2]
            (v5'.150) -- (v1'.150) arc (-210:-30:0.2cm) -- (v5'.330) arc (-30:150:0.2cm);

        \filldraw [draw=black, fill=blue, minimum size=0pt, opacity=0.2]
            (v1'.270) -- (v3'.270) arc(-90:90:0.2cm) -- (v1'.90) arc (90:270:0.2cm);

        \filldraw [draw=black, fill=red, minimum size=0pt, opacity=0.2]
            (v3'.30) -- (v5'.30) arc(30:210:0.2cm) -- (v3'.210) arc(-150:30:0.2cm);

        \filldraw [draw=black, fill=green, minimum size=0pt, opacity=0.2, even odd rule]
            circle (1.2cm) (0:0.8) arc (0:360:0.8cm);
       \foreach \l in {1,...,6}{
      \filldraw [black] (v\l) circle (2pt) node [inner sep=5pt, label=below:$v_{\l}$] {};
    }
    \end{tikzpicture}
}
\end{center}
\caption{\vspace{-4mm}One of two ways to add a $4$th edge to extend case (b) without adding a new vertex.}
\label{fig:f2}
\end{figure}
		\end{enumerate}
	\end{itemize}
From the fifth edge onward, an analogous case analysis yields at each step that only two types of edges can be added at this point which preserve the $1$-intersecting property: A ``crossing edge'' (as in Figure~\ref{fig:f1}) $\set{v_1,v_6,v_7}$, $\set{v_3,v_4,v_7}$, or $\set{v_2,v_5,v_7}$, or a ``circular edge'' (as in Figure~\ref{fig:f2}) $\set{v_2,v_4,v_6}$. When all such edges are added, we arrive at $7$ edges total and obtain the Fano hypergraph. No fewer edges satisfies the $m=n$ property, and clearly one cannot add another edge at this point while satisfying the $1$-intersecting property. Finally, that $G$ has an SDR can either be verified directly or by applying Corollary~\ref{cor:jukna}.
\end{proof}

\subsection{Linear hypergraphs}\label{sscn:linear}
In Section~\ref{ssscn:twointersect}, we studied the set of $3$-uniform hypergraphs in which if a pair of distinct edges intersect, then the intersection size is at least $2$. We now take the complementary approach by asking that any pair of distinct edges intersect on \emph{at most one} vertex. Such hypergraphs are called \emph{linear}; note that linear hypergraphs are not necessarily $1$-intersecting.

To set the context, linear hypergraphs in general (i.e. not necessarily with an SDR) are a complicated, but well-studied, class. Let $L_3^l$ denote the set of edge intersection graphs of linear $3$-uniform hypergraphs. (An \emph{edge intersection graph} (EIG) is the generalization of a line graph to the setting of hypergraphs $G$; namely, the vertices of the EIG are the hyperedges of $G$, and two vertices of the EIG are neighbors if and only if their hyperedges in G intersect. Note that an EIG is a graph, i.e. a $2$-uniform hypergraph.) It is known that $L_3^l$ has no ``finite characterization'' in terms of a finite list of forbidden induced subgraphs~\cite{NAIK1982159}. However, the same paper does give the following characterization: A graph $G$ is in $L_k^l$ if and only if $G$ has so-called \emph{Krausz dimension} $k$. Unfortunately, it was later shown~\cite{Hlinen1997ComputationalCO} that deciding if Krausz dimension is at most $3$ is NP-complete (even on planar graphs of degree at most $5$); thus, determining if a graph is an EIG of a linear $3$-uniform hypergraph is NP-complete, suggesting the class of linear $3$-uniform hypergraphs is quite complex.

In our setting, we study linear $3$-uniform hypergraphs $G$ with the additional guarantees that $m=n$ and that $G$ has an SDR. Nevertheless, we are not able to complete a structural characterization. Instead, we find various interesting ``canonical-looking'' examples. For brevity, we henceforth denote a linear $3$-uniform hypergraph with $m=n$ and an SDR as an LGraph.

\paragraph*{Minimum size and the Fano plane.} What is the minimum size an LGraph can have?

\begin{theorem}
    The minimum size of an LGraph is $m=n=7$, and the Fano hypergraph (Figure~\ref{fig:fano}) is the unique LGraph of this size.
\end{theorem}
\begin{proof}
    The lower bound on minimum size follows from an old lemma of Corr\'{a}di.
\begin{lemma}[Corr\'adi 1969 ~\cite{SJ011}]
 Let $A_1,...,A_N$ be $r-$element sets and $X$ be their union. If $|A_i\cap A_j|\leq k$ for all $i\neq j$, then
$
\abs{X}\geq (r^2N)/(r+(N-1)k).
$
\end{lemma}
\noindent Plugging in our LGraph parameters $k=1$, $r=3$, $E=\{A_1,...,A_m\}$, ${V}=X=\{1,...,n\}$, yields that $\abs{E}=\abs{V}\geq 7$. Note that the Fano hypergraph is an LGraph with $m=n=7$, which also happens to be $1$-intersecting. To show uniqueness, we hence show that \emph{any} LGraph with $m=n=7$ must be $1$-intersecting; the claim then follows from Lemma~\ref{l:1intersect}.

      To show that any LGraph with $m=n=7$ must be $1$-intersecting, suppose, for sake of contradiction, that $G=(V,E)$ is an LGraph with $m=n=7$ and edges $e_1$,$e_2$ with $e_1\cap e_2=\emptyset$. Then, the subgraph of $G$ induced by $\set{e_1,e_2}$ is isomorphic to\\
    	\begin{center}
\resizebox{300pt}{!}{%
\begin{tikzpicture}
    \coordinate (v1) at (0,0);
    \coordinate (v2) at (2,0) {};
    \coordinate (v3) at (4,0);
    \coordinate (v4) at (6,0) {};
    \coordinate (v5) at (8,0) {};
    \coordinate (v6) at (10,0) {};
    \coordinate (v7) at (12,0) {};
    \foreach \v in {v1,v2,v3,v4,v5,v6,v7}{
      \node [circle, minimum width=0.2cm] (\v') at (\v) {};
    }
    \filldraw [draw=black, fill=red, opacity=0.2]
		(v1'.270) -- (v3'.270) arc (-90:90:0.2)
		-- (v2'.90) arc (90:90:0.2)
		-- (v1'.90) arc (90:270:0.2) -- cycle;
		
	\filldraw [draw=black, fill=green, opacity=0.2]
		(v5'.90) -- (v7'.90) arc (90:-90:0.2)
		-- (v6'.-90) arc (-90:-90:0.2)
		-- (v5'.-90) arc (-90:-270:0.2) -- cycle;
		
    \foreach \l in {1,...,7}{
      \filldraw [black] (v\l) circle (2pt) node [label=below:$v_\l$] {};
    }
\end{tikzpicture}
}
\end{center}
Since $G$ is linear, any edge we add must contain $v_4$, and in turn, must also contain a \emph{unique} vertex from each of $\set{v_1,v_2,v_3}$ and $\set{v_5,v_6,v_7}$. Thus, the maximum number of edges we may add is $3$, giving a total of $5$ edges, which is a contradiction.
\end{proof}

\paragraph*{The tiling of the torus.} There is a sense in which the Fano hypergraph is the ``most compact'' LGraph. On the opposite extreme, an example of what seems the ``least compact'' LGraph is the ``tiling of the torus'', given in Figure~\ref{fig:torus}. More precisely, this is the specialization of Example \ref{ex:torus} to the case $k=3$ and $a_1=a_2=3$.

\paragraph*{The Helly property.} Another well-studied hypergraph property is the \emph{Helly} property, which requires that any intersecting family $F$ of hyperedges of the hypergraph $G$ have non-empty intersection. An equivalent, perhaps more geometrically intuitive, characterization is the following:
\begin{corollary}[Page 23 of \cite{Be85}, stated as Corollary 5.1 of~\cite{BRETTO2001177}]
    A hypergraph $H$ has the Helly property if and only if for any three vertices $a_1,a_2,a_3$, the family of hyperedges which contains at least two of these vertices has non-empty intersection.
\end{corollary}
\noindent In the special case of linear hypergraphs, the Helly property forbids a ``linear 3-cycle'', as illustrated in Figure~\ref{fig:nohelly}. In this sense, linear Helly hypergraphs roughly generalize triangle-free graphs.

Observe now that the two ``canonical-looking'' LGraphs we have discussed thus far, the Fano plane and the tiling of the torus, do \emph{not} satisfy the Helly property. This raises the question: Do all LGraphs violate the Helly property? In Lemma~\ref{l:helly}, we answer this question negatively by giving an LGraph with the Helly property. Thus, even with numerous restrictions (i.e. $m=n$, linear, Helly), the set of $3$-uniform hypergraphs with SDRs seems non-trivial.
\begin{lemma}\label{l:helly}
    There exists a linear $3$-uniform hypergraph $G=(V,E)$ with $\abs{E}=\abs{V}=28$ such that $G$ satisfies the Helly property, and $G$ has an SDR.
\end{lemma}
\begin{proof}
    Since the Helly property forbids a linear $3$-cycle, the idea of the construction is to build a longer linear cycle, with added ``links'' between vertices to ensure $m=n$. We call this the ``interlinked cycle'' (iCycle) for short; a formal specification is given in Figure~\ref{fig:helly}, and a graphical depiction in Figure~\ref{fig:icycle}. The claimed properties are tedious, but straightforward, to verify.
    \begin{figure}[t]
        \[
        \begin{tabular}{|c|c|c|c|c|c|c|c|c|c|}
          \hline
          e & v & e & v & e & v & e & v & e& v
           \\
          \hline
          1 & (1,2,\underline{28}) & 7 & (12,\underline{13},14) & 13 & (24,25,\underline{26}) & 19 & (8,14,\underline{19}) & 25 & (\underline{3},15,27) \\
          2 & (\underline{2},3,4) & 8 & (\underline{14},15,16) & 14 & (26,\underline{27},28) & 20 & (6,12,\underline{21}) & 26 & (1,\underline{6},24)  \\
          3 & (\underline{4},5,6) & 9 & (\underline{16},17,18) & 15 & (\underline{5},22,28) & 21 & (4,10,\underline{23}) & 27 & (\underline{1},10,20)  \\
          4 & (6,7,\underline{8}) & 10& (\underline{18},19,20) & 16 & (\underline{7},20,26) & 22 & (2,8,\underline{25}) & 28 & (1,13,\underline{17})    \\
          5 & (8,9,\underline{10}) & 11 & (20,21,\underline{22}) & 17 & (\underline{9},18,24) & 23 & (10,15,\underline{20}) &  &  \\
          6 & (10,11,\underline{12}) & 12 & (22,23,\underline{24}) & 18 & (\underline{11},16,22) & 24 & (6,\underline{15},24) & &\\
          \hline
        \end{tabular}
        \]
    \caption{A formal specification of the iCycle. The columns labelled \emph{e} and \emph{v} index the edges and $3$-tuples corresponding to those edges, respectively. The underlined element of each $3$-tuple denotes the qubit matched to that edge.}
    \label{fig:helly}
    \end{figure}
    \begin{figure}[t]
	\begin{center}
\resizebox{300pt}{!}{%
\begin{tikzpicture}
[inner sep=0pt, minimum size = 0pt]
	\coordinate (v0) at (0,0) {};
	\coordinate (v1) at (-8,0) {};
	\coordinate (v2) at (-6,1) {};
	\coordinate (v3) at (-5,3) {};
	\coordinate (v4) at (-4,1) {};
	\coordinate (v5) at (-3,3) {};
	\coordinate (v6) at (-2,1) {};
	\coordinate (v7) at (-1,3) {};
	\coordinate (v8) at (0,1) {};
	\coordinate (v9) at (1,3) {};
	\coordinate (v10) at (2,1) {};
	\coordinate (v11) at (3,3) {};
	\coordinate (v12) at (4,1) {};
	\coordinate (v13) at (5,3) {};
	\coordinate (v14) at (6,1) {};
	\coordinate (v15) at (8,0) {};
	\coordinate (v16) at (6,-1) {};
	\coordinate (v17) at (5,-3) {};
	\coordinate (v18) at (4,-1) {};
	\coordinate (v19) at (3,-3) {};
	\coordinate (v20) at (2,-1) {};
	\coordinate (v21) at (1,-3) {};
	\coordinate (v22) at (0,-1) {};
	\coordinate (v23) at (-1,-3) {};
	\coordinate (v24) at (-2,-1) {};
	\coordinate (v25) at (-3,-3) {};
	\coordinate (v26) at (-4,-1) {};
	\coordinate (v27) at (-5,-3) {};
	\coordinate (v28) at (-6,-1) {};
	
	\foreach \v in {v1,v2,v3,v4,v5,v6,v7,v8,v9,v10,v11,v12,v13,v14,v15,v16,v17,v18,v19,v20,v21,v22,v23,v24,v25,v26,v27,v28}{
      \node [circle, minimum size=0.4cm, line width=0pt] (\v') at (\v) {};
    }

	\filldraw [draw=black, fill=yellow, opacity=0.2]
		(v4'.30) -- (v3'.30) arc (30:150:0.2cm)
		-- (v2'.150) arc (150:270:0.2)
		-- (v4'.270) arc (-90:30:0.2) -- cycle;
		
	\filldraw [draw=black, fill=yellow, opacity=0.2]
		(v6'.30) -- (v5'.30) arc (30:150:0.2)
		-- (v4'.150) arc (150:270:0.2)
		-- (v6'.270) arc (-90:30:0.2) -- cycle;
	\filldraw [draw=black, fill=yellow, opacity=0.2]
		(v8'.30) -- (v7'.30) arc (30:150:0.2)
		-- (v6'.150) arc (150:270:0.2)
		-- (v8'.270) arc (-90:30:0.2) -- cycle;
	\filldraw [draw=black, fill=yellow, opacity=0.2]
		(v10'.30) -- (v9'.30) arc (30:150:0.2)
		-- (v8'.150) arc (150:270:0.2)
		-- (v10'.270) arc (-90:30:0.2) -- cycle;
	\filldraw [draw=black, fill=yellow, opacity=0.2]
		(v12'.30) -- (v11'.30) arc (30:150:0.2)
		-- (v10'.150) arc (150:270:0.2)
		-- (v12'.270) arc (-90:30:0.2) -- cycle;
	\filldraw [draw=black, fill=yellow, opacity=0.2]
		(v14'.30) -- (v13'.30) arc (30:150:0.2)
		-- (v12'.150) arc (150:270:0.2)
		-- (v14'.270) arc (-90:30:0.2) -- cycle;
		
	\filldraw [draw=black, fill=yellow, opacity=0.2]
		(v16'.-30) -- (v17'.-30) arc (-30:-150:0.2)
		-- (v18'.210) arc (210:90:0.2)
		-- (v16'.90) arc (90:-30:0.2) -- cycle;
	\filldraw [draw=black, fill=yellow, opacity=0.2]
		(v18'.-30) -- (v19'.-30) arc (-30:-150:0.2)
		-- (v20'.210) arc (210:90:0.2)
		-- (v18'.90) arc (90:-30:0.2) -- cycle;
	\filldraw [draw=black, fill=yellow, opacity=0.2]
		(v20'.-30) -- (v21'.-30) arc (-30:-150:0.2)
		-- (v22'.210) arc (210:90:0.2)
		-- (v20'.90) arc (90:-30:0.2) -- cycle;
	\filldraw [draw=black, fill=yellow, opacity=0.2]
		(v22'.-30) -- (v23'.-30) arc (-30:-150:0.2)
		-- (v24'.210) arc (210:90:0.2)
		-- (v22'.90) arc (90:-30:0.2) -- cycle;
	\filldraw [draw=black, fill=yellow, opacity=0.2]
		(v24'.-30) -- (v25'.-30) arc (-30:-150:0.2)
		-- (v26'.210) arc (210:90:0.2)
		-- (v24'.90) arc (90:-30:0.2) -- cycle;
	\filldraw [draw=black, fill=yellow, opacity=0.2]
		(v26'.-30) -- (v27'.-30) arc (-30:-150:0.2)
		-- (v28'.210) arc (210:90:0.2)
		-- (v26'.90) arc (90:-30:0.2) -- cycle;
		
	\filldraw [draw=black, fill=yellow, opacity=0.2]
		(v1'.120) -- (v2'.120) arc (120:0:0.2)-- (v28'.0) arc (0:-120:0.2)--(v1'.-120) arc (-120:-240:0.2) -- cycle;
	\filldraw [draw=black, fill=yellow, opacity=0.2]
		(v15'.60) -- (v14'.60) arc (60:180:0.2)-- (v16'.180) arc (180:300:0.2)--(v15'.-60) arc (-60:30:0.2) -- cycle;
%
	\filldraw [draw=black, fill=red, opacity=0.2, line width=0.5pt]
		plot[smooth, tension=1] coordinates { (v28'.120) (-4,0.5) (v5'.180)} arc(180:0:0.2) --
		plot[smooth, tension=1] coordinates { (v5'.0) (-2,0.5) (v22'.60)} arc(60:-120:0.2)
		plot[smooth, tension=1] coordinates { (v22'.-120) (-3,0) (v28'.-60)}
			arc(-60:-240:0.2);
	\filldraw [draw=black, fill=red, opacity=0.2, line width=0.5pt]
		plot[smooth, tension=1] coordinates { (v26'.120) (-2,0.5) (v7'.180)} arc(180:0:0.2) --
		plot[smooth, tension=1] coordinates { (v7'.0) (0,0.5) (v20'.60)} arc(60:-120:0.2)
		plot[smooth, tension=1] coordinates { (v20'.-120) (-1,0) (v26'.-60)}
			arc(-60:-240:0.2);
	\filldraw [draw=black, fill=red, opacity=0.2, line width=0.5pt]
		plot[smooth, tension=1] coordinates { (v24'.120) (0,0.5) (v9'.180)} arc(180:0:0.2) --
		plot[smooth, tension=1] coordinates { (v9'.0) (2,0.5) (v18'.60)} arc(60:-120:0.2)
		plot[smooth, tension=1] coordinates { (v18'.-120) (1,0) (v24'.-60)}
			arc(-60:-240:0.2);
	\filldraw [draw=black, fill=red, opacity=0.2, line width=0.5pt]
		plot[smooth, tension=1] coordinates { (v22'.120) (2,0.5) (v11'.180)} arc(180:0:0.2) --
		plot[smooth, tension=1] coordinates { (v11'.0) (4,0.5) (v16'.60)} arc(60:-120:0.2)
		plot[smooth, tension=1] coordinates { (v16'.-120) (3,0) (v22'.-60)}
			arc(-60:-240:0.2);
%
	\filldraw [draw=black, fill=green, opacity=0.2, line width=0.5pt]
		plot[smooth, tension=1] coordinates { (v8'.-90) (2,0) (v19'.-180)} arc(-180:0:0.2) --
		plot[smooth, tension=1] coordinates { (v19'.0) (4,0) (v14'.-90)} arc(-90:120:0.2)
		plot[smooth, tension=1] coordinates { (v14'.120) (3,0.5) (v8'.60)}
		arc(60:270:0.2);
	\filldraw [draw=black, fill=green, opacity=0.2, line width=0.5pt]
		plot[smooth, tension=1] coordinates { (v6'.-90) (0,0) (v21'.-180)} arc(-180:0:0.2) --
		plot[smooth, tension=1] coordinates { (v21'.0) (2,0) (v12'.-90)} arc(-90:120:0.2)
		plot[smooth, tension=1] coordinates { (v12'.120) (1,0.5) (v6'.60)}
		arc(60:270:0.2);
	\filldraw [draw=black, fill=green, opacity=0.2, line width=0.5pt]
		plot[smooth, tension=1] coordinates { (v4'.-90) (-2,0) (v23'.-180)} arc(-180:0:0.2) --
		plot[smooth, tension=1] coordinates { (v23'.0) (0,0) (v10'.-90)} arc(-90:120:0.2)
		plot[smooth, tension=1] coordinates { (v10'.120) (-1,0.5) (v4'.60)}arc(60:270:0.2);
	\filldraw [draw=black, fill=green, opacity=0.2, line width=0.5pt]
		plot[smooth, tension=1] coordinates { (v2'.-90) (-4,0) (v25'.-180)} arc(-180:0:0.2) --
		plot[smooth, tension=1] coordinates { (v25'.0) (-2,0) (v8'.-90)} arc(-90:120:0.2)
		plot[smooth, tension=1] coordinates { (v8'.120) (-3,0.5) (v2'.60)}
		arc(60:270:0.2);
	\filldraw [draw=black, fill=blue, opacity=0.2, line width=0.5pt]
		plot[smooth, tension=1] coordinates { (v10'.150) (3,4.5) (v15'.0) (3,-4.5)  (v20'.210)}
		arc(210:30:0.2) --
		plot[smooth, tension=1] coordinates {(v20'.30) (3,-4) (v15'.180) (3,4) (v10'.-30) }
		arc(-30:-210:0.2)
		;
	\filldraw [draw=black, fill=cyan, opacity=0.2, line width=0.5pt]
		plot[smooth, tension=1] coordinates { (v6'.150) (1.5,6) (v15'.0) (1,-5.5)  (v24'.210)}
		arc(210:30:0.2) --
		plot[smooth, tension=1] coordinates {(v24'.30) (1,-5) (v15'.180) (1.5,5.5) (v6'.-30) }
		arc(-30:-210:0.2)
		;	
	\filldraw [draw=black, fill=red, opacity=0.2, line width=0.5pt]
		plot[smooth, tension=1] coordinates { (v3'.150) (-2,8.5) (v15'.0) (-2,-8)  (v27'.210)}
		arc(210:30:0.2) --
		plot[smooth, tension=1] coordinates {(v27'.30) (-2,-7.5) (v15'.180) (-2,8) (v3'.-30) }
		arc(-30:-210:0.2)
		;	
	\filldraw [draw=black, fill=blue, opacity=0.2, line width=0.5pt]
		plot[smooth, tension=1] coordinates { (v6'.150) (-3,4) (v1'.0) (-3,-4)  (v24'.210)}
		arc(210:30:0.2) --
		plot[smooth, tension=1] coordinates {(v24'.30) (-3,-4.5) (v1'.180) (-3,4.5) (v6'.-30) }
		arc(-30:-210:0.2)
		;
	\filldraw [draw=black, fill=cyan, opacity=0.2, line width=0.5pt]
		plot[smooth, tension=1] coordinates { (v10'.150) (-1.5,5.5) (v1'.0) (-1.5,-5)  (v20'.210)}
		arc(210:30:0.2) --
		plot[smooth, tension=1] coordinates {(v20'.30) (-1.5,-5.5) (v1'.180) (-1.5,6) (v10'.-30) }
		arc(-30:-210:0.2)
		;
	\filldraw [draw=black, fill=red, opacity=0.2, line width=0.5pt]
		plot[smooth, tension=1] coordinates { (v13'.150) (2,8) (v1'.0) (2,-7.5)  (v17'.210)}
		arc(210:30:0.2) --
		plot[smooth, tension=1] coordinates {(v17'.30) (2,-8) (v1'.180) (2,8.5) (v13'.-30) }
		arc(-30:-210:0.2)
		;

    \foreach \l in {1,...,28}{
      \filldraw [black] (v\l) circle (2pt) node [inner sep=5pt, label=below:$v_{\l}$] {};
    }

\end{tikzpicture}
}
\end{center}
\caption{The iCycle.}
\label{fig:icycle}
\end{figure}
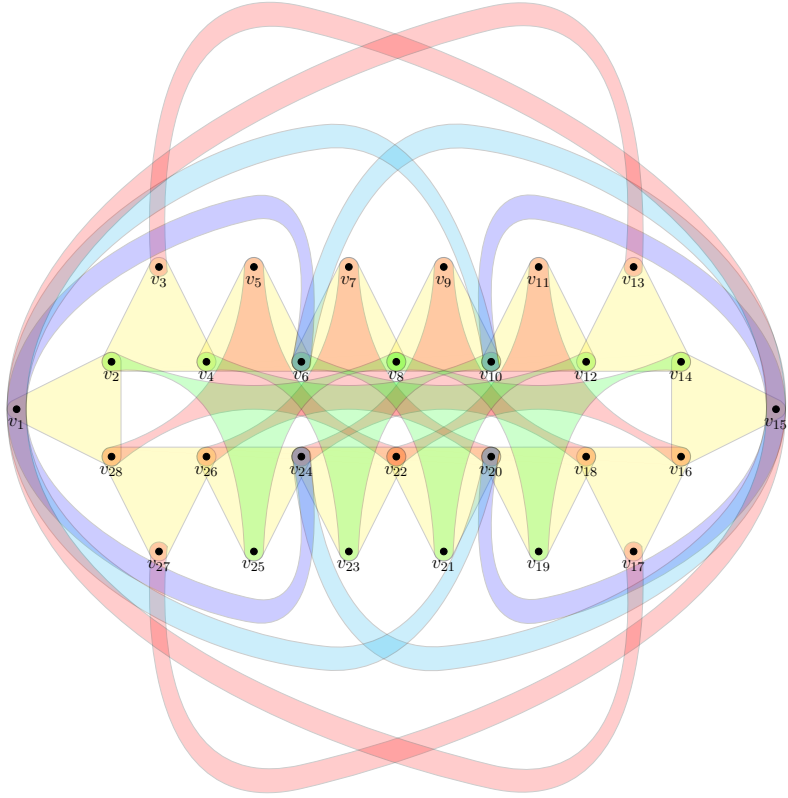
\end{proof}

\paragraph*{Do all linear $3$-uniform hypergraphs with $m=n$ have an SDR?}
The examples we have studied in this section are all LGraphs, even with constraints such as the Helly property. This raises the question: Could \emph{all} linear $3$-uniform hypergraphs with $m=n$ have SDR's? Intuitively, this is plausible, as the linear property forces a $3$-uniform hypergraph to be ``spread out'' in a rough sense. We now answer this question in the negative, and the construction of the counterexample will teach us something about the structure of LGraphs.

For this, we first require a definition of a \emph{block} and useful lemma.

\begin{definition}[Block]
    For hypergraph $G=(V,E)$, let $G_{E'}$ denote the edge-induced subgraph of $G$ by $E'\subseteq E$. Then, we call $G_{E'}$ a \emph{block} if it has an equal number of vertices and edges.
\end{definition}
\begin{lemma}\label{l:block}
    Let $G=(V,E)$ be a hypergraph with an SDR. Let $A$ denote a block in $G$. Then, the SDR matches the vertices in $A$ to the edges in $A$. Moreover, for any other block $B$ in $G$, $A$ and $B$ are vertex-disjoint.
\end{lemma}
\begin{proof}
    Let $A=(V_A,E_A)$ be a block in $G$. Since $A$ is edge-induced, the only vertices which can match to the edges in $E_A$ are those in $V_A$. Since $\abs{V_A}=\abs{E_A}$ by definition of a block, our first claim follows. The second claim now follows as a corollary, since if $A$ and $B$ intersect on vertex $v\in V$, then \emph{both} $A$ and $B$ require $v$ as part of their respective SDR's by our first claim, which is impossible since each vertex can only be matched to a single edge.
\end{proof}

\begin{theorem}\label{thm:linearcounter}
    There exists a linear hypergraph with $m=n$ which does not have an SDR.
\end{theorem}
\begin{proof}
    The intuitive idea is to ``stitch together'' three copies of (say) the Fano hypergraph with an extra edge, so that there is no way to match the new edge to a vertex. This loses the $m=n$ property; we then add further edges and vertices to recover the property. The resulting hypergraph is depicted in Figure~\ref{fig:counter}. That the graph lacks an SDR follows from Lemma~\ref{l:block}; specifically, let $e$ denote the unique edge incident on all three Fano hypergraphs. Since the latter are blocks, none of their vertices can be matched to $e$, implying $e$ cannot be matched.
\end{proof}
\begin{figure}[t]
\begin{center}
\resizebox{220pt}{!}{%
\begin{tikzpicture}
	\coordinate (v0) at (0,0) {};
    \coordinate (v1) at (0,3) {};
    \coordinate (v2) at (0,2) {};
    \coordinate (v3) at (0,1){};
    \coordinate (v4) at (1,0) {};
    \coordinate (v5) at (2,0) {};
    \coordinate (v6) at (3,0) {};
    \coordinate (v7) at (-1,0) {};
    \coordinate (v8) at (-2,0) {};
    \coordinate (v9) at (-3,0) {};
    \coordinate (v10) at (-3,6) {};
    \coordinate (v11) at (3,6) {};
    \coordinate (v12) at (6,3) {};
    \coordinate (v13) at (6,-3) {};
    \coordinate (v14) at (-6,-3) {};
    \coordinate (v15) at (-6,3) {};

    \foreach \v in {v1,v2,v3,v4,v5,v6,v7,v8,v9,v10,v11,v12,v13,v14,v15}{
      \node [circle, minimum size=0.4cm, line width=0pt] (\v') at (\v) {};
    }

   \filldraw [draw=black, fill=blue, minimum size=0pt, opacity=0.2, even odd rule]
         circle (1.2cm) (0:0.8) arc (0:360:0.8cm);

   \filldraw [draw=black, fill=blue, minimum size=0pt, opacity=0.2, even odd rule]
         circle (2.2cm) (0:1.8) arc (0:360:1.8cm);
   \filldraw [draw=black, fill=blue, minimum size=0pt, opacity=0.2, even odd rule]
         circle (3.2cm) (0:2.8) arc (0:360:2.8cm);

   \filldraw [draw=black, fill=yellow, opacity=0.2]
		(v11'.-50) -- (v1'.-45) arc (-45:-140:0.2)
		-- (v10'.230) arc (230:90:0.2)
		-- (v11'.90) arc (90:-50:0.2) -- cycle;
    	
	\filldraw [draw=black, fill=yellow, opacity=0.2]
		(v12'.150) -- (v6'.150) arc (150:240:0.2)
		-- (v13'.210) arc (-150:0:0.2)
		-- (v12'.0) arc (0:150:0.2) -- cycle;
	\filldraw [draw=black, fill=yellow, opacity=0.2]
		(v15'.45) -- (v9'.45) arc (45:-45:0.2)
		-- (v14'.-45) arc (-45:-180:0.2)
		-- (v15'.180) arc (180:45:0.2) -- cycle;
   \filldraw [draw=black, fill=red, minimum size=0pt, opacity=0.2]
            (v4'.270) -- (v6'.270) arc(-90:90:0.2cm) -- (v4'.90) arc (90:270:0.2cm);
  \filldraw [draw=black, fill=red, minimum size=0pt, opacity=0.2]
            (v9'.270) -- (v7'.270) arc(-90:90:0.2cm) -- (v9'.90) arc (90:270:0.2cm);
   \filldraw [draw=black, fill=red, minimum size=0pt, opacity=0.2]
            (v1'.0) -- (v3'.0) arc(0:-180:0.2cm) -- (v1'.180) arc (180:0:0.2cm);

    \foreach \l in {1,...,9}{
      \filldraw [black] (v\l) circle (2pt) node [inner sep=5pt, label=below:$v_{\l}$] {};
    }
\end{tikzpicture}
}
\end{center}
\caption{A linear hypergraph with $m=n$ and no SDR. The three large triangles denote any block with an SDR, such as the Fano hypergraph.}
\label{fig:counter}
\end{figure}
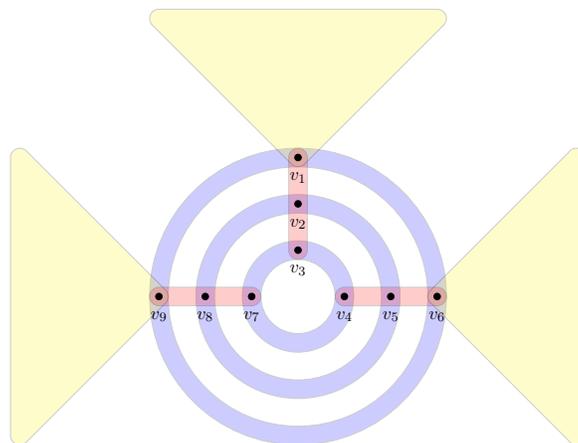
\noindent As an aside, note that the construction of Theorem~\ref{thm:linearcounter} can be modified to also satisfy the Helly property --- namely, replace each copy of the Fano hypergraph with a copy of the iCycle, discard the inner two circular edges $\set{v_2,v_8,v_5}$ and $\set{v_3,v_4,v_7}$ and the two ``straight-line'' edges $\set{v_9,v_8,v_7}$ and $\set{v_4,v_5,v_6}$, and discard vertices $v_4,v_5,v_7,v_8$. The resulting hypergraph is linear, Helly, and satisfies $m=n$, but does not have a matching. Thus, the construction of Theorem~\ref{thm:linearcounter} appears to yield a fairly systematic approach for constructing hypergraphs without SDR's, but still satisfying other desirable properties such as being linear or Helly.

\section*{Acknowledgements}
The first result of this project was partially completed while NdB was affiliated with the Centrum Wiskunde \& Informatica; SG thanks Ronald de Wolf and Centrum Wiskunde \& Informatica for their hospitality. SG thanks Howard Barnum and David Reeb regarding discussions on algebraic geometry, and David Gosset for discussions on Quantum SAT. NdB acknowledges support from the EPSRC National Quantum Technology
Hub in Networked Quantum Information Processing. SG acknowledges support from NSF grants CCF-1526189 and CCF-1617710.

\bibliography{3QSAT}

\end{document}